%% file: ms.tex
\documentclass[awide]{article}
\usepackage[margin=3cm]{geometry}

\input{osdi-pkgs.tex}

\input{techreport-pkgs} 
\input{paper-cmds.tex}
\input{techreport-cmds} 

\title{\sys: Push-Button Verification and Optimization for Synchronization Primitives on Weak Memory Models\\
(Technical Report)}

\input{authors}

\begin{document}
\date{}
\maketitle

\begin{abstract}
	This technical report contains material accompanying our work with same title published at ASPLOS'21~\cite{vsync}.
	We start in \S\ref{s:spinning} with a detailed presentation of the core innovation of this work, Await Model Checking (AMC).
	The correctness proofs of AMC can be found in \S\ref{s:amc}.
	Next, we discuss three study cases in \S\ref{s:cases},
	presenting bugs found and challenges encountered when applying \sys to {existing code bases}.
	Finally, in \S\ref{s:eval} we describe the setup details of our evaluation and report further experimental results.
\end{abstract}
\pagebreak
\tableofcontents
\pagebreak
\input{wt}

\input{amc}
\pagebreak
\input{cases}
\pagebreak
\input{eval}

\pagebreak
\bibliographystyle{plain}
\bibliography{references}

\end{document}

%% file: osdi-pkgs.tex
\usepackage{tikz}
\usepackage{amssymb}
\usepackage{amsmath}
\usepackage{amsthm}

\usepackage[normalem]{ulem}
\usepackage{endnotes,microtype,xspace,graphicx,fancyvrb,multirow}
\usepackage{xcolor}
\usepackage{xspace}
\usepackage{balance}
\usepackage{pifont}
\usepackage{braket}
\usepackage{anyfontsize}
\usepackage{booktabs}
\usepackage{fp}
\usepackage{array}
\usepackage{balance}
\usepackage{textcomp}
\usepackage{stmaryrd}
\usepackage{enumitem}
\usepackage{subcaption}
\usepackage{float}
\usepackage{stfloats}
\usepackage{underscore}
\usepackage{multirow}
\usepackage{wrapfig}
\usepackage{caption}
\usepackage{textgreek}

\usepackage{centernot}
\usepackage[vlined,ruled]{algorithm2e}
\usepackage{etoolbox}

\makeatletter
\patchcmd{\@algocf@start}
  {-1.5em}
  {0pt}
  {}{}
\makeatother
\usepackage{latexsym}

\usepackage{tikz}
\usetikzlibrary{positioning, calc, decorations.markings,arrows,decorations.pathmorphing,fit,spy,shapes,intersections}

\tikzset{
    hyperlink node/.style={
        alias=sourcenode,
        append after command={
            let \p1=(sourcenode.north west),
                \p2=(sourcenode.south east),
                \n1={\x2-\x1},
                \n2={\y1-\y2} in
            node [inner sep=0pt, outer sep=1pt,anchor=north west, at=(\p1)]
	    {\hyperlink{#1}{\XeTeXLinkBox{\phantom{\rule{\n1}{\n2}}}}}
        }
    }
}

\usepackage{listings}

\definecolor{darkred}{RGB}{184, 5, 69}
\definecolor{darkgreen}{rgb}{0.0, 0.2, 0.13}
\definecolor{gray}{gray}{0.25}
\definecolor{lightblue}{HTML}{8080FF}

\lstdefinelanguage
	[RISCV]{Assembler}
	[x86masm]{Assembler}
	{morekeywords={amoswap.w,amoswap.w.aq,amoswap.w.rl,amoswap.w.aq.rl.sc,fence,lw,sw,rw,r,w}}

\lstdefinelanguage
	[ARM]{Assembler}
	[x86masm]{Assembler}
	{morekeywords={SWP,SWPA,SWPL,SWPAL}}	

\lstdefinelanguage{pcal}{
	morekeywords={CHOOSE, CONSTANT, EXCEPT, UNION, SUBSET, LET, IN, CASE, OTHER, ASSUME, ASSUMPTION, AXIOM, THEOREM, assert, await, begin, call, define, do, either, else, elsif, end, goto, if, macro, or, print, procedure, process, return, skip, then, variable, variables, when, while, with},
	morecomment=[l]{\\*},
	morecomment=[s]{(*}{*)},
	morestring=[b]"
}
\lstdefinestyle{compact}{
	basicstyle=\scriptsize\ttfamily,
	keywordstyle=\scriptsize\color{blue}\ttfamily,
	commentstyle=\scriptsize\color{gray}\ttfamily,
	stringstyle=\scriptsize\color{blue}\ttfamily,
	numberstyle=\tiny\color{gray},
	columns=fullflexible,
	tabsize=2,
	showstringspaces=false,
	moredelim=**[is][\color{darkred}]{@}{@},
	moredelim=**[is][\color{barrier}]{$}{$},
}

\lstdefinestyle{mainexample}{
	basicstyle=\small\ttfamily,
	keywordstyle=\small\color{darkred}\ttfamily,
	commentstyle=\small\color{gray}\ttfamily,
	stringstyle=\small\color{blue}\ttfamily,
	numberstyle=\small,
	columns=fullflexible,
	tabsize=2,
	showstringspaces=false
}

\lstdefinestyle{mystyle}{
	basicstyle=\scriptsize\ttfamily,
	keywordstyle=\scriptsize\color{darkred}\ttfamily,
	commentstyle=\scriptsize\color{gray}\ttfamily,
	stringstyle=\scriptsize\color{blue}\ttfamily,
	numberstyle=\tiny,
	columns=fullflexible,
	numbers=left,
	tabsize=4,
	showstringspaces=false
}

\lstdefinestyle{hlstyle}{
	basicstyle=\scriptsize\color{gray!20}\ttfamily,
	keywordstyle=\scriptsize\color{gray!20}\ttfamily,
	commentstyle=\scriptsize\color{gray!20}\ttfamily,
	stringstyle=\scriptsize\color{gray!20}\ttfamily,
	numberstyle=\tiny,
	columns=fullflexible,
	numbers=left,
	tabsize=4,
	showstringspaces=false
}
\lstdefinestyle{tabularstyle}{
	basicstyle=\small\ttfamily,
	keywordstyle=\small\color{darkred}\ttfamily,
	commentstyle=\small\color{gray}\ttfamily,
	stringstyle=\small\color{blue}\ttfamily,
	columns=fullflexible,
	frame=none,
	numbers=none,
	tabsize=2,
	showstringspaces=false
}

\definecolor{barrier}{rgb}{1,0.5,0}
\lstdefinestyle{tutorialcode}{
	language=C,
	emptylines=1,
	breaklines=true,
	basicstyle=\ttfamily\color{black},
	moredelim=**[is][\color{darkred}]{@}{@},
	moredelim=**[is][\color{barrier}]{$}{$},
	numbers=left,              
	stepnumber=1,              
	numbersep=5pt,              
	numberstyle=\tiny\color{gray},   
	keywordstyle=\color{blue},   
	numberfirstline=false,
	tabsize=4,
	escapechar=|,
	keepspaces,
}

%% file: techreport-pkgs.tex
\usepackage[T1]{fontenc}
\usepackage[utf8]{inputenc}
\usepackage{pslatex}

\usepackage{authblk}
\usepackage{cite}               
\usepackage{url}                
\usepackage{xcolor}             
\usepackage{hyperref}           
\hypersetup{                    
	linktocpage=true,
  colorlinks,
  linkcolor={green!60!black},
  citecolor={red!70!black},
  urlcolor={blue!70!black}
}
\usepackage{breakurl}           
\usepackage[capitalise]{cleveref}
\usepackage{tikz-qtree}
\usepackage{multicol}
\usepackage{todonotes} 

\usepackage{etoolbox}

\definecolor{barrier}{rgb}{1,0.5,0}
\lstdefinestyle{tutorialcode}{
	language=C,
	emptylines=1,
	breaklines=true,
	basicstyle=\ttfamily\color{black},
	moredelim=**[is][\color{darkred}]{@}{@},
	moredelim=**[is][\color{barrier}]{$}{$},
	numbers=left,              
	stepnumber=1,              
	numbersep=5pt,              
	numberstyle=\tiny\color{gray},   
	keywordstyle=\color{blue},   
	numberfirstline=false,
	tabsize=4,
	escapechar=|,
	keepspaces,
}

\lstdefinestyle{verbcode}{
	basicstyle=\scriptsize\ttfamily,
	keywordstyle=\scriptsize\ttfamily,
	commentstyle=\scriptsize\ttfamily,
	stringstyle=\scriptsize\ttfamily,
	numberstyle=\tiny\color{gray},
	emphstyle=\scriptsize\ttfamily,
	columns=fullflexible,
	tabsize=8,
	showstringspaces=false,
	moredelim=**[is][\bf\ttfamily]{@}{@},
	moredelim=**[is][\bf]{$}{$},
	escapechar=|,
	keepspaces,
}
\lstdefinestyle{verbcodelarge}{
	style=verbcode,
	basicstyle=\ttfamily,
	keywordstyle=\ttfamily,
	commentstyle=\ttfamily,
	stringstyle=\ttfamily,
	emphstyle=\ttfamily,
}
\lstdefinestyle{casecode}{
	language=C,
	basicstyle=\scriptsize\ttfamily,
	keywordstyle=\scriptsize\color{blue}\ttfamily,
	commentstyle=\scriptsize\color{gray}\ttfamily,
	stringstyle=\scriptsize\color{blue}\ttfamily,
	numberstyle=\tiny\color{gray},
	keywordstyle=\color{blue},
	columns=fullflexible,
	tabsize=2,
	showstringspaces=false,
	moredelim=**[is][\color{darkred}]{@}{@},
	moredelim=**[is][\color{barrier}]{$}{$},
	stepnumber=1,
	numberfirstline=false,
	escapechar=|,
	keepspaces,
    numberstyle=\tiny,
	numbersep=5pt,
    xleftmargin=1.5em,
    framexleftmargin=4pt,
    numbers=left,
}

\lstdefinestyle{trcode}{
	style=casecode,
	basicstyle=\small\ttfamily,
	keywordstyle=\small\color{blue}\ttfamily,
	commentstyle=\small\color{gray}\ttfamily,
	stringstyle=\small\color{blue}\ttfamily,
	numberstyle=\tiny\color{gray},
}

%% file: paper-cmds.tex
\newcommand{\sys}{\mbox{\textsc{VSync}}\xspace}

\newcommand{\semicheckmark}{\cmark\kern-1.1ex\raisebox{.7ex}{\rotatebox[origin=c]{125}{\textbf{--}}}}

\newcommand{\eg}{\emph{e.g.}}
\newcommand{\ie}{\emph{i.e.}}

\newcommand{\mo}[1][]{\text{\color{orange} mo}\xspace}
\newcommand{\rf}[1][]{\text{\color{teal}{}#1{}rf}\xspace}
\newcommand{\po}[1][]{\text{\color{blue} po}\xspace}

\newcommand{\cmark}{\text{\ding{51}}}
\newcommand{\xmark}{\text{\ding{55}}}
\mathchardef\hy="2D

\newcommand{\event}[3]{\langle\!\langle #1, \ #2 \rangle\!\rangle_{#3}}
\newcommand{\execution}[1]{\tikz[baseline=-0.75ex]{\node (fignum) {\tiny #1}; \draw[black] (fignum) circle [radius=3.5pt];}}
\tikzstyle{event}=[inner sep=1pt]
\newcommand{\barrier}[1]{\ifmmode\mathit{#1}\else\emph{#1}\fi}

\newcommand{\seqc}{\barrier{sc}}
\newcommand{\rlx}{\barrier{rlx}}
\newcommand{\acq}{\barrier{acq}}
\newcommand{\rel}{\barrier{rel}}

\newcommand{\SeqCst}{\seqc\xspace}

\newcommand{\Acquire}{\acq\xspace}
\newcommand{\Release}{\rel\xspace}

\newcommand\bh[1]{\tikz[remember picture] \node (begin highlight #1) {};}
\newcommand\eh[1]{\tikz[remember picture] \node (end highlight #1) {};}

\tikzstyle{decision} = [diamond, draw, fill=white, text=black, 
 text badly centered, outer sep =0pt, inner sep=0pt]
\tikzstyle{block} = [rectangle, draw,text=black,fill=black!10, 
text centered, rounded corners]
\tikzstyle{line} = [draw, -latex']
\tikzstyle{cloud} = [draw, ellipse,fill=blue!5, node distance=3cm,
minimum height=2em]
\colorlet{highlightnewstuffcolor}{yellow}

\newcommand{\finitegraphs}{\ensuremath{\mathbb G_\ast^F}}
\newcommand{\infinitegraphs}{\ensuremath{\mathbb G_\ast^\infty}}

\newcommand{\qspinopttime}{11 minutes}

\makeatletter
\def\thm@space@setup{\thm@preskip=1pt
	\thm@postskip=1pt}
\makeatother

\usepackage[compact]{titlesec}
\titlespacing*{\section}{0pt}{*2.5}{*1}
\titlespacing*{\subsection}{0pt}{*1.5}{*1}
\titlespacing*{\subsubsection}{0pt}{*1}{*1}

\newcommand{\etal}{\textit{et al.}}
\newcommand{\footnoteurl}[1]{\footnote{\;\url{#1}}}

\newtheorem{lemma}{Lemma}
\newtheorem{theorem}{Theorem}
\newtheorem{definition}{Definition}
\newtheorem*{definition*}{Definition}

\newif \ifcomments

\ifcomments

\else

\fi

%% file: authors.tex

\author[1,2]{Jonas~Oberhauser}
\author[1,2]{Rafael~Lourenco~de~Lima~Chehab}
\author[1,2]{Diogo~Behrens}
\author[1,2]{Ming~Fu}
\author[1,2]{Antonio~Paolillo}
\author[1,2]{Lilith~Oberhauser}
\author[1,2]{Koustubha~Bhat}
\author[2]{Yuzhong~Wen}
\author[2,3]{Haibo~Chen}
\author[1,2]{Jaeho~Kim}
\author[4]{Viktor~Vafeiadis}
\affil[1]{Huawei~Dresden~Research~Center}
\affil[2]{Huawei~OS~Kernel~Lab}
\affil[3]{Shanghai~Jiao~Tong~University}
\affil[4]{Max~Planck~Institute~for~Software~Systems (MPI-SWS)}

%% file: wt.tex
\hypertarget{s:spinning}{
\section{Await Model Checking in Detail}
}
\label{s:spinning}

AMC is an enhancement of \emph{stateless model checking} (SMC) capable of handling programs with awaits on WMMs.
SMC constructs all possible executions of a program, and filters out those inconsistent with respect to the underlying memory model.
However, SMC falls short when the program has infinitely many or non-terminating executions (\eg, due to await loops) because the check never terminates.
AMC overcomes this limitation by filtering out executions in which multiple iterations of an await loop read from the same writes.

We start introducing basic notation and definitions, including execution graphs, which are used to represent executions.
Next, we explain how awaits lead to infinitely many and/or non-terminating executions, and how AMC overcomes these problems.
We present sufficient conditions under which AMC correctly verifies programs including await termination (AT) and safety.
Finally, we show the integration of AMC into a stateless model checker from the literature.

\subsection{Notations and Definitions}


\paragraph{\emph{Executions as graphs.}}
An execution graph $G$ is a formal abstraction of executions,
where nodes are events such as reads and writes,
and the edges indicate (\po) program order, (\mo) modification order,
and (\rf) read-from relationships, as illustrated in \cref{ua execution graph}. 
A read event $R^m_T(x,v)$ reads the value $v$ from the variable $x$ by the thread $T$ with the mode $m$, a write event $W^m_T(x,v)$ updates $x$ with $v$, and $W_{init}(x,v)$ initializes. The short notations $R_T(x,v)$ and $W_T(x,v)$ represent the relaxed events $R^{\rlx}_T(x,v)$ and $W^{\rlx}_T(x,v)$ respectively.
Note that the \po is identical in \execution{a} and \execution{b} because it is
the order of events in the program text.
In contrast, \mo and \rf edges differ; \eg, in \execution{a}, $W_{T_1}(l,1)$ precedes $W_{T_2}(l,0)$ in \mo, while in \execution{b} it is the other way around.
Furthermore, the \rf edges indicate that $W_{T_1}(l,1)$ is never read in \execution{a}, while it is read by $R_{T_2}(l,1)$ in \execution{b}.

\begin{figure}[h]
	\hspace{-5mm}
	\begin{minipage}[b]{.5\textwidth}
		\hspace{-4mm}
	\begin{tabular}{c}
		\begin{lstlisting}[language=C,style=mainexample]
locked = 0, q = 0;
		\end{lstlisting}
		\\[1em]
		\begin{tabular}{c||c}
			 $T_1: $ \lstinline[language=C,style=mainexample]|lock| &  $T_2: $ \lstinline[language=C,style=mainexample]|unlock| \\
			{\lstset{showlines=true}
			\begin{lstlisting}[language=C,style=mainexample]
locked = 1;
q = 1;
while (locked == 1);
/* Critical Section */
			\end{lstlisting}}
			&
			{\lstset{showlines=true}
			\begin{lstlisting}[language=C,style=mainexample]
while (q == 0);
locked = 0;
assert (locked == 0);

		\end{lstlisting}}
		\end{tabular}
	\end{tabular}

	\caption{Awaits in one path of a partial MCS lock. $T_1$ signals \lstinline[language=C,style=mainexample]|q = 1| to notify $T_2$ that it enqueued, and $T_2$ waits for the notification, then signals \lstinline[language=C,style=mainexample]|locked = 0| to pass the lock to $T_1$.}\label{unbounded wait}
\end{minipage}\hspace{8pt}
\begin{minipage}[b]{.5\textwidth}
	\centering
	\begin{tabular}{cc}
		\begin{tikzpicture}
		\node (fignum) at (-1.65,0.1) {\tiny a};
		\draw[black] (fignum) circle [radius=3.5pt];
		\clip (-1.75,-4.3) rectangle (2.08, 0.25);
		\node[event] (INITx) at (-0.76,0) {$W_{init}(l,0)$};
		\node[event] (INITy) at (0.76,0) {$W_{init}(q,0)$};
		\node[event] (updx) at (-1,-0.9) {$W_{T_1}(l,1)$};
		\node[event] (updy) at (-1,-1.7) {$W^{\rel}_{T_1}(q,1)$};
		\node[event] (updx2) at (1,-3.3) {$W_{T_2}(l,0)$};
		
		\draw[->,draw=orange] (INITx) -- node[midway,left,orange] {\small mo} (updx);
		\draw[->,draw=orange] (INITy) -- node[right,orange,pos=0.6] {\small mo} (updy);
		\draw[->,draw=orange,inner sep=0] (updx) to[in = 170, out = 210] node[pos=0.7,right,orange] {\small mo} 
		(updx2);
		
		\draw[->,draw=blue] (updx) -- (updy) node[midway,left,blue] {\small po};
		
		\node[event] (y1) at (1,-0.9) {$R^{\acq}_{T_2}(q,0)$};
		\node[event] (y2) at (1,-1.7) {$R^{\acq}_{T_2}(q,0)$};
		\node[event] (y3) at (1,-2.5) {$R^{\acq}_{T_2}(q,1)$};
		
		\node[event] (x) at (1,-4.1) {$R_{T_2}(l,0)$};

		\node[event] (unlock) at (-1,-3.3) {$R_{T_1}(l,0)$};
		
		\draw[->,draw=blue] (y1) -- (y2) node[midway,left,blue] {\small po};
		\draw[->,draw=blue] (y2) -- (y3) node[midway,left,blue] {\small po};
		\draw[->,draw=blue] (y3) -- (updx2) node[midway,left,blue] {\small po};
		\draw[->,draw=blue] (updx2) -- (x) node[midway,left,blue] {\small po};
		
		\draw[->,draw=blue] (updy) -- (unlock) node[midway,left,blue] {\small po};
		
		\draw[->,draw=teal] (updx2) to[in = 30, out = 330] node[midway,left ,teal] {\small rf} (x);
		\draw[->,draw=teal] (INITy) to[in = 30, out = 340] node[midway,left ,teal] {\small rf} (y1);
		\draw[->,draw=teal] (INITy) to[in = 30, out = 340] node[midway,right,teal] {\small rf} (y2);
		\draw[->,draw=teal] (updy) --  node[midway,above,teal] {\small rf} (y3);
		\draw[->,draw=teal] (updx2) --  node[midway,above,teal] {\small rf} (unlock);
		
		\end{tikzpicture}
		&
		\begin{tikzpicture}
		\node (fignum) at (-1.65,0.1) {\tiny b};
		\draw[black] (fignum) circle [radius=3.5pt];
		\clip (-1.8,-4.3) rectangle (2.08, 0.25);
		\node[event] (INITx) at (-0.76,0) {$W_{init}(l,0)$};
		\node[event] (INITy) at (0.76,0) {$W_{init}(q,0)$};
		\node[event] (updx) at (-1,-0.9) {$W_{T_1}(l,1)$};
		\node[event] (updy) at (-1,-1.7) {$W^{\rel}_{T_1}(q,1)$};
		\node[event] (updx2) at (1,-3.3) {$W_{T_2}(l,0)$};
		
		\draw[->,draw=orange] (INITx) to[out = 310, in=162] node[near end,right,orange] {\small mo} (updx2);
		\draw[->,draw=orange] (INITy) -- node[midway,right,orange,pos=0.6] {\small mo} (updy);
		
		\draw[draw=yellow,opacity=0.5,line width=3mm] (updx)  -- (updy);
		\draw[->,draw=blue,very thick] (updx) -- (updy) node[midway,left,blue] {\small po};
		
		\node[event] (y1) at (1,-0.9) {$R^{\acq}_{T_2}(q,0)$};
		\node[event] (y2) at (1,-1.7) {$R^{\acq}_{T_2}(q,0)$};
		\node[event] (y3) at (1,-2.5) {$R^{\acq}_{T_2}(q,1)$};
		
		\node[event] (x) at (1,-4.1) {$R_{T_2}(l,1)$ \color{red}{\xmark}};
		
		\node[event] (unlockf) at (-1,-3.3) {$R_{T_1}(l,1)$};
		
		\node[event] (forever) at (-1,-4) {$\vdots$};		
		
		\draw[->,draw=blue] (y1) -- (y2) node[midway,left,blue] {\small po};
		\draw[->,draw=blue] (y2) -- (y3) node[midway,left,blue] {\small po};
		\draw[draw=yellow,opacity=0.5,line width=3mm] (y3)  -- (updx2);
		\draw[->,draw=blue,very thick] (y3) -- (updx2) node[midway,left,blue] {\small po};
		\draw[->,draw=blue] (updx2) -- (x) node[midway,left,blue] {\small po};
		
		\draw[->,draw=blue] (updy) -- (unlockf) node[midway,left,blue] {\small po};
		\draw[->,draw=blue] (unlockf) -- (forever.north) node[midway,left,blue] {\small po};
		
		\draw[->,draw=teal] (updx) to[in = 160, out = 211] node[midway,left ,teal] {\small rf} (x);
		\draw[->,draw=teal] (INITy) to[in = 30, out = 340] node[midway,left ,teal]{\small rf} (y1);
		\draw[->,draw=teal] (INITy) to[in = 30, out = 340] node[midway,right,teal] {\small rf} (y2);
		\draw[draw=yellow,opacity=0.5,line width=3mm] (updy) -- (y3);
		\draw[->,draw=teal,very thick] (updy) --  node[midway,above,teal] {\small rf} (y3);
		\draw[->,draw=teal] (updx) to[out=214, in=135]  
									(unlockf);
		
		\draw[draw=yellow,opacity=0.5,line width=3mm] (updx2) to[in = 214, out = 170, looseness=1.2] (updx);
		\draw[->,draw=orange,very thick] (updx2) to[in = 214, out = 170, looseness=1.2] node[pos = 0.3,right,orange] {\small mo} (updx);
		
		\end{tikzpicture}
	\end{tabular}
	\caption{Two execution graphs of \cref{unbounded wait} where $l =$ \lstinline[language=C,style=mainexample]|locked|.} \label{ua execution graph}
\end{minipage}
\end{figure}

\paragraph{\emph{Consistency predicates.}}
A weak memory model $M$ is defined by a consistency predicate $\mathrm{cons}_M$ over graphs, where $\mathrm{cons}_M(G)$ holds iff $G$ is consistent with $M$.
For instance, the `IMM' model used by \sys forbids the cyclic path\footnote{A path is cyclic if it starts and ends with the same node} of {\setlength{\fboxsep}{0pt}\colorbox{yellow!50}{highlighted}} edges in \execution{b} of \cref{ua execution graph} due to the \Release and \Acquire modes, as it forbids all cyclic paths consisting of edges in this order: 1) \po ending in $W^\rel$, 2) \rf ending in $R^\acq$, 3) \po, 4) \mo.
Such a path is compactly written as ``\(\po ; [W^{\rel}] ; \rf ; [R^{\acq}] ; \po ; \mo\)'', and is never cyclic in graphs consistent with IMM.
Thus $\mathrm{cons}_\mathrm{IMM}(\execution{b})$ does not hold.
If say the \rel\ barriers on the accesses to \lstinline[language=C,style=mainexample]|q| would be removed, the graph would be consistent with IMM.



\begin{figure}
	\centering
	\begin{minipage}{.4\textwidth}
        \begin{lstlisting}[language=C]
/* lock acquire */
do {
  atomic_await_neq(&lock, 1);
} while(atomic_xchg(&lock, 1) != 0);

x++; /* CS */

/* lock release */
atomic_write(&lock, 0);
\end{lstlisting}
\caption{TTAS lock example.} \label{fig:ttas-lock}
\end{minipage}
\end{figure}

\paragraph{\emph{Awaits}.}
Intuitively, an await is a special type of loop which waits for a write of another thread.
To make this intuition more precise, imagine a demonic scheduler that prioritizes threads currently inside awaits. 
Under such a scheduler, an await has two possible outcomes: either the write of the other thread is currently visible, and the await terminates immediately; or the write of the other thread is not visible.
In the latter case, the scheduler continuously prevents the write from becoming visible by never scheduling the writer thread, and hence the await never terminates.
A more precise definition for {\em await} is a loop that, for every possible value of the polled variables\footnote{%
Polled variables refers the variables read in each loop iteration to evaluate the loop's condition.},
either exits immediately or loops forever when executed in isolation.
We illustrate this with the two loops of the TTAS lock from  \cref{fig:ttas-lock}.
The inner loop is an await; to show this, we need to consider every potential value $v$ of \lstinline[language=C,style=mainexample]|lock|: 
for $v = 0$, the loop repeats forever, and for $v \not= 0$ the loop exits during the first iteration.
The outer loop is not an await; to show this, we need to find one value $v$ of \lstinline[language=C,style=mainexample]|lock| for which the loop is not executed infinitely often but also not exited immediately. One such value is $v=1$, for which the thread never reaches the outer loop again after entering the inner loop.
%

\paragraph{\emph{Sets of execution graphs.}}\label{infinity mc}
Given a fair scheduler, awaits either exit after a number of failed iterations or (intentionally) loop forever.
We separate execution graphs that satisfy $\mathrm{cons}_M$ into two sets, $\mathbb G^F$ and $\mathbb G^\infty$.
$\mathbb G^F$ is the set of execution graphs where all awaits exit after a finite number of failed iterations.
In \cref{ua execution graph}, for example, \execution{a} is in $\mathbb G^F$ since $\mathrm{cons}_M(\execution{a})$ holds
and the await of $T_2$ exits after two failed iterations. 
Note $\mathbb G^F$ consists exactly of the graphs in which all awaits terminate, but that does not imply that these graphs are finite: there may still be infinitely many steps \emph{outside} of the awaits. We will later state a sufficient condition that excludes such cases.
$\mathbb G^\infty$ is the set of the remaining consistent graphs. In each of these at least one await loops forever. We define:
\begin{definition}[\textbf{Await termination}]\label{def:at} \em  AT holds iff $\mathbb G^\infty$ is $\emptyset$.
\end{definition}
\noindent
Due to the barriers on \lstinline[language=C,style=mainexample]|q|, $\mathrm{cons}_\mathrm{IMM}(\execution{b})$ does not hold, and hence \execution{b} is not in $\mathbb G^\infty$. In fact, with these $\acq$ and $\rel$ barriers, all graphs with an infinite number of failed iterations violate consistency, and hence $\mathbb G^\infty$ is empty; AT is not violated.

Note that we can splice an additional failed iteration into \execution{a} by repeating $R^\acq_{T_2}(q,0)$, resulting in a new graph in $\mathbb G^F$. We generalize this idea. Let $\mathbb{G}_k \subseteq \mathbb{G}^F$ be the set of consistent
execution graphs with a total\footnote{We count here the sum of failed iterations of all executed instances of awaits, including multiple instances by the same thread, \eg, when the inner await loop in the TTAS lock from \cref{fig:ttas-lock} is executed multiple times.} of $k \in \{0,1,\ldots\}$ failed iterations.
Thus, $\mathbb G_0$ is the set of consistent graphs with no failed iterations. With two failed iterations of $T_2$'s await and zero of $T_1$'s await, $\execution a$ has a total of $2+0=2$ failed iterations of awaits, and is thus in $\mathbb{G}_2$.
Let now $G \in \mathbb{G}_k$ with $k>0$, \ie, $G$ has at least one failed iteration; we can always repeat one of its failed iterations to obtain a graph $G' \in \mathbb{G}_{k+1}$ due to the non-deterministic number of iterations of await loops.
Since all $\mathbb G_k$ are disjoint, their union (denoted by $\mathbb G^F = \biguplus_{k \in \Set{0,1,\ldots}} \mathbb G_k$) is infinite despite every set $\mathbb G_k$ being finite.

\subsection{Verifying Await Termination with AMC}
State-of-the-art stateless model checkers\cite{GenMC,HMC,RCMC} cannot construct all execution graphs in $\mathbb G^F$ or any in $\mathbb G^\infty$.
In order for SMC to complete in finite time, the user has to limit the search space to a finite subset of executions graphs in $\mathbb G^F$.
Consequently, SMC cannot verify AT (\cref{def:at}) and can only verify safety within this subset of executions graphs.





\newcommand{\innerloop}{\color{purple}{$B$}}
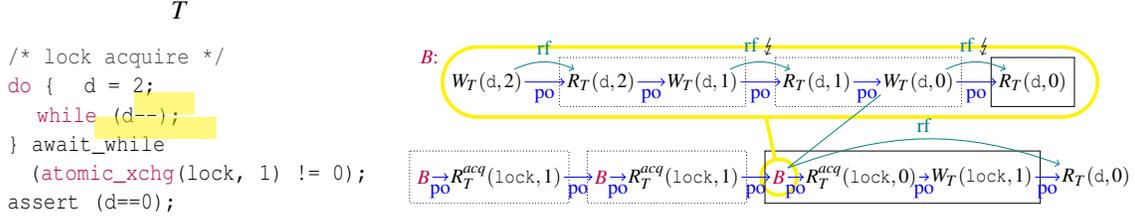
\begin{figure*}
	\centering
	\hfill
	\begin{minipage}{.30\linewidth}
		\centering $T$
		{\lstset{showlines=true,}			
			\begin{lstlisting}[language=C, style=mainexample, escapechar=$]
/* lock acquire */
do { $\bh 1$d = 2;$\eh 1$
  $\!\!\bh 2$while (d--);$\eh 2$
} await_while$\label{whiletas}$
  (atomic_xchg(lock, 1) != 0);
assert (d==0);
\end{lstlisting}}
	\end{minipage}
	\hfill
	\begin{minipage}{.63\linewidth}
		\resizebox{\textwidth}{!}{
		\small 
		\begin{tikzpicture}
		\node[event] (aWd0) at (0,0) {$W_{T}(\texttt{d},2)$};
		\node (IL) at (aWd0) at ($(aWd0)+(-0.9,0.4)$) {\innerloop:};
		\begin{scope}[local bounding box=ingroup1]
		\node[event] (aRd1) at ($(aWd0)+(1.8,0)$) {$R_{T}(\texttt{d},2)$};
		\node[event] (aWd1) at ($(aRd1)+(1.6,0)$) {$W_{T}(\texttt{d},1)$};
		\end{scope}
		\begin{scope}[local bounding box=ingroup2]
		\node[event] (aRd2) at ($(aWd1)+(1.8,0)$) {$R_{T}(\texttt{d},1)$};
		\node[event] (aWd2) at ($(aRd2)+(1.6,0)$) {$W_{T}(\texttt{d},0)$};
		\end{scope}
		\begin{scope}[local bounding box=ingroup3]
		\node[event] (aRd3) at ($(aWd2)+(1.8,0)$) {$R_{T}(\texttt{d},0)$};
		\end{scope}
		\node (rightsep) at ($(aRd3)+(0.8cm,0cm)$) {};
		\node[fit=(aWd0)(rightsep), draw=yellow, ultra thick, name path=testX, rounded rectangle, inner xsep=2pt,inner ysep=10pt] (X) {};
		\draw[densely dotted] ($(ingroup1.north west) + (-2.5pt, 5.5pt)$) rectangle ($(ingroup1.south east) + (2.5pt, -5.5pt)$);
		\draw[densely dotted] ($(ingroup2.north west) + (-2.5pt, 5.5pt)$) rectangle ($(ingroup2.south east) + (2.5pt, -5.5pt)$);
		\draw[] ($(ingroup3.north west) + (-2.5pt, 5.5pt)$) rectangle ($(ingroup3.south east) + (2.5pt, -5.5pt)$);
		\draw[->,draw=blue] (aWd0) -- node[midway,below,blue] {\small po} (aRd1);
		\draw[->,draw=blue] (aRd1) -- node[midway,below,blue] {\small po} (aWd1);
		\draw[->,draw=blue] (aWd1) -- node[midway,below,blue] {\small po} (aRd2);
		\draw[->,draw=blue] (aRd2) -- node[midway,below,blue] {\small po} (aWd2);
		\draw[->,draw=blue] (aWd2) -- node[midway,below,blue] {\small po} (aRd3);
		
		\draw[->,draw=teal] (aWd0) to[in = 155, out = 25] node[midway,above,teal] {\small rf} (aRd1);
		\draw[->,draw=teal] (aWd1) to[in = 155, out = 25] node[midway,above] {\small \rf $\lightning$} (aRd2);
		\draw[->,draw=teal] (aWd2) to[in = 155, out = 25] node[midway,above] {\small \rf $\lightning$} (aRd3);
		\coordinate (AWD0) at (current bounding box.center);
		\coordinate (FinalD) at (aWd2.south west);
		
		
		\begin{scope}[local bounding box=outgroup1]
		\node[event] (aB) at ($(aWd0)+(-1,-1.5)$) {\innerloop};
		\node[event] (aRx) at ($(aB)+(1.3,0)$) {$R^\acq_{T}(\texttt{lock} ,1)$};
		\end{scope}
		\draw[densely dotted] ($(outgroup1.north west) + (-2.5pt, 5.5pt)$) rectangle ($(outgroup1.south east) + (2.5pt, -5.5pt)$);
		\draw[->,draw=blue] (aB) -- node[midway,below,blue] {\small po} (aRx);

		\begin{scope}[local bounding box=outgroup2]
		\node[event] (bB) at ($(aRx)+(1.5,0)$) {\innerloop};
		\node[event] (bRx) at ($(bB)+(1.3,0)$) {$R^\acq_{T}(\texttt{lock},1)$};
		\end{scope}
		\draw[densely dotted] ($(outgroup2.north west) + (-2.5pt, 5.5pt)$) rectangle ($(outgroup2.south east) + (2.5pt, -5.5pt)$);
		\draw[->,draw=blue] (aRx) -- node[midway,below,blue] {\small po} (bB);
		\draw[->,draw=blue] (bB) -- node[midway,below,blue] {\small po} (bRx);
		
		\begin{scope}[local bounding box=outgroup3]
		\node[event] (cB) at ($(bRx)+(1.5,0)$) {\innerloop};
		\node[event] (cRx) at ($(cB)+(1.3,0)$) {$R^\acq_{T}(\texttt{lock},0)$};
		\node[event] (cWx) at ($(cRx)+(1.9,0)$) {$W_{T}(\texttt{lock},1)$};
		\end{scope}
		\coordinate (AB) at (cB);
		\draw[yellow, ultra thick, name path=testY] (AB) circle[radius=8pt];
		
		\draw[] ($(outgroup3.north west) + (-2.5pt, 5.5pt)$) rectangle ($(outgroup3.south east) + (2.5pt, -5.5pt)$);
		\draw[->,draw=blue] (bRx) -- node[midway,below,blue] {\small po} (cB);
		\draw[->,draw=blue] (cB) -- node[midway,below,blue] {\small po} (cRx);
		\draw[->,draw=blue] (cRx) -- node[midway,below,blue] {\small po} (cWx);

		\node[event] (Rd) at ($(cWx)+(1.8,0)$) {$R_{T}(\texttt{d},0)$};
		\coordinate (FINALB) at (cB.north east);
		\draw[->,draw=blue] (cWx) -- node[midway,below,blue] {\small po} (Rd);
		\draw[->,draw=teal] (cB.north east) to[out=20, in=160] node[midway,above,teal] {\small rf} (Rd.north west);

		\path[name path=connectfigs] (AB) -- (AWD0);
		\node[name intersections={of=testX and connectfigs, by={INT1}}] (int1) at ($(INT1)$) {};
		\draw[ultra thick,yellow,name intersections={of=testY and connectfigs, by={INT2}}] (int1.center) -- (INT2);
		\draw[draw=teal] (FinalD) -- (FINALB);		
		\end{tikzpicture}
		\begin{tikzpicture}[overlay,remember picture]
		\draw[yellow,line width=10pt,opacity=0.5] (begin highlight 1) -- (end highlight 1);
		\draw[yellow,line width=10pt,opacity=0.5] (begin highlight 2) -- (end highlight 2);
		\end{tikzpicture}
		}
	\end{minipage}
	\caption{Execution graph {\innerloop} represents the inner loop of $T$ (marked \colorbox{yellow!50}{yellow}). Failed await iterations are indicated with dotted boxes, the final (non-failed) iteration with a solid box. The inner loop of $T$ violates the Bounded-Effect principle, as \rf-edges (marked with $\lightning$) leave failed await iterations. The outer loop obeys the principle, as \rf-edges only leave the final await iteration.} \label{Bounded-Effect principle violated}
\end{figure*}

\paragraph{\em Key challenges.} For SMC to become feasible in our problem domain, we need to solve three key challenges:
\begin{description}[topsep=2pt, parsep=4pt, itemsep=0pt]
	\item[Infinity:] We need to produce an answer in finite time without a user-specified search space, even though the search space $\mathbb G^F \cup \mathbb G^\infty$ is infinite.
	\item[Soundness:] We need to make sure not to miss any execution graph that may potentially uncover a safety bug.
	\item[Await termination:] We need to verify that $\mathbb G^\infty$ is $\emptyset$.
\end{description}
Under certain conditions specified later, AMC overcomes these
challenges through three crucial implications:
\begin{enumerate}[topsep=2pt, parsep=4pt, itemsep=0pt]
	\item The infinite set $\mathbb G^F$ is collapsed into a finite set of finite executions graphs $\finitegraphs \subset \mathbb G^F$.
	Moreover, the infinite execution graphs in $\mathbb G^\infty$ are collapsed into finite execution graphs in a (possibly infinite) set \infinitegraphs{}.
	AMC explores at most all graphs in \finitegraphs{} and up to one graph in \infinitegraphs{}.
\item For all $G\in \mathbb G^F$, there exists $G'\in \finitegraphs{}$ such that $G$ and $G'$ are equivalent. Thus, bugs present in $\mathbb G^F$ are also in \finitegraphs{}.
	\item  Detecting whether there exists a finite execution graph in \infinitegraphs{} is sufficient to conclude whether $\mathbb G^\infty$ is empty.
	Thus AMC can stop after exploring one graph in the set \infinitegraphs{} and report an AT violation, and if AMC does not come across such a graph, AT is not violated.
\end{enumerate}
We now explain how AMC achieves these three implications, as well as the conditions under which it does so.

\paragraph{\em The key to AMC.}
\begin{figure}[b!]
	\centering
	\begin{tabular}{cc}
		\multicolumn{1}{c|}{
		\begin{tikzpicture}[baseline=-60pt,auto]
		\node (fignum) at (-0.6,0.1) {\tiny \textalpha};
		\draw[black] (fignum) circle [radius=3.5pt];
		\node[event] (INIT) at (0.9,0) {$W_\mathit{init}(l,0)$};
		\node[event] (x1) at (0,-0.65) {$W_{T_1}(l,1)$};
		\node[event] (x2) at (1.8,-0.65
		) {$W_{T_2}(l,0)$};
		
		\draw[->,draw=orange] (INIT) -- (x1) node[midway,left,orange] {\small mo};
		\draw[->,draw=orange] (x1) -- (x2)  node[midway,above,orange] {\small mo};
		
		\node[event] (ra) at (0,-1.6) {$R_{T_1}(l,1)$};
		\draw[->,draw=teal] (x1) to[out=225,in=135] node[left,teal] {\small rf}  (ra);
		\draw[->,draw=blue] (x1) -- (ra)  node[midway,right,blue] {\small po};
		
		\node[event] (rb) at (0,-2.35) {$R_{T_1}(l,1)$};
		\draw[->,draw=blue] (ra) -- (rb)  node[midway,right,blue] {\small po};
		\draw[->,draw=teal] (x1) to[out=225,in=135] (rb);

		\node[event] (rc) at (0,-3.1) {$R_{T_1}(l,0)$};
		\draw[->,draw=blue] (rb) -- (rc)  node[midway,right,blue] {\small po};
		\draw[->,draw=teal] (x2) to[out=255,in=25]  node[midway,above,teal] {\small rf}  (rc);
		
		\end{tikzpicture}} & 
		\begin{tikzpicture}[baseline=-60pt]
		\node (fignum) at (-0.6,0.1) {\tiny \textbeta};
		\draw[black] (fignum) circle [radius=3.5pt];
		\node[event] (INIT) at (0.9,0) {$W_\mathit{init}(l,0)$};
		\node[event] (x1) at (0,-0.65) {$W_{T_1}(l,1)$};
		\node[event] (x2) at (1.8,-0.65) {$W_{T_2}(l,0)$};
		
		\draw[->,draw=orange] (INIT) -- (x2) node[midway,right,orange] {\small mo};
		\draw[->,draw=orange] (x2) -- (x1)  node[midway,above,orange] {\small mo};
		
		\node[event] (ra) at (0,-1.6) {$R_{T_1}(l,1)$};
		\draw[->,draw=teal] (x1) to[out=225,in=135] node[left,teal] {\small rf}  (ra);
		\draw[->,draw=blue] (x1) -- (ra)  node[midway,right,blue] {\small po};
		
		\node[event] (rb) at (0,-2.35) {$R_{T_1}(l,\lightning)$};
		\draw[->,draw=blue] (ra) -- (rb)  node[midway,right,blue] {\small po};
		
		\end{tikzpicture} \\
		\hline
		\begin{tikzpicture}
		\node (fignum) at (-0.6,0.1) {\tiny 1};
		\draw[black] (fignum) circle [radius=3.5pt];
		\node[event] (INIT) at (0.9,0) {$W_\mathit{init}(l,0)$};
		\node[event] (x1) at (0,-0.65) {$W_{T_1}(l,1)$};
		\node[event] (x2) at (1.8,-0.65) {$W_{T_2}(l,0)$};
		
		\draw[->,draw=orange] (INIT) -- (x1) node[midway,left,orange] {\small mo};
		\draw[->,draw=orange] (x1) -- (x2)  node[midway,above,orange] {\small mo};
		
		\node[event] (ra) at (0,-1.6) {$R_{T_1}(l,1)$};
		\draw[->,draw=teal] (x1) to[out=225,in=135] node[left,teal] {\small rf}  (ra);
		\draw[->,draw=blue] (x1) -- (ra)  node[midway,right,blue] {\small po};
		
		\node[event] (rb) at (0,-2.35) {$R_{T_1}(l,0)$};
		\draw[->,draw=blue] (ra) -- (rb)  node[midway,right,blue] {\small po};
		\draw[->,draw=teal] (x2) to[out=235,in=25]  node[midway,above,teal] {\small rf}  (rb);
		
		\end{tikzpicture} &
		\begin{tikzpicture}
		\node (fignum) at (-0.6,0.1) {\tiny 2};
		\draw[black] (fignum) circle [radius=3.5pt];
		\node[event] (INIT) at (0.9,0) {$W_\mathit{init}(l,0)$};
		\node[event] (x1) at (0,-0.65) {$W_{T_1}(l,1)$};
		\node[event] (x2) at (1.8,-0.65) {$W_{T_2}(l,0)$};
		
		\draw[->,draw=orange] (INIT) -- (x1) node[midway,left,orange] {\small mo};
		\draw[->,draw=orange] (x1) -- (x2)  node[midway,above,orange] {\small mo};
		
		\node[event] (ra) at (0,-1.6) {$R_{T_1}(l,0)$};
		\draw[->,draw=teal] (x2) to[out=245,in=15]  node[midway,above,teal] {\small rf}  (ra);
		\draw[->,draw=blue] (x1) -- (ra)  node[midway,right,blue] {\small po};
		
				\node[event, draw=none] (rb) at (0,-2.35) {\phantom{$R_{T_1}(l,0)$}};
		\end{tikzpicture}
	\end{tabular}
	\caption{Execution graphs of \cref{unbounded wait} where $l =$ \lstinline[language=C,style=mainexample]|locked|.} \label{await AP}
\end{figure}
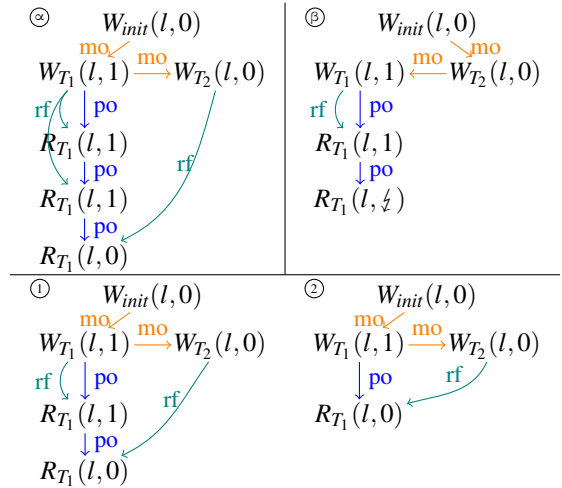

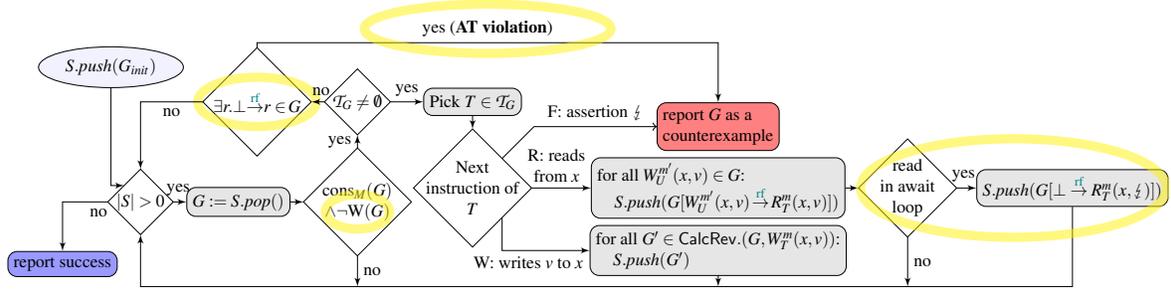
\begin{figure*}[t]
	\centering
	\resizebox{\textwidth}{!}{
	\begin{tikzpicture}[scale = 0.75, every node/.style={scale=0.75}, auto]
	\node [cloud] (start) {$S.\mathit{push}(G_{init})$};
	\node [decision, below right =1.8cm and 0.2cm of start.center] (done) {$\left| S \right| > 0$};
	\node [block, right = 0.2cm of done] (pop) {$G := S.\mathit{pop}()$};
	\node [decision, right= 0.2cm of pop,align=center] (cons) {$\mathrm{cons}_M(G)$\\$\land \neg\mathrm{W}(G)$};
	\draw [draw=highlightnewstuffcolor,line width=4pt,opacity=0.5]($(cons)+(0,-0.2)$) ellipse (0.65cm and 0.35cm);

	\node [decision, above= 0.2cm of cons] (term) {$\mathcal T_G \not= \emptyset$};
	\node [decision, left= 0.2cm of term,align=center] (AT) {$\exists r. \bot{\overset\rf\rightarrow}r\in\!G$};
	\draw [draw=highlightnewstuffcolor,line width=4pt,opacity=0.5](AT) ellipse (1.2cm and 0.5cm);
	\node [block, right = 0.5cm of term] (pick) {Pick $T \in \mathcal T_G$};
	\node [decision, below = 0.2cm of pick,inner sep =-4pt,align=center] (switch) {Next\\instruction of\\$T$};
	
	\path [line] (cons) -- node [midway] {yes} (term);
	
	\node [block,right=1.8 cm of switch.center,align=left] (read) {for all $W^{m'}_U(x,v) \in G$:\\
		\hspace{1em}$S.\mathit{push}(G[W^{m'}_U(x,v) \!\overset{\rf}\rightarrow\! R^{m}_T(x,v)])$}; 
	\path [line] (switch) -- node [above, pos=0.42,align=center] {R: reads\\from $x$} (read);

	\node [block,below=0.1 cm of read,align=left] (write) {for all $G' \in \mathsf{CalcRev.}(G,W^m_T(x,v))$:
		\\\hspace{1em}$S.\mathit{push}(G')$}; 
	\path [line] (switch.south east) |- node [below, pos=0.65] {W: writes $v$ to $x$} (write);

	\node [decision, right = 0.2cm of read,align=center, inner sep=-2pt] (inUA) {read\\in await\\loop};
	\node [block, right = 0.3 cm of inUA, align=left] (UAread) {$S.\mathit{push}(G[\bot \overset{\rf}\rightarrow R^{m}_T(x,\lightning)])$};
	
	\node [block,above = 0.1 cm of read,align=left,fill=red!50,text=black] (error) {report $G$ as a\\counterexample}; 
	
	\path [name path=angledpath
	] (switch.north east) -- ++ (45:2cm);
	\path [name path=horizontalpath
	] (error) -- +(-4cm,0);
	\draw [->,name intersections={of=angledpath and horizontalpath, by=X}] (switch.north east) -- (X) -- node {F: assertion $\lightning$} (error);
	
	\path [line] (AT.north) |- +(0,0.1cm) -| node [near start] (ATyes) {yes (\textbf{AT violation})} (error);
	\draw [draw=highlightnewstuffcolor,line width=4pt,opacity=0.5] (ATyes) ellipse (2.5cm and 0.5cm);
	
	\path [line] (term) -- node [above, near start] {no} (AT);
	\path [line] (AT) -| node [near start] {no} (done);
	\path [line] (term) -- node {yes} (pick);
	\path [line] (pick) -- (switch);
	\path [line] (pop) -- (cons);
	\path [line] (done) -- node [pos=0.38] {yes} (pop);
	
	\path [line] (read) -- (inUA);
	\path [line] (inUA) -- node [pos=0.35] {yes} (UAread);
	
	\path [line, name path=returnarrow] (UAread) -- +(0,-2cm) -| (done);
	\node [block,below left= 0.5cm and 0.1cm of done,fill=blue!40,text=black] (end) {report success};
	\clip (current bounding box.south west) rectangle (current bounding box.north east);

	\path [name path=downinUA] (inUA) -- +(0,-2cm);
	\path [name path=downwrite] (write) -- +(0,-2cm);
	\path [name path=downcons] (cons) -- +(0,-2cm);
	
	\path [line, name intersections={of=downinUA and returnarrow, by=XinUA}] (inUA) -- node{no} (XinUA);
	\path [line, name intersections={of=downwrite and returnarrow, by=Xwrite}] (write) -- (Xwrite);
	\path [line, name intersections={of=downcons and returnarrow, by=Xcons}] (cons) -- node{no} (Xcons);
	\draw [draw=highlightnewstuffcolor,line width=4pt,opacity=0.5] ($(inUA.west)!0.5!(UAread.east)$) ellipse (3.1cm and 1cm);

	\path [line] (done.west) -| node [below,pos=0.1] {no} (end);
	\path [line] (start) |- (done.north west);
	\end{tikzpicture}
}
	\caption{AMC exploration}\label{exploration algorithm}
\end{figure*}
In contrast to existing SMCs, AMC filters out execution graphs that contain
awaits where multiple iterations read from the same writes to the polled variables.
This idea is captured by the predicate $\mathrm{W}(G)$ which defines wasteful executions.
\begin{definition} [\textbf{Wasteful}]\em An execution graph $G$ is wasteful, \ie, $\mathrm{W}(G)$ holds, if an await in $G$ reads the same combination of writes in two consecutive iterations.
\end{definition}

AMC does not generate wasteful executions, as they do not add any additional information.
For instance, AMC does not generate execution graph $\execution{\textalpha} \in \mathbb G_2$ from \cref{await AP} because $\mathrm{W}(\execution\textalpha)$ holds:
$T_1$  reads from its own write twice in the await.
Similarly, AMC does not generate any of the infinitely many variations of \execution{\textalpha} in which $T_1$ reads even more often from that write.
Instead, if we remove all references to \lstinline[language=C,style=mainexample]|q| in the program of \cref{unbounded wait}, AMC generates only the two execution graphs in $\mathbb G_\ast^F = \Set{\execution{1},
\execution{2} }$, in which each write is read at most once by the await of $T_1$.
$T_1$ can read from at most two different writes, thus there is at most one failed iteration of the await in the execution graphs in \finitegraphs{}, and we have $\finitegraphs{} \subseteq  \biguplus_{k\in\Set{0,1}}\mathbb G_k$; in general, if there are at most $n \in \mathbb{N}$ writes each await can read from, there are at most
$n-1$ failed iterations of any await in graphs in 
$\mathbb{G}^F_\ast$.
If there are at most $a \in \mathbb N$ executed instances of awaits, we have
$\finitegraphs \subseteq  \biguplus_{k\in\Set{0,\ldots,a \cdot (n-1)}}\mathbb G_k$ which is a union of a finite number of finite sets and thus finite. We will later define sufficient conditions to ensure this.

We proceed to discuss how AMC discovers AT violations.
Consider execution graph \execution{\textbeta}. 
In the first iteration of the await, $T_1$ reads from $W_{T_1}(l,1)$.
In the next iteration, $T_1$'s read has no incoming \rf-edge; 
coherence forbids $T_1$ from reading an older write and the
await progress condition forbids it from reading the same write.
Since there is no further write to the same location, AMC detects an AT violation and uses the finite graph \execution{\textbeta} as the evidence.
In general, if the \mo of every polled variable is finite, then AT violations from $\mathbb G^\infty$ are represented by graphs in $\infinitegraphs$ where some read has no incoming \rf-edge. AMC exploits this fact to detect AT violations. 

\paragraph{\em Conditions of AMC.} \label{amc conditions}
State-of-the-art SMC only terminates and produces correct results for terminating and loop-free programs\footnote{Exploration depths are often used to transform programs with loops into loop-free programs, potentially changing the behavior of the program.}.
AMC extends the domain of SMC to a fragment of looping and/or non-terminating programs,
on which AMC not only produces correct results but can also decide termination.
With the generic client code provided by \sys, all synchronization primitives we have studied are in this fragment, showing it is practically useful.
The fragment includes
all programs satisfying the following two principles:
\begin{description}[topsep=2pt, parsep=4pt, itemsep=0pt]
	\item[Bounded-Length Principle:]
		There is a bound $b$ (chosen globally for the program) so that all executions in $\mathbb{G}_1$ of the program have length $\le b$.
\item[Bounded-Effect Principle:] Failed await iterations satisfy the \emph{bounded-effect principle}, that the effect of the loop iteration is limited to that loop iteration. 
\end{description}

Informally, \textbf{the Bounded-Length principle} means that the number of execution steps outside of awaits is bounded, and each individual iteration of a failed await is also bounded.
Obviously, infinite loops in the client code are disallowed by the Bounded-Length principle.

\textbf{The Bounded-Effect principle} means that no side-effects from failed await iterations must be referenced by either subsequent loop iterations, other threads, or outside the loop.
The principle can be defined more precisely in terms of execution graphs: \rf-edges starting with writes\footnote{Only writes that change the value of the variable matter here.} generated by a failed await iteration must go to read events that are generated in the same iteration.
\Cref{Bounded-Effect principle violated} illustrates the principle: \rf-edges from decrements in the failed iterations of the loop body {\innerloop} go to subsequent iterations of the loop, but for the outer loop only the final iteration has outgoing \rf-edges.
The Bounded-Effect principle allows removing any failed iteration from a graph without affecting the rest of the graph since the effects of the failed iteration are never referenced outside the iteration.
This implies that any bugs in graphs from $\mathbb{G}^F$ are also present in graphs in $\mathbb G_1$.
Furthermore, if the Bounded-Effect principle and the Bounded-Length principle hold,
then graphs in $\mathbb G_k$ are bounded for every $k$.
The bound for $\mathbb G_k$ can be computed as \mbox{$b + (k-1) \cdot x$} where $b$ is the bound for $\mathbb G_1$ and $x$ is the maximum number of steps in a failed iteration of an await in $\mathbb G_1$. 

The two principles jointly imply that {\mo}s and the number of awaits are bounded, thus, as discussed before, $\finitegraphs$ is a finite set and $\infinitegraphs$ contains only finite graphs, and AMC always terminates.
In synchronization primitives, awaits either just poll a variable (without side effects) or perform some operation which only changes global state if it succeeds, \eg,
\lstinline[language=C,style=mainexample]{await_while(q==0);} \quad or \quad \lstinline[language=C,style=mainexample]{await_while(!trylock(&L));}
These awaits satisfy the Bounded-Effect principle. The former does not have any side effects. The latter encapsulates its local side effects inside \lstinline[language=C,style=mainexample]|trylock(&L)|, which can therefore not leave failed iteration of the loop. A global side effect (\ie, acquiring the lock) only occurs in the last iteration of the await (cf. \cref{Bounded-Effect principle violated}).
When called in our generic client code, synchronization primitives also satisfy the Bounded-Length principle: the client code invokes the functions of the primitives only a bounded number of times, and each function of the primitives is also bounded.


\paragraph{\em AMC Correctness.}
For programs which satisfy {the Bounded-Length  principle} and {the Bounded-Effect principle}, 1) AMC terminates, 2) AMC detects every possible safety violation, 3) AMC detects every possible non-terminating await, and 4) AMC has no false positives.
See \cref{s:amc} for the formal proof.

\subsection{Implementing AMC} \label{AMCawaitwhile}

We implement AMC on top of GenMC \cite{GenMC,HMC}, a highly
advanced SMC from the literature.
The exploration algorithm in \cref{exploration algorithm} extends GenMC's algorithm  with the {\setlength{\fboxsep}{0pt}\colorbox{yellow!50}{highlighted}} essential changes: 1) detecting AT violations through reads with no incoming \rf-edge; 2) checking if $\mathrm{W}(G)$ holds to filter out graphs $G$ in which an await reads from the same writes in multiple iterations.
AMC builds up execution graphs through a kind of depth-first search, starting with an empty graph $G_\mathit{init}$, which is extended step-by-step with new events and edges.
The search is driven by a stack $S$ of possibly incomplete and/or inconsistent graphs that is initialized to contain only $G_\mathit{init}$.

Each iteration of the exploration pops a graph $G$ from $S$. If the graph $G$
violates the consistency predicate $\mathrm{cons}_M$
or is wasteful (\ie, $\mathrm{W}(G)$ holds), it is discarded and the iteration ends.
Otherwise, a program state is reconstructed by emulating the execution of threads until every thread executed all its events in $G$; \eg, if a thread $U$ executes a read instruction that corresponds to $R_U(x,0)$ in $G$, the emulator looks into $G$ for the corresponding read event and returns the value read by the event, in this case $0$.
We denote the set of runnable threads in the reconstructed program state by $\mathcal T_G$.
Initially, all threads are runnable; a thread is removed from $\mathcal T_G$ once it terminates or if it is stuck in an await.
If the set is not empty, we can explore further and pick some arbitrary thread $T \in \mathcal T_G$ to run next.
We emulate the next instruction of $T$ in the reconstructed program state. 
For the sake of brevity, we discuss only three types of instructions:
failed assertions\footnote{The assertion expression \lstinline[language=C,style=mainexample]|assert(x==0)| consists of at least two instructions: the first reads \lstinline[language=C,style=mainexample]|x|, and the second just compares the result to zero and potentially fails the assertion. Note that once the \rf-edge for the event corresponding to the first instruction has been fixed, whether the second instruction is a failed assertion or not is also fixed.}, writes, and reads.
\begin{description}[topsep=2pt, parsep=4pt, itemsep=0pt]
	\item[F:] an assertion failed. We stop the exploration and report $G$. 
\end{description}
If the instruction executes a read or write, a new graph with the corresponding event should be generated. Usually there are several options for the event; for each option, a new graph is generated and pushed on a stack $S$. In particular:
\begin{description}[topsep=2pt, parsep=4pt, itemsep=0pt]
	\item[W:] a write event $w$ is added. In this case, for every existing read $r$ to the same variable, a partial\footnote{For details we refer the reader to the function $\mathsf{CalcRevisits}$ from \cite{GenMC}.} copy of $G$ with an edge $w \overset{\rf[\tiny]}\rightarrow r$ is generated and pushed into $S$.
	\item[R:] a read event $r$ is added; for every write event $w$ in $G$ a copy of $G$ with an additional edge $w \overset{\rf[\tiny]}\rightarrow r$ is generated and pushed into $S$.
\end{description}
Crucially, if the read event $r$ is in an await, an additional copy of $G$ is generated in which $r$ has no incoming \rf-edge (we write this new graph as $G[\bot \overset{\rf[\tiny]}\rightarrow r]$). 
This missing \rf-edge indicates a potential AT violation. It is not an actual AT violation \emph{yet} because a new write $w$ to the same variable might be added by another thread later during exploration.
If such a write is added, it leads to the generation of two types of graphs: graphs in which $r$ is still missing an \rf-edge, and graphs with the edge $w \overset{\rf[\tiny]}\rightarrow r$ where $r$ no longer has a missing \rf-edge (and the potential AT violation got resolved).
Otherwise, if a missing \rf-edge is still present and no other thread can be run, we know that such a write cannot become available anymore; the potential AT violation turns into an actual AT violation.
The algorithm detects the violation after popping a graph $G$ in which no threads can be run ($\mathcal T_G = \emptyset$), but a read without incoming \rf-edge is present ($\bot\overset{\rf[\tiny]}\rightarrow r\in G$), and this missing \rf-edge could not be resolved except through a wasteful execution, \ie, every consistent graph $G'$ obtained by adding the missing \rf-edge and completing the await iteration is wasteful.

%% file: amc.tex
\section{Proving the Correctness of Await Model Checking}\label{s:amc}
In this section we formally prove \cref{thm:amc} to show that our AMC is correct. In order to formally prove this theorem, first we define a tiny concurrent assembly-like language in \S\ref{sec:lan}, which follows our Bounded-Length principle and  allows us to map the execution graphs to the execution of instruction sequences. Then we formalize the Bounded-Effect principle in \S\ref{sec:vegas}. Finally we give the formal representation of \cref{thm:amc} and prove it in \S\ref{sec:main}.
Throughout this section we use the notation of \cite{GenMC}.

\begin{theorem}[\textbf{AMC Correctness}]\label{thm:amc}\em
	For programs which satisfy {the Bounded-Length principle} and {the Bounded Effect principle}, 1) AMC terminates, 2) AMC detects every possible safety violation, 3) AMC detects every possible non-terminating await, and 4) AMC has no false positives.
\end{theorem}

\subsection{The tiny concurrent assembly-like language}\label{sec:lan}
To have a formal foundation for these proofs we need to provide a formal programming language.
Without such a programming language, execution graphs just float in the air, detached from any program. 
Using the consistency predicate $\mathrm{cons}_M$ we can state that an execution graph can be generated by the weak memory model, but not whether it can be generated by a given program.
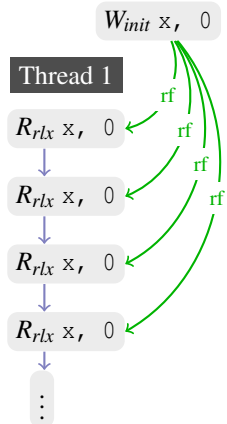
\begin{figure}[H]
	\centering
	\begin{tikzpicture}[
	node distance = 9mm and 2mm,
	po/.style={->,thick, blue!50!black!50},
	rf/.style={->,thick, green!70!black},
	mo/.style={->,thick, red},
	ctrl/.style={->,thick, dashed},
	ev/.style={rounded corners, fill=gray!10},
	evx/.style={rounded corners, fill=red!70!black!40},
	event/.style={ev, anchor=west},
	thread/.style={text=white, fill=black!70},
	]
	\node[ev] (l) {$W_\mathit{init}$ {\tt x, 0}};
	
	\node[below left=9mm and -4mm of l,ev] (t1 init next) {$R_\rlx$ {\tt x, 0}};
	\node[below=of t1 init next.west,event] (t1 init lock) {$R_\rlx$ {\tt x, 0}};
	\node[below=of t1 init lock.west,event] (t1 xchg read) {$R_\rlx$ {\tt x, 0}};
	\node[below=of t1 xchg read.west,event] (t1 xchg tail) {$R_\rlx$ {\tt x, 0}};
	\node[below right=9mm and 3mm of t1 xchg tail.west,event] (t1 write next) {\vdots};
	
	\draw[po] ($(t1 init next.south west)+(5mm,0)$) -- ($(t1 init lock.north west)+(5mm,0)$);
	\draw[po] ($(t1 init lock.south west)+(5mm,0)$) -- ($(t1 xchg read.north west)+(5mm,0)$);
	\draw[po] ($(t1 xchg read.south west)+(5mm,0)$) -- ($(t1 xchg tail.north west)+(5mm,0)$);
	\draw[po] ($(t1 xchg tail.south west)+(5mm,0)$) -- ($(t1 write next.north west)+(2mm,0)$);
	

	\draw[rf] (l) to[bend right=-50]
	node[midway, fill=white, font=\footnotesize] {rf}
	(t1 init next.east);
	
	\draw[rf] (l) to[bend right=-50]
	node[midway, fill=white, font=\footnotesize] {rf}
	(t1 init lock.east);
	
	\draw[rf] (l) to[bend right=-50]
	node[midway, fill=white, font=\footnotesize] {rf}
	(t1 xchg read.east);
	
	\draw[rf] (l) to[bend right=-50]
	node[midway, fill=white, font=\footnotesize] {rf}
	(t1 xchg tail.east);
	
	\node[thread, above=2mm of t1 init next] {Thread 1};
	\end{tikzpicture}
	\caption{Divergent execution graph} \label{non-terminating execution}
\end{figure}
For example, the non-terminating execution graph in \cref{non-terminating execution} consisting only of reads from the initial store is consistent with all standard memory models, but is obviously irrelevant to most programs.
If we were only to decide whether a non-terminating execution graph is consistent with the memory model or not, we could simply always return \lstinline[language=C,style=mainexample]|true| and be done with it.
What we really want to decide is whether a given program (which satisfies our two principles) has a non-terminating execution graph which is consistent with the memory model.
For this purpose we will in this section define a tiny assembly-like programming language, and define whether a program $P$ can generate an execution graph $G$ through a new consistency predicate $\mathrm{cons}^P(G)$.
This will require formally defining an execution-graph driven semantics for the language, in which threads are executed in isolation using values provided by an execution graph.
After we define the programming language, we formally define the Bounded-Effect principle.
We then show that if the Bounded-Effect principle is satisfied, we can always remove one failed iteration of an await from a graph without making the graph inconsistent with the program.
This will allow us to show that graphs in $\mathbb G^F$ can always be ``trimmed'' to a graph in \finitegraphs{} which has the same error events, and thus all safety violations are detected by AMC.
Next we show that we can also add failed iterations of an await; indeed, for graphs in \infinitegraphs{} we can add infinitely many such iterations.
Thus these graphs can be ``extended'' to graphs in $\mathbb G^\infty$ which are consistent with the program.
This implies that there are not false positives.
Finally we show that due to the Bounded-Length dance principle and the Bounded-Effect principle, graphs in \finitegraphs{} and \infinitegraphs{} have a bounded number of failed iterations of awaits, and the remaining steps are bounded as well.
Thus the search space itself must be finite, and AMC always terminates.
\lstdefinelanguage{Lambda}{%
	morekeywords={%
		if,then,else,fix,match,await,step,with 
	},%
	morekeywords={[2]int},   
	otherkeywords={:}, 
	literate={
		{->}{{$\to$}}{2}
		{lambda}{{$\lambda$}}{1}
		{rlx}{{$^\rlx$}}{2}
		{|}{{$\big|$}}{1}
		{sigma}{{$\sigma$}}{1}
	},
	basicstyle={\sffamily},
	keywordstyle={\bfseries},
	keywordstyle={[2]\itshape}, 
	keepspaces,
	columns=flexible,
	mathescape 
}[keywords,comments,strings]%

\newcommand{\mode}{m}
\newcommand{\evt}{e}
\newcommand{\state}{\sigma}
\newcommand\dowhile[2]{\lstinline[language=Lambda]|await(#1,#2)|}
\newcommand\loopcons{\kappa}
\newcommand\inst[2]{\lstinline[language=Lambda]|step(#1,#2)|}
\newcommand{\egen}{\epsilon}
\newcommand{\strans}{\delta}
\newcommand{\update}{\mu}
\begin{figure}[H]
	\begin{align*}
	\textit{(Program)} \ P &:= T_0 \parallel \ldots  \parallel \ T_i \parallel\ldots  \ \parallel T_n && \text{Composition of parallel threads $T_i$} \\
	\textit{(Thread)} \ T &:= S_0 ; \ \ldots ; \ S_n && \text{Sequence of statements $S_i$} \\
	\textit{(Stmts)} \ S &:=  \dowhile{$n$}{$\loopcons$} \mid \inst{$\egen$}{$\strans$} && \text{Await-loop and non-await-loop steps.}
	\\ \midrule 
	\textit{(LoopCon)} \   \loopcons &\in \mathit{State} \to \Set{0,1} && \text{Loop condition}
	\\  \textit{(EvtGen)} \  \egen &\in \mathit{State} \to \mathit{Events} && \text{Event generator}
	\\  \textit{(StTrans)} \  \strans &\in \mathit{State} \times \mathit{Value}? \to \textit{Update} && \text{State transformer}
	\\ \midrule
	\textit{(Events)} \ \evt &:= R^\mode(x) \mid W^\mode(x,v) \mid F^\mode \mid E && \text{Read, write, fence, error events. $x \in \mathit{Location}$, $v \in \mathit{Value}$.}
	\\ \textit{(Modes)} \ \mode &:= \rlx \mid \rel \mid \acq \mid \seqc && \text{Barrier modes}
	\\ \midrule  
	\textit{(State)}\ \state  & \in \mathit{Register} \to \mathit{Value} && \text{Set of thread-local states.}
	\\
	\textit{(Update)}\ \update  & \in \mathit{Register} \rightharpoonup \mathit{Value} && \text{Updated register values.}
	\\
	\textit{(Register)} & \quad \ldots && \text{Set of thread-local registers.}
	\\  \textit{(Location)} & \quad  \ldots&& \text{Set of shared memory locations.}
	\\ \textit{(Value)} & \quad \ldots  && \text{Set of possible values of registers and memory locations.}
	\end{align*}
	\caption{Compact Syntax and Types of our Language}\label{syntax summary}
\end{figure}

\begin{figure}[H]
	\centering
	\begin{tabular}{lp{0.05\textwidth}l}
		C-like program && Our toy language \\
		{\lstset{showlines=true,}
			\begin{lstlisting}
			x = r1;
			r1 = y;
			if (r1 == 0)
			r2 = x;

			\end{lstlisting}}
		&&
		\begin{lstlisting}[language=Lambda]
		step(lambda sigma. Wrlx(x, sigma(r1)), lambda sigma _. [ ]);
		step(lambda sigma. Rrlx(y),        lambda sigma v. [r1 -> v]);
		step(lambda sigma. match sigma(r1) with
		| 0 -> Rrlx(x)
		| _ -> Frlx,
		lambda sigma v. match sigma(r1) with
		| 0 -> [r2 -> v]
		| _ -> [ ])             
		\end{lstlisting}
	\end{tabular}
	\caption{Using lambda functions inside \lstinline[language=Lambda]|step| to implement different control paths} \label{example control flow}
\end{figure}

\subsubsection{Syntax}
Recall that the Bounded-Length principle requires that the number of steps inside await loops and the number of steps outside awaits are bounded.
We define a tiny concurrent assembly-like language which represents such programs, but not programs that violate the Bounded-Length principle.
The purpose of this language is to allow us to prove things easily, not to conveniently program in it.
Thus instead of a variety of statements, we consider only two statements: await loops (\lstinline[language=Lambda]|await|) and event generating instructions (\lstinline[language=Lambda]|step|).
The syntax and types of our language is summarized in \cref{syntax summary}. The event generating instructions \lstinline[language=Lambda]|step| use a pair of two lambda functions to generate events and modify the thread local state.
Thus the execution of the steps yields a sequence $\sigma^{(t)}$ of thread-local states. We illustrate this at hand of the small example in \cref{example control flow}, which we execute starting with the (arbitrarily picked for demonstrative purposes) thread local state
\[ \sigma^0(r) = \begin{cases} 5 & r = \lstinline[language=Lambda]|r1| \\ 0 & r = \lstinline[language=Lambda]|r2| \end{cases} \]
in which the value of \lstinline[language=Lambda]|r1| is 5 and the value of \lstinline[language=Lambda]|r2| is 1.
The first instruction of the program first evaluates \lstinline[language=Lambda]|lambda sigma. Wrlx(x, sigma(r1))| on $\sigma^0$ to determine which event (if any) should be generated by this instruction. In this case, the generated event is
\[ W^\rlx(\lstinline[language=Lambda]|x|,\sigma^0(\lstinline[language=Lambda]|r1|)) = W^\rlx(\lstinline[language=Lambda]|x|,5) \]
which writes 5 to variable \lstinline[language=Lambda]|x|. Next, the function \lstinline[language=Lambda]|lambda _ _. [ ]| is evaluated on $\sigma^0$ to determine the set of changed registers and their new value in the next thread local state $\sigma^1$. This function takes a second parameter which represents the value returned by the generated event in case the generated event is a read. Because the event in this case is not a read, no value is returned, and the function simply ignores the second parameter. The empty list \lstinline[language=Lambda]|[ ]| indicates that no registers should be updated. Thus
\[ \sigma^1 = \sigma^0 \]
and execution proceeds with the next instruction. The second instruction generates the read event
\[ R^\rlx(y) \]
which reads the value of variable \lstinline[language=Lambda]|y|. Assume for the sake of demonstration that this read (\eg, due to some other thread not shown here) returns the value 8.
Now the function \lstinline[language=Lambda]|lambda sigma v. [r1 -> v]| is evaluated on $\sigma^1$ and $v=8$. The result \lstinline[language=Lambda]|[r1 -> 8]| indicates that the value of \lstinline[language=Lambda]|r1| should be updated to 8, and the next state $\sigma^2$ is computed as
\[ \sigma^2(r) = \begin{cases} 8 & r = \lstinline[language=Lambda]|r1| \\ \sigma^1(r) & \text{o.w.} \end{cases} \]
In this state, the third instruction is executed. Because in $\sigma^2$, the value of \lstinline[language=Lambda]|r1| is not 0, the \lstinline[language=Lambda]|match| goes to the second case, in which no event is generated (indicated by \lstinline[language=Lambda]|F(rlx)|, \ie, a relaxed fence which indicates a NOP). Thus again there is no read result of $v$, and the next state $\sigma^3$ is computed simply as
\[ \sigma^3 = \sigma^2 \]

\begin{figure}[H]
\centering
\begin{tabular}{lp{0.04\textwidth}ll}
\toprule 
normal code && \multicolumn{2}{l}{encoding in toy language} \\
&&  event generators & state transformers
\\\midrule
{\lstset{showlines=true,}
\begin{lstlisting}
r1 = x;
\end{lstlisting}}
&&
{\lstset{showlines=true,}
\begin{lstlisting}[language=Lambda]
lambda _. Rrlx(x)
\end{lstlisting}}
&
{\lstset{showlines=true,}
\begin{lstlisting}[language=Lambda]
lambda _ v. [r1 -> v]
\end{lstlisting}}
\\\midrule
{\lstset{showlines=true,}
\begin{lstlisting}
y = r1+2;
\end{lstlisting}}
&&
{\lstset{showlines=true,}
\begin{lstlisting}[language=Lambda]
lambda sigma. Wrlx(x, sigma(r1)+2)
\end{lstlisting}}
&
{\lstset{showlines=true,}
\begin{lstlisting}[language=Lambda]
lambda _ _. [ ]
\end{lstlisting}}
\\\midrule
{\lstset{showlines=true,}
\begin{lstlisting}
if (x==1){
	y = 2;
	r2 = z;
} else {
	r2 = z;
	y = r2;
}

\end{lstlisting}}
&&
{\lstset{showlines=true,}
\begin{lstlisting}[language=Lambda]
lambda _. Rrlx(x)
lambda sigma. match sigma(r1) 
  with 
  | 1 -> Wrlx(y, 2)
  | _ -> Rrlx(z)
lambda sigma. match sigma(r1)
  with
  | 1. Rrlx(z)
  | _. Wrlx(y, sigma(r2))
\end{lstlisting}}
&
{\lstset{showlines=true,}
\begin{lstlisting}[language=Lambda]
lambda _ v. [r1 -> v]
lambda sigma v. match sigma(r1) 
  with
  | 1. [ ]
  | _. [r1 -> v]
lambda sigma v. match sigma(r1) 
  with
  | 1. [r2 -> v]
  | _. [ ]
\end{lstlisting}}
\\\midrule
\begin{lstlisting}
for (r1 = 0; r1 < 3; r1++)
{
	x = r1;
}
\end{lstlisting} 
&&
{\lstset{showlines=true,}
\begin{lstlisting}[language=Lambda]
lambda sigma. Wrlx(x, sigma(r1))
lambda sigma. Wrlx(x, sigma(r1))
lambda sigma. Wrlx(x, sigma(r1))

\end{lstlisting}}
&
{\lstset{showlines=true,}
\begin{lstlisting}[language=Lambda]
lambda sigma _. [r1 -> sigma(r1)+1]
lambda sigma _. [r1 -> sigma(r1)+1]
lambda sigma _. [r1 -> sigma(r1)+1]

\end{lstlisting}}
\\\bottomrule
\end{tabular}
\caption{Example Encodings of Language Constructs as Event-generator/State-transformer Pairs} \label{fig:example encodings}
\end{figure}

\begin{figure}[h]
\centering
\begin{tabular}{lp{0.05\linewidth}l}
\begin{lstlisting}
do_awaitwhile({
    r1 = y;
},x==1)

\end{lstlisting} 
&&
{\lstset{showlines=true,}
\begin{lstlisting}[language=Lambda]
step(lambda _. Rrlx(y), lambda _ v. [r1 -> v]);
step(lambda _. Rrlx(x), lambda _ v. [r2 -> v]);
await(2, lambda sigma. sigma(r2) == 1)
\end{lstlisting}}
\end{tabular}
\caption{Encoding of Do-Await-While as a Program in our Language} \label{fig:await example}
\end{figure}

Note that each thread's program text is finite, and the only allowed loops are awaits.
Thus programs with infinite behaviors or unbounded executions (which violate the Bounded-Length principle) can not be represented in this language.

Each statement generates up to one event that depends on a thread-local state, and modifies the thread-local state based on the previous state and (in case a read event was generated) the result of the read.
This is encoded using two types of lambda functions: the \emph{event generators} that map the current state an event (possibly $F^\rlx$), \ie, have type
\[ \mathit{State} \rightarrow \mathit{Event} \]
and the \emph{state transformers} that map the current state and possibly a read result to an update to thread-local state representing the new value of all the registers that are changed by the instruction, \ie, have type 
\[ \mathit{State} \times \mathit{Value}? \rightarrow \mathit{Update} \]
Here $T?$ is a so called \emph{option} type, which is similar to the nullable types of C\#: each value $v \in T?$ is either a value of $T$ or $\bot$ (standing for ``none'' or ``null''):
\[ v \in T? \iff v \in T \lor v = \bot \]

Await loops are the only control construct in our language. Apart from awaits, the control of the thread only moves forward, one statement at a time. 
Different conditional branches are implemented through internal logic of the event generating instructions: the state keeps track of the active branch in the code, and the event generator and state transformer functions do a case split on this state. 
Bounded loops have to be unrolled. See \cref{fig:example encodings}.

We formalize the syntax of the language. There is a fixed, finite set of threads $\mathcal T$ and each thread $T \in \mathcal T$ has a finite program text $P_T$ which is a sequence of statements.
We denote the $k$-th statement in the program text of thread $T$ by $P_T(k)$. 
A statement is either an an event generating instruction or a do-await-while statement. 
We assume that the set of registers, values, and locations are all finite.

\paragraph{Event Generating Instruction}
An event generating instruction has the syntax

\begin{minipage}{\linewidth}
\centering	\lstinline[language=Lambda]|step($\egen$,$\strans$)|
\end{minipage}
where 
\[ \egen : \mathit{State} \rightarrow \mathit{Event} \]
is an event generator and 
\[ \strans : \mathit{State} \times \mathit{Value}? \rightarrow  \mathit{Update} \]
is a state transformer.
Note that the event generating instruction is roughly a tuple of two functions $\egen$ and $\strans$. When the statement is executed in a thread-local state $s \in \mathit{State}$, we first evaluate $\egen(\sigma)$ to determine which event is generated. If this event is a read, it returns a value $v$ (defined based on reads-from edges in an execution graph $G$), which is then passed to $\strans$ to compute the new values for updated registers in the update $\strans(\sigma,v)$. The next state is then defined by taking all new values from $\strans(\sigma,v)$ and the remaining (unchanged) values from $s$.

\paragraph{Do-Await-While}
A do-await-while statement has the syntax

\begin{minipage}{\linewidth}
\centering	\lstinline[language=Lambda]|await($n$,$\loopcons$)|
\end{minipage}
where $n \in \Set{0,1,2,\ldots}$ is the number of statements in the loop, and the loop condition $\loopcons : \mathit{State} \rightarrow \Set{0,1}$ is a predicate over states telling us whether we must stay in the loop. 
If this statement is executed in thread-local state $s$, we first evaluate $\loopcons(\sigma)$. 
In case $\loopcons(\sigma)$ evaluates to true, the control jumps back $n$ statements, thus repeating the loop; otherwise it moves one statement ahead, thus exiting the loop.

\paragraph{Syntactic Restriction of Awaits}

We add two syntactic restrictions of awaits: 1) no nesting of awaits and 2) an await which jumps back $n$ statements needs to be at least at position $n$ in the program.
\[ P_T(k) = 	\lstinline[language=Lambda]|await($n$,_)| \quad\to\quad n \le k \ \land \ \forall k' \in [k-n:k).\ P_T(k') \not= \lstinline[language=Lambda]|await(_,_)| \]
These restrictions will allow us to easily identify steps in an iteration of an await as steps in the range $[k-n:k)$

\begin{figure}[t]
\centering
\begin{tikzpicture}
\draw (0,0) rectangle (2,0.5) node[pos=.5] (N) {\lstinline[language=Lambda]|await($n$,_);|};
\draw (0,1.2) rectangle (2,1.7) node[pos=.5] (V) {\lstinline[language=Lambda]|await($v$,_);|};
\draw (0,-0.5) -- (0,4);
\draw (2,-0.5) -- (2,4);
\draw [->] (N.east) -- ++(0.5,0) -- node [right] {$n$} ++(0,3.5) -- ++(-0.5,0);
\draw (0,2.8) rectangle (2,3.3) node[pos=.5] (Q) {};
\draw [->,draw=blue!60!black] (V.east) -- ++(0.4,0) -- node [right,text=blue!60!black] {$v$} ++($(Q.center)-(V.center)$) -- ++(-0.4,0);
\node [left=0.9 of N.west] (t) {$k_G^T(t)$};
\draw [->] (t) -- (N.west);

\draw (t |- Q) node (m) {$k_G^T(t)-m$};
\draw [->] (m) -- (N.west |- Q);

\draw (t |- V) node (v) {$k_G^T(t)-m+v$};
\draw [->] (v) -- (N.west |- V);

\node at ($(Q)!0.45!(V)$) {$\vdots$};
\node at ($(V)!0.45!(N)$) {$\vdots$};
\end{tikzpicture}
\caption{Two overlapping awaits}\label{fig:awaits position}
\end{figure}
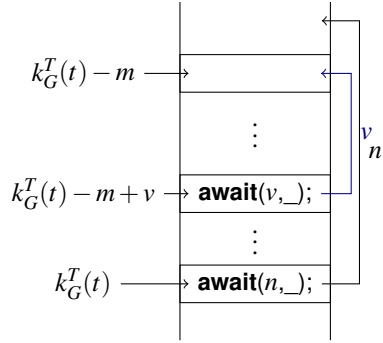

\subsubsection{Semantics}
The semantics of our language consist of two components: an execution graph $G$, which represents the concurrent execution of the events, and local instruction sequences that generate and refer to these events. 
At first glance, the instruction sequences and the event graph interlock like gears: the instruction sequences generate the events in the event graph, \eg, the reads and writes, and the event graph generates the values that are returned by those reads to the instruction sequences.
Of course, the values returned by reads determine which events are generated next by the instruction sequences.
Unfortunately, it is a bit more complex than this: due to weak memory models, the interlock is actually \emph{cyclical}; a write event $w$ of thread $A$ can be generated based on a value returned by a previous read of $A$, which reads from a write event of thread $B$, which is generated based on a previous read of $B$ which reads from $w$.
Thus a simple step-by-step parallel construction of $G$ and the instruction sequences is not possible.

Instead, we follow a more indirect (so called axiomatic) semantics: we take an arbitrary (potentially cyclical) execution graph $G$ and try to justify it ad-hoc by finding local instruction sequences that are consistent with the events in the graph, \ie, 1) every event in $G$ is generated by the instruction sequences, and 2) the instruction sequences use read-results from the graph $G$.
We define this in two steps: at first we ignore for simplicity the consistency predicate $\mathrm{cons}_M(G)$ which states that $G$ is consistent with the memory model, and only check that $G$ can be justified by the program text.
We define this by a predicate $\mathrm{cons}^P(G)$ stating that $G$ is consistent with the program.
Then we make the definition complete by combining the two consistency predicates into a single predicate
\[ \mathrm{cons}^P_M(G) = \mathrm{cons}^P(G) \land \mathrm{cons}_M(G) \]
which states that the execution graph $G$ is consistent with the program $P$ under the weak memory model $M$. We borrow the notation of execution graphs from the work of Vafeiadis et. al \cite{GenMC}, with the minor change that our events include barrier modes.

\paragraph{Defining $\mathrm{cons}^P(G)$}
The semantics of our language is defined with relation to an execution graph $G$, which provides the values individual reads read. With reference to these values, the local program text of each thread can be executed locally.
The graph is consistent with the program text if it contains exactly the events that occur during this local execution.
The local execution of thread $T$ is described by four sequences, which are defined by mutual recursion on the number of executed steps $t \in \Set{0,1,2,\ldots}$: the thread local state $\sigma_G^T(t)$ after $t$ steps, the position of control $k_G^T(t)$ after $t$ steps, the (potential) event generated in the $t$-th step $e_G^T(t)$, and the (potential) read result $v_G^T(t)$ of that event. The definitions of $e_G^T(t)$ and $v_G^T(t)$ are not themselves recursive but refer to $\sigma_G^T(t)$, while themselves being referenced in the definition of $\sigma_G^T({t+1})$.
For this reason, the four sequences do not all have the same length. We denote the number of execution steps of thread $T$ by $N_G^T \in \mathbb N \cup \Set{ \infty }$, where $N_G^T = \infty$ indicates that the thread does not terminate and hence makes infinitely many steps.
The number of steps $N_G^T$ coincides with the length $\lvert e_G^T(t) \rvert$ and $\lvert v_G^T(t) \rvert$ of the sequences $e_G^T(t)$  of events and $v_G^T(t)$ of read results of thread $T$ as every step generates up to one event and returns up to one read result
\[ \lvert e_G^T(t) \rvert = \lvert v_G^T(t) \rvert = N_G^T\]
As usual for these fence post cases, the number of states and positions of control is $N_G^T+1$
\[ \lvert \sigma_G^T(t) \rvert = \lvert k_G^T(t) \rvert = N_G^T+1 \]

\paragraph{Position of Control} The position of control
\[ k_G^T(t) \in \Set{0,1,2,\ldots} \]
is the index of the next statement $P_T(k_G^T(t))$ to be executed by thread $T$ after executing $t$ steps.
All programs start at the first statement, \ie,
\[ k_G^T(0) = 0 \]
and thus the first statement to be executed is the first statement of the program $P_T(k_G^T(0)) = P_T(0)$.
After $t \le N_G^T$ steps, the position of control may leave the program text, \ie, $k_G^T(t)$ may no longer be an index in the sequence $P_T((k))$ of statements of thread $T$
\[ k_G^T(t) \ge \lvert P_T((k)) \rvert \]
In this case, the computation of thread $T$ ends:
\[ k_G^T(t) \ge \lvert P_T((k)) \rvert  \ \to \ N_G^T = t \]
and thus there is no $k_G^T({t+1})$ that needs to be defined. Otherwise, the execution of the $t$-th step changes the position based on the statement $P_T(k_G^T(t))$ executed in that step.
We abbreviate
\[ S_G^T(t) = P_T(k_G^T(t)) \]

If the this statement is an event generating instruction, we always move to the next statement, \ie, 
\[ P_T(k_G^T(t)) = \lstinline[language=Lambda]|step(_,_)| \quad\to\quad k_G^T({t+1}) = k_G^T(t) + 1\]
For a do-await-while \lstinline[language=Lambda]|await($n$,$\loopcons$)|, $k$ is either also incremented by 1 (the loop is exited) or decremented by $n$ (the loop is continued), depending on whether $\loopcons$ evaluates to false or true in the internal state $\sigma_G^T(t)$ of thread $T$ after $t$ steps:
\[ P_T(k_G^T(t)) = \lstinline[language=Lambda]|await($n$,$\loopcons$)| \quad\to\quad  k_G^T({t+1}) = \begin{cases} k_G^T(t) + 1 & \loopcons(\sigma_G^T(t)) = 0 \\ k_G^T(t) - n & \text{o.w.} \end{cases} \]

\paragraph{Event Sequence}
This $t$-th step (with $t < N_G^T$) generates the event
\[ e_G^T(t) \in \mathit{Event} \]
which is either $\egen(\sigma_G^T(t))$  if the statement executed in this step is an event generating instruction \lstinline[language=Lambda]|step($\egen$,_)|
\[ P_T(k_G^T(t)) = \lstinline[language=Lambda]|step($\egen$,_)| \quad\to\quad e_G^T(t) = \egen(\sigma_G^T(t)) \]
or, in case the statement is a do-await-while, a NOP event
\[ P_T(k_G^T(t)) = \lstinline[language=Lambda]|await(_,_)| \quad\to\quad  e_G^T(t) = F^\rlx\]
This event must be in $G$ or $G$ is not consistent with the program.
Recall that $G$ stores the event together with meta data indicating the thread $T$ and the event index $t$ in the program order of $T$.
If no event with this meta data exists in the graph, the graph represents a partial execution of the program. In this case we stop execution of thread $T$ before the event is generated
\[ \langle T,\ t,\ - \rangle \not\in G.\mathrm{E} \quad\to\quad  N_T = t \]
Otherwise, the event and its meta data from the triplet $\langle T,\ t,\ e_G^T(t) \rangle$.
If some event with this meta data exists in the graph, but not this particular event, then the program generated a different event than the one provided by the graph; the graph is inconsistent with the program.
\[ \langle T,\ t,\ e \rangle \in G.\mathrm{E} \ \land \ e \not= e_G^T(t) \quad\to\quad \neg \mathrm{cons}^P(G) \]
Note that $T$ and $t$ already uniquely identify $e_G^T(t)$ in a consistent execution graph. To avoid redundancy we abbreviate
\[ \event{T}{t}{G} = \langle T,\ t,\ e_G^T(t) \rangle  \]

\paragraph{Read Result} The read result
\[ v_G^T(t) \in \mathit{Value}? \]
is the value returned by a read event generated in step $t < N_G^T$ of thread $T$.
If no read event is generated, there is no read result $v_G^T(t)$
\[ e_G^T(t) \not= R^-(-) \quad \to\quad v_G^T(t) = \bot \]
Otherwise,  the read reads from the write $w= G.\mathrm{\rf}(\event{ T}{ t}{ G} \rangle)$ and returns the value $w.\mathrm{val}$ written by that write. 
Note that in the case of a missing \rf-edge, $w$ may be $\bot$ even though $e_G^T(t)$ is a read event. In such cases, we return the read result $\bot$.
However, an instruction that generates a read event usually depends on the read result to compute the next state.
Thus we will define in the next section that the thread-local execution terminates in case $e_G^T(t)$ is a read event but the read result is $\bot$. 
We collect the read result as the value $v_G^T(t)$
\[ e_G^T(t)=R^o(x) \quad\to\quad v_G^T(t) =  \begin{cases} G.\mathrm{\rf}(\event{T}{t}{G}).\mathrm{val} & G.\mathrm{\rf}(\event{T}{t}{G} \rangle) \not= \bot \\ \bot & \text{o.w}  \end{cases} \]

\paragraph{State Sequence}
The thread-local state of thread $T$ after executing $t$ steps
\[ \sigma_G^T(t) \in \mathit{State} \]
contains the values of all thread-local registers. We leave the initial state
\[ \sigma_G^T(0) \]
of thread $T$ uninterpreted. Each step then updates the local state based on the executed statement. Do-await-whiles never change the program state
\[ P_T(k_G^T(t)) = \lstinline[language=Lambda]|await(_,_)| \quad\to\quad  \sigma_G^T({t+1}) = \sigma_G^T(t)\]
For event generating instructions, we consider two cases: the first (regular) case is that the value $v_G^T(t)$ matches the event $e_G^T(t)$ in the sense that $v_G^T(t)$ provides a value if $e_G^T(t)$ is a read.
In this case the event generating instruction \lstinline[language=Lambda]|step(_,$\strans$)| with state transformer $\strans$ updates the state based on $\strans$ under control of the read result $v_G^T(t)$:
\begin{align*} 
P_T(k_G^T(t)) &= \lstinline[language=Lambda]|step(_,$\strans$)| \ \land \ (e_G^T(t) \not= R^-(-) \, \lor \, v_G^T(t) \not=\bot) \\ &\quad\to\quad \sigma_G^T({t+1}) \ =\  \sigma_G^T(t) \ll \strans(\sigma_G^T(t), v_G^T(t)) 
\end{align*}
Here $\ll$ is the update operator which takes all the new (updated) register values from $\strans(\sigma,v)$ and the remaining (unchanged) register values from $s$:
\[ (\sigma \ll \sigma')(r) = \begin{cases}\sigma'(r) & r \in \mathit{Dom}(\sigma') \\ \sigma(r) & \text{o.w.} \end{cases} \]
In the second (irregular) case, step $t$ generated a read event but no read-result was returned. In this case the computation gets stuck. We define that $t$ is the last step. Since steps $0, \ \ldots, \ t$ have been executed the number of steps is thus $N_G^T = t+1$. Thus formally we need to define a state $\sigma_G^T({t+1})$, despite the computation being stuck. We arbitrarily define $\sigma_G^T({t+1}) = \sigma_G^T(t)$.
\[ P_T(k_G^T(t)) = \lstinline[language=Lambda]|step(_,$\strans$)| \ \land \ e_G^T(t) = R^-(-) \  \land \ v_G^T(t) =\bot \quad\to\quad N_G^T = t+1  \ \land \ \sigma_G^T({t+1}) = \sigma_G^T(t) \]

\paragraph{No Superfluous Events}
We have so far checked that every event generated by the program is also in $G$. However, $G$ is only consistent with a program if there is also no event in $G$ that was not generated by the program, \ie, there are no superfluous events.
More precisely, $G$ is consistent with the program $P$ exactly when the set of events $G.\mathrm{E}$ in $G$ is exactly the set of events (plus meta data) generated by $P$:
\[ \mathrm{cons}^P(G) \quad\iff\quad G.\mathrm{E} = \Set{ \event{T}{ t}{G} | T \in \mathcal T,\ t < N_G^T } \]

\paragraph{Register Read-From}
Recall that the Bounded-Effect principle states that there must not be a visible side effect of a failed iteration of an await. 
Between threads, the only potentially visible side effects are the generated stores.
Within a thread, updates to registers can also be visible, provided these registers are not overwritten in the mean-time.
We define a register read-from relation within events of a single thread
\[ G.\mathrm{rrf} \subseteq G.\mathrm{po} \]
which holds between events $e$ and $e'$ exactly when the statement that generated $e'$ ``depends'' on any registers updated by the statement that generated $e$ which have not been overwritten in the meantime.
More precisely, we put a register read-from edge in the graph between the $t$-th and $u$-th steps (with $u\ge t$) if any function in the statement $P_T(k_G^T(u))(G)$ executed by the $u$-th step depends on the visible output of the $t$-th to the $u$-th step:
\[ \event{ T }{ t }{ G } \overset{\mathrm{rrf}}\longrightarrow \event{T}{ u}{ G } \quad\iff\quad u \ge t \, \land \, \exists f \in F(P_T(k_G^T(u))(G)). \ \mathit{depends\hy on}(f, \mathit{vis}_G^T(t,u)) \]
To define what it means to ``depend'' on these registers, we look at the functions $\egen$, $\strans$, $\loopcons$ of the statement and see whether the registers can affect the functions.
We collect the functions in statement $S$ in a set $F(S)$ defined as follows
\[ F(S) = \begin{cases} \Set{ \egen, \strans } & S = \lstinline[language=Lambda]|step($\egen$, $\strans$)| \\ \Set{ \loopcons } & S = \lstinline[language=Lambda]|await(_, $\loopcons$)| \end{cases} \]
A function $f \in F(S)$ \emph{depends on} a set of registers $R \subseteq \mathit{Register}$ if there are two states which only differ on registers in $R$ on which $f$ produces different results\footnote{For the sake of simplicity we use here curried notation for $f = \strans$, \ie, $\strans(\sigma) \not= \strans(\sigma')$ iff there is a $v$ such that $\strans(\sigma,v) \not= \strans(\sigma',v)$.}
\[ \mathit{depends\hy on}(f,R) \quad\iff\quad \exists \sigma,\, \sigma' \in \mathit{State}. \ (\forall r \not\in R.\ \sigma(r) = \sigma'(r)) \ \land \ f(\sigma) \not= f(\sigma') \]
To unify notation we define an update $\delta_G^T(t)$ for each step $t$ by
\[ \strans_G^T(t) = \begin{cases} \strans(\sigma_G^T(t), v_G^T(t)) & P_T(k_G^T(t)) = \lstinline[language=Lambda]|step(_, $\strans$)| \land (e_G^T(t) \not= R^-(-) \lor v_G^T(t) \not= \bot) \\ \emptyset & \text{o.w.} \end{cases} \]
where $\emptyset$ is the empty update (no registers changed). A straightforward induction shows:
\begin{lemma} \label{lem:strans sequence}
\[ \sigma_G^T(t+1) = \sigma_G^T(t) \ll \delta_G^T(t) \]
\end{lemma}
We define the visible output of the $t$-th step the the $u$-th step to be the set of registers that are updated by the $t$-th step but not by the steps before $u$
\[ \mathit{vis}_G^T(t,u) \ = \ \mathit{Dom}(\strans_G^T(t)) \setminus \bigcup_{u' \in (t:u)} \mathit{Dom}(\strans_G^T(u'))  \]

\paragraph{Iterations of Await}
In this section we define the steps that constitute iterations of await.
These are steps that execute statements with numbers $k' \in [k-n:k]$ where statement number $k$ is a do-await-while statement that jumps back $n$ steps. We enumerate these steps $k$ which are endpoints of await iterations.
\begin{align*}
\mathit{end}_G^T(0) &= \min \Set{ t |  P_T(k_G^T(t)) = \lstinline[language=Lambda]|await(_, _)| }\\ 
\mathit{end}_G^T(q+1) &= \min \Set{t > \mathit{end}_G^T(q) | P_T(k_G^T(t)) = \lstinline[language=Lambda]|await(_, _)| }
\end{align*} 
We denote the length of such an iteration, \ie, the number $n$ of steps jumped back by the do-await-while-statement, by
\[ \mathit{len}_G^T(q) =n \quad\text{where}\quad P_T(k_G^T(\mathit{end}_G^T(q))) = \lstinline[language=Lambda]|await($n$, _)| \]
The start point ($k-n$) of the $q$-th iteration is defined by
\[ \mathit{start}_G^T(q) \ = \ \mathit{end}_G^T(q) - \mathit{len}_G^T(q) \] 

We show that the intervals $[\mathit{start}_G^T(q):\mathit{end}_G^T(q)]$ for all $q$ do not overlap.
For this it suffices to show that only the last step in such an interval executes a do-await-while statement.
\begin{lemma}\label{lem:step forward iteration}
\[ t \in [\mathit{start}_G^T(q):\mathit{end}_G^T(q)) \quad\to\quad P_G^T(k_G^T(t)) \not= \lstinline[language=Lambda]|await(_, _)| \]
\end{lemma}
\begin{proof}
Assume for the sake of contradiction that step $t$ executes a do-await-while.
W.l.o.g. $t$ is the last step to do so
\[ t = \max \Set{ u \in [\mathit{start}_G^T(q):\mathit{end}_G^T(q)) | P_G^T(k_G^T(u)) = \lstinline[language=Lambda]|await(_, _)| } \]
Since the remaining steps between $t$ and $\mathit{end}_G^T(q)$ are not do-await-while statements, they move the position of control ahead one statement at a time. Thus the difference in the positions of control is equal to the difference in the ste number.
\[ k_G^T(\mathit{end}_G^T(q))  -  k_G^T(t+1) \ =\  \mathit{end}_G^T(q)-(t+1) \]
Furthermore, the loop condition of step $t$ can not be satisfied since otherwise control would jump back and then have to cross position $k_G^T(t)$ a second time; but this contradicts the assumption that no more do-await-while statements are executed.
Thus
\[ k_G^T(t+1) = k_G^T(t) \]
and we conclude with simple arithmetic
\[ k_G^T(t) \ =\  k_G^T(\mathit{end}_G^T(q))  -  (\mathit{end}_G^T(q) - t) \]
Since $t$ is in the interval $[\mathit{start}_G^T(q):\mathit{end}_G^T(q))$ which by definition has length $\mathit{len}_G^T(q)$, we conclude first
\[ \mathit{end}_G^T(q) - t < \mathit{len}_G^T(q) \]
and then that the await must be positioned in the last $\mathit{len}_G^T(q)$ statements before the $q$-th executed await
\[ k_G^T(t) \in [k_G^T(\mathit{end})-\mathit{len}_G^T(q):k_G^T(\mathit{end})) \]
but this contradicts the assumption that awaits are never nested.
\end{proof}

We conclude that interval number $q+1$ starts after interval number $q$ ends. Monotonicity then immediately implies that the intervals are pairwise disjoint.
\begin{lemma}\label{lem:iterations disjoint}
\[ \mathit{end}_G^T(q)< \mathit{start}_G^T(q+1) \]
\end{lemma}
\begin{proof}
By definition, step $\mathit{end}_G^T(q)$ executes a do-await-while statement
\[  P_T(k_G^T(\mathit{end}_G^T(q))) = \lstinline[language=Lambda]|await(_, _)| \]
By contraposition of \cref{lem:step forward iteration}, the step is not in interval number $q+1$
\[ \mathit{end}_G^T(q) \not\in [\mathit{start}_G^T(q+1):\mathit{end}_G^T(q+1)) \]
Due to monotonicity interval number $q$ ends before interval number $q+1$
\[ \mathit{end}_G^T(q) < \mathit{end}_G^T(q+1) \]
and the claim follows
\[ \mathit{end}_G^T(q) < \mathit{start}_G^T(q+1) \]
\end{proof}

Iteration $q$ is failed if the loop condition in step $\mathit{end}_G^T(q)$ evaluates to $1$:
\[ \mathit{fail}_G^T(q) \ \iff \ \loopcons(\sigma_G^T(\mathit{end}_G^T(q)))=1 \quad\text{where}\quad  P_T(k_G^T(\mathit{end}_G^T(q))) = \lstinline[language=Lambda]|await(_, $\loopcons$)| \]
We plan to cut all failed iterations from the graph which result in a wasteful graph. These are failed iterations $q$ in which the next failed iteration reads from exactly the same stores.
We define this precisely through a predicate $\mathit{WI}_G^T(q)$:
\begin{align*}
\mathit{WI}_G^T(q) \ \iff \  \mathit{fail}_G^T(q) \ \land \ \forall m \le \mathit{len}_G^T(q).&\ e_G^T(\mathit{start}_G^T(q)+m) = \mathit{R}^-(-) 
  \\ & \ \to \ G.\mathrm{rf}(\event{ T}{ \mathit{start}_G^T(q) + m}{G}) = G.\mathrm{rf}(\event{ T}{ \mathit{start}_G^T(q+1) + m}{G})
\end{align*}
Iff a graph has any such iterations, we say that is wasteful. Formally:
\[ \mathrm{W}(G) \ \iff \ \forall T,\,q.\ \neg \mathit{WI}_G^T(q) \]

\subsection{Formalizing the Bounded-Effect principle}\label{sec:vegas}
The Bounded-Effect principle is now easily formalized. We define $\mathit{BE}(G)$ to hold if in graph $G$ no register reads-from arrow leaves a failed iteration of an await.
For the sake of simplicity we fully forbid generating write events in failed iterations of awaits. This is unlikely to be a practical restriction\footnote{It is possible to construct theoretical examples where this restriction changes the logic of the code, but we are not aware of any practical examples. In any case, the restriction can be lifted with considerable elbow grease.}; as reads-from edges from such writes that leave the iteration are anyways forbidden, this directly only affects loop-internal write-read pairs, which can be simulated using registers.
\begin{definition}
Graph $G$ satisfies the Bounded-Effect principle, \ie,
\( \mathit{BE}(G) \), iff for all threads $T$, failed await iterations $q$ with $\mathit{fail}_G^T(q)$, and step numbers $t \in [\mathit{start}_G^T(q):\mathit{end}_G^T(q))$ we have: 
\begin{enumerate}
\item no write event is generated in step $t$ of thread $T$
\[ e_G^T(t) \not= W^-(-,-) \]
\item if the event generated in step $t$ of thread $T$ is register read-from by the event generated in step $u$ of thread $T$, then $u$ is in the same failed iteration $q$
\[ \event{ T}{t}{G} \overset{\mathrm{rrf}}\longrightarrow \event{T}{u}{G} \quad\to\quad u \in [\mathit{start}_G^T(q):\mathit{end}_G^T(q)) \]
\end{enumerate}
\end{definition}
In the remainder of this text, we assume: all graphs that are consistent with the memory model and with $P$ satisfy the Bounded-Effect principle
\begin{equation} \forall G.\ \mathrm{cons}_M^P(G) \quad\to\quad \mathit{BE}(G) \end{equation}

\subsection{Proving the main theorem}\label{sec:main}
We define the sets $\mathbb G^F$ and $\mathbb G^\infty$ of execution graphs with a finite resp. infinite number of failed await iterations by
\begin{align*}
 \mathbb G^F &= \Set{ G | \mathrm{cons}_M^P(G) \ \land \ \forall T.\ \exists q.\ \forall q' \ge q. \ \neg \mathit{fail}_G^T(q') } 
 \\
 \mathbb G^\infty &= \Set{ G | \mathrm{cons}_M^P(G) \ \land \ \exists T.\ \forall q.\ \exists q' \ge q. \ \mathit{fail}_G^T(q') } 
 \\ &= \Set{ G | \mathrm{cons}_M^P(G)} \setminus \mathbb G^F
\end{align*}
Await-termination holds if $\mathbb G^\infty$ is empty
\[ \mathit{AT} \ \iff \ \mathbb G^\infty = \emptyset \]

We now show a series of lemmas that lead us to the main theorem:
\begin{reptheorem}{thm:amc} \leavevmode
\begin{enumerate}
\item The set $\finitegraphs{} = \Set{G \in \mathbb G^F | \neg \mathrm{W}(G)}$ of non-wasteful execution graphs is finite, and every $G \in \infinitegraphs{} = \Set{ G \in \finitegraphs{} | \mathrm{stagnant}(G) }$ which is stagnant is finite.
\item if an error event $E$ exists in a graph $G \in \mathbb G^F$, it also exists in a graph $G' \in \finitegraphs$
\item every graph $G \in \mathbb G^\infty$ can be cut to a graph $G' \in \infinitegraphs{}$
\item every graph $G \in \infinitegraphs{}$ can be extended to a graph $G' \in \mathbb G^\infty$
\end{enumerate}
\end{reptheorem}

In the proofs we ignore the memory model consistency. This can be easily proven on a case by case basis (\ie, for concrete $M$) but a generic proof must rely on certain abstract features of the memory model which are hard to identify in generic fashion. We leave a generic proof with appropriate conditions as future work.

We first show that after a failed iteration, we return to the startof the await and immediately repeat the iteration (possibly with a different outcome).
\begin{lemma} \label{lem:failed repeat}
\begin{align*}
 \mathit{fail}_G^T(q) \quad\to\quad k_G^T(\mathit{end}_G^T(q)+1) &= k_G^T(\mathit{start}_G^T(q)) 
 \\ {}\land \mathit{start}_G^T(q+1) &=  \mathit{end}_G^T(q)+1
 \\ {}\land \mathit{end}_G^T(q+1) &=  \mathit{end}_G^T(q)+1+n
\end{align*}
\end{lemma}
\begin{proof}
By definition the loop condition $\loopcons$ in step $\mathit{end}_G^T(q)$ is satisfied
\[ P_T(k_G^T(\mathit{end}_G^T(q))) = \lstinline[language=Lambda]|await($n$,$\loopcons$)| \ \land \ \loopcons(\sigma_G^T(\mathit{end}_G^T(q))) = 1 \]
Thus by the semantics of the language, control jumps back $n = \mathit{len}_G^T(q)$ steps
\[ k_G^T(\mathit{end}_G^T(q)+1) = k_G^T(\mathit{end}_G^T(q)) - n \]
By \cref{lem:step forward iteration} the previous $n$ steps all do not execute do-await-while statements and thus moved control forward linearly; the first part of the claim follows
\[ k_G^T(\mathit{end}_G^T(q)) - n = k_G^T(\mathit{end}_G^T(q)-n) = k_G^T(\mathit{start}_G^T(q))\]
We next prove the third part of the claim. Observe that the next $n$ statements are exactly the same (non-await) statements, and thus after an additional $n$ steps we have again
\[ k_G^T(\mathit{end}_G^T(q)+1+n) = k_G^T(\mathit{end}_G^T(q)+1)+n = k_G^T(\mathit{end}_G^T(q)) \]
which is a do-await-while. Hence by the definition of $\mathit{end}_G^T$ we have
\[ \mathit{end}_G^T(q+1) = \mathit{end}_G^T(q)+1+n \] 
which is the third part of the claim. For the remaining second part of the claim, note that the do-await-while still jumps back $n$ statements.
Thus by definition of $\mathit{start}_G^T$ and $\mathit{len}_G^T$ we have
\[ \mathit{start}_G^T(q+1) = \mathit{end}_G^T(q+1)  - \mathit{len}_G^T(q+1) = \mathit{end}_G^T(q)+1+n - n = \mathit{end}_G^T(q)+1 \]
which is the claim.
\end{proof}
If the iteration is wasteful, then the next iteration repeats exactly the same events, positions of control, and states (except that registers that were modified in the first iteration may have changed values)
\begin{lemma} \label{lem:failed progress}
Let $q$ be the index of an iteration that is wasteful, and $m$ be the number of steps taken by the thread inside the iteration (without leaving it)
\[ \mathit{WI}_G^T(q) \ \land \ m \le \mathit{len}_G^T(q) \]
Then all of the following hold:
\begin{enumerate}
\item the same events are generated in iterations $q$ and $q+1$ after $m$ steps \label{lem:failed repeat:same e}
\[ e_G^T(\mathit{start}_G^T(q)+m) = e_G^T(\mathit{start}_G^T(q+1)+m) \]
\item the same values are observed \ldots \label{lem:failed repeat:same v}
\[ v_G^T(\mathit{start}_G^T(q)+m) = v_G^T(\mathit{start}_G^T(q+1)+m) \]
\item the same position of control is reached \ldots \label{lem:failed repeat:same k}
\[ k_G^T(\mathit{start}_G^T(q)+m) = k_G^T(\mathit{start}_G^T(q+1)+m) \]
\item if the value of register $r$ is not the same after $m$ steps in the two iterations \label{lem:failed repeat:same state}
\[ \sigma_G^T(\mathit{start}_G^T(q)+m)(r) = \sigma_G^T(\mathit{start}_G^T(q+1)+m)(r) \]
then $r$ must be an output of one of the steps $u$ of the failed iteration $q$ which is still visible after $m$ steps
\[ \exists u \in [\mathit{start}_G^T(q):\mathit{end}_G^T(q)].\ r \in \mathit{vis}(u,\mathit{start}_G^T(q+1)+m) \]
\end{enumerate}
\end{lemma}
\begin{proof}
We first show that claims \ref{lem:failed repeat:same e} and \ref{lem:failed repeat:same v} follow from claims \ref{lem:failed repeat:same k} and \ref{lem:failed repeat:same state}.
We know from claim \ref{lem:failed repeat:same k} that the position of control is the same. Thus also the executed statement is the same
\[ P_T(k_G^T(\mathit{start}_G^T(q)+m)) = P_T(k_G^T(\mathit{start}_G^T(q+1)+m)) \]
We split cases on the type of statement executed by the steps; in the case it is a do-await-while we are done as no event or read-result is generated. In the other case, we have
\[ P_T(k_G^T(\mathit{start}_G^T(q)+m)) = P_T(k_G^T(\mathit{start}_G^T(q+1)+m)) = \lstinline[language=Lambda]|step($\egen$, _)| \]
From the Bounded-Effect principle we know that there is no register reads-from from a step $u \in [\mathit{start}_G^T(q):\mathit{end}_G^T(q)]$ of the failed iteration $q$ to step $\mathit{start}_G^T(q+1)+m$ which is outside that iteration (\cref{lem:iterations disjoint})
\[ \event{T}{u}{G} \centernot{\overset{\mathrm{rrf}}\longrightarrow} \event{T}{\mathit{start}_G^T(q+1)+m}{G} \]
and thus in particular $\egen$ does not depend on the visible outputs of step $u$ to that step
\[ \neg \mathit{depends\hy on}(\egen, \mathit{vis}(u, \mathit{start}_G^T(q+1)+m) \]
From claim \ref{lem:failed repeat:same state} we know that the only differences between the two states are on registers which are such visible outputs. Thus with the definition of $\mathit{depends\hy on}$ we know
\[ \egen(\sigma_G^T(\mathit{start}_G^T(q)+m)) = \egen(\sigma_G^T(\mathit{start}_G^T(q+1)+m)) \]
and thus the generated events are the same, which is claim \ref{lem:failed repeat:same e}
\[ e_G^T(\mathit{start}_G^T(q)+m) = e_G^T(\mathit{start}_G^T(q+1)+m) \]
For claim \ref{lem:failed repeat:same v} we only consider read events
\[ e_G^T(\mathit{start}_G^T(q)+m) = R^-(-) \]
We have by assumption that iteration $q$ is wasteful, thus the two events read from the same store
\[ G.\mathrm{rf}(\event{T}{ \mathit{start}_G^T(q)+m}{G}) = G.\mathrm{rf}(\event{T}{ \mathit{start}_G^T(q+1)+m}{G}) \]
and thus read the same value. Claim \ref{lem:failed repeat:same v} immediately follows.

Claims \ref{lem:failed repeat:same state} and \ref{lem:failed repeat:same k} are shown by joint induction on $m$. In the base case $m=0$ claim \ref{lem:failed repeat:same k} follows immediately from \cref{lem:failed repeat}
\[ k_G^T(\mathit{start}_G^T(q+1)) = k_G^T(\mathit{end}_G^T(q)+1) = k_G^T(\mathit{start}_G^T(q)) \]
For claim \ref{lem:failed repeat:same state} observe that every register $r$ which differs before and after iteration $q$
\[ \sigma_G^T(\mathit{start}_G^T(q))(r) \not= \sigma_G^T(\mathit{end}_G^T(q)+1)(r)\]
must have been modified by some $u \in [\mathit{start}_G^T(q):\mathit{end}_G^T(q)]$ in that iteration
\[ r \in \mathit{Dom}(\strans_G^T(u)) \]
W.l.o.g. $u$ is the last such write, in which case the effect is still visible to step $\mathit{end}_G^T(q)+1$
\[ r \in \mathit{vis}_G^T(u,\mathit{end}_G^T(q)+1) \]
and the claim follows as by \cref{lem:failed repeat}, step $\mathit{end}_G^T(q)+1$ is the start of iteration $q+1$
\[ \mathit{end}_G^T(q)+1 = \mathit{start}_G^T(q+1)\]

In the induction step $m \to m+1$, we know by the induction hypothesis that the position of control is the same after $m$ steps in the respective iteration and that the states are the same (modulo visible outputs). 
As we have shown before, this implies that the read result is also the same
\[ v_G^T(\mathit{start}_G^T(q)+m) = v_G^T(\mathit{start}_G^T(q+1)+m) \]
Analogous to the proof that $\egen$ produces the same event due to the Bounded-Effect principle, one can also conclude that the new position of control must be the same (\ie, claim \ref{lem:failed repeat:same k} holds)
\[ k_G^T(\mathit{start}_G^T(q)+m+1) = k_G^T(\mathit{start}_G^T(q+1)+m+1) \]
and that the state transformer is the same (which apart from the state also depends on the read result)
\[ \strans_G^T(\mathit{start}_G^T(q)+m+1) =\strans_G^T(\mathit{start}_G^T(q+1)+m+1) \]
Assume for the sake of showing the only remaining claim that register $r$ has a different value in the new states
\[ \sigma_G^T(\mathit{start}_G^T(q)+m+1)(r) \not= \sigma_G^T(\mathit{start}_G^T(q+1)+m+1)(r) \]
By \cref{lem:strans sequence} the functions $\strans_G^T$ determine the state change; since the states were updated in the same way, $r$ can not have been updated
\[ r \not\in \mathit{Dom}(\strans_G^T(\mathit{start}_G^T(q)+m+1)) \]
and the difference was already present before the step
\[ \sigma_G^T(\mathit{start}_G^T(q)+m)(r) \not= \sigma_G^T(\mathit{start}_G^T(q+1)+m)(r) \]
By the induction hypothesis this implies that $r$ was a visible output of some step $u$ from iteration $q$ to the previous step
\[ \exists  u \in [\mathit{start}_G^T(q):\mathit{end}_G^T(q)] . \  r \in \mathit{vis}_G^T(u,\mathit{start}_G^T(q+1)+m) \]
and since it is not updated in this step, it is still visible to the next step
\[ \exists  u \in [\mathit{start}_G^T(q):\mathit{end}_G^T(q)] . \  r \in \mathit{vis}_G^T(u,\mathit{start}_G^T(q+1)+m+1) \]
which is the claim.
\end{proof}

Next we show that such a failed iteration can be safely removed. For this we define a deletion operation $G-(T,q)$ which deletes the $q$-th iteration of thread $T$ from the graph. We only define this in case the $q$-th iteration failed. 
In this case by the Bounded-Effect principle there are no write events that are deleted, and thus we do not have to pay attention to deleting writes that are referenced by other reads.
Events of other threads are not affected at all. Neither are events generated before the start of the deleted iteration. For events started after the deleted iteration we simply reduce the event index by the number of events in the deleted iteration ($\mathit{len}_G^T(q)+1$)
\begin{align*}
(G-(T,q)).\mathrm{E}_U &= G.\mathrm{E}_U \quad \text{if} \quad U \not= T
\\
(G-(T,q)).\mathrm{E}_T &= \Set{ \langle T,\ t,\ e \rangle \in G.\mathrm{E}_T | t < \mathit{start}_G^T(q) } \cup \Set{ \langle T,\ t-(\mathit{len}_G^T(q)+1),\ e \rangle \in G.\mathrm{E}_T | t > \mathit{end}_G^T(q) }
\end{align*}
This can also be defined by means of a partial, invertible renaming function
\[ r : G.\mathrm{E} \rightharpoonup (G-(T,q)).\mathrm{E} \]
which maps each non-deleted event to its renamed event in $G-(T,q)$:
\[ r(\langle U,\ t,\ e \rangle ) = \begin{cases} \langle U,\ t,\ e \rangle  & U \not= T \lor t < \mathit{start}_G^T(q) \\ \langle U,\ t-(\mathit{len}_G^T(q)+1),\ e \rangle & U=T \land t > \mathit{end}_G^T(q) \end{cases} \]
We have:
\[ (G-(T,q)).\mathrm{E} = r(G.\mathrm{E}) \]
For the reads-from relationship, we simply re-map the edges between the renamed events:
\[ (G-(T,q)).\mathrm{rf}(e) = r(G.\mathrm{rf}(r^{-1}(e))) \]

We show that this graph still is consistent with the program.
\begin{lemma} \label{lem:cons preserved}
\[ \mathrm{cons}^P(G) \ \land \ \mathit{fail}_G^T(q) \quad\to\quad \mathrm{cons}^P(G-(T,q)) \]
\end{lemma}
\begin{proof}
For threads other than $T$ there is nothing to show as the event sequences and \rf-edges are fully unchanged.
For $T$, we focus on the steps after the deletion, which may be affected by the change in registers.
We will show: any changes to registers after step $\mathit{start}_G^T(q)$ were visible changes of a register by one of the deleted steps. Other things have not changed.
Thus any dependence on these changed registers would imply a register-read-from relation in the original graph $G$, which is forbidden by the Bounded-Effect principle.
For the sake of brevity we define
\[ G' = G-(T,q) \]
\begin{lemma} \label{lem:deleted:aux}
If $t$ is a step behind the deleted parts in the new graph,
\[ t \ge \mathit{start}_G^T(q) \]
then both of the following hold:
\begin{enumerate}
\item position of control in, event generated by, and read result seen in step $t$ are unaffected by the deletion (relative to the original values of step $t+\mathit{len}_G^T(q)+1$)
\[ k_{G'}^T(t) = k_G^T(t+\mathit{len}_G^T(q)+1) \ \land \ e_{G'}^T(t) = e_G^T(t+\mathit{len}_G^T(q)+1) \ \land \ v_{G'}^T(t) = v_G^T(t+\mathit{len}_G^T(q)+1) \]
\item if $r$ is a register whose value was changed by the deletion
\[ \sigma_{G'}^T(t)(r) \not= \sigma_G^T(t+\mathit{len}_G^T(q)+1)(r) \]
then there is a step $u$ in the original graph which still has a visible effect on $r$ 
\[ \exists u \in [\mathrm{start}_G^T(q):\mathrm{end}_G^T(q)].\ r \in \mathit{vis}_G^T(u,t+\mathit{len}_G^T(q)+1) \]
\end{enumerate}
\end{lemma}
\begin{proof}
The proof is analogous to the proof of \cref{lem:failed progress} and omitted.
\end{proof}
Now to show that $\mathrm{cons}^P$ is preserved we simply consider the sets of events of the individual threads $U \in \mathcal T$ and show that they are not affected:
\begin{equation} \forall U. \ G'.\mathrm{E}_U = \Set{ \event{U}{t}{G'} | t < N_{G'}^U  } \label{lem:deleted:same events} \end{equation}
We split cases on $U \not= T$. For threads $U\not=T$ other than $T$, nothing has changed and the claim follows from the consistency of $G$
\begin{align}
 G'.\mathrm{E}_U &= G'.\mathrm{E}_U \nonumber
 \\ &= \Set{ \event{U}{t}{G} | t < N_{G}^U  } \nonumber
\\ &= \Set{ \event{U}{t}{G'} | t < N_{G'}^U  } \label{lem:deleted:U unchanged}
\end{align}
For thread $T$, we split the set into those events generated before the cut (which have not changed)
\begin{equation} \Set{ \event{T}{t}{G} | t < \mathit{start}_{G}^T(q) } = \Set{ \event{T}{t}{G'} | t < \mathit{start}_{G}^T(q) } \label{lem:deleted:prev unchanged} \end{equation}
and the events generated after the cut, for which \cref{lem:deleted:aux} shows that only indices have changed:
\begin{align}
&\hphantom{{}={}} \Set{ \langle T, \ t-(\mathit{len}_G^T(q)+1), \ e_{G}^T(t) \rangle | t > \mathit{end}_G^T(q) \land t < N_{G}^T } \nonumber
\\ &= \Set{ \langle T, \ t-(\mathit{len}_G^T(q)+1), \ e_{G}^T(t) \rangle | t > \mathit{start}_G^T(q)+(\mathit{len}_G^T(q)+1) \land t < N_{G}^T } \nonumber
\\ &= \Set{ \langle T, \ t-(\mathit{len}_G^T(q)+1), \ e_{G}^T(t) \rangle | t-(\mathit{len}_G^T(q)+1) > \mathit{start}_G^T(q) \land t < N_{G}^T } \nonumber
\\ &= \Set{ \langle T, \ t, \ e_{G}^T(t+(\mathit{len}_G^T(q)+1)) \rangle | t > \mathit{start}_G^T(q) \land t+(\mathit{len}_G^T(q)+1) < N_{G}^T } \nonumber & \text{rebase $t$}
\\ &= \Set{ \langle T, \ t, \ e_{G'}^T(t) \rangle | t > \mathit{start}_G^T(q) \land t < N_{G}^T-(\mathit{len}_G^T(q)+1) } \nonumber & \text{L \ref{lem:deleted:aux}}
\\ &= \Set{ \event{T}{t}{G'} | t > \mathit{start}_G^T(q) \land t < N_{G'}^T } \label{lem:deleted:succ unchanged}
\end{align}
Jointly with \cref{lem:deleted:prev unchanged} this proves \cref{lem:deleted:same events} for $U:=T$:
\begin{align*}
 G'.\mathrm{E}_T &= r(G.\mathrm{E}_T)
 \\ &= \Set{ \event{T}{t}{G} | t < \mathit{start}_{G}^T(q) } \cup \Set{ \langle T, \ t-(\mathit{len}_G^T(q)+1), \ e_{G}^T(t) \rangle | t > \mathit{end}_G^T(q) \land t < N_{G}^T }
 \\ &= \Set{ \event{T}{t}{G'} | t < \mathit{start}_{G}^T(q) } \cup \Set{ \event{T}{t}{G'} \rangle | t \ge \mathit{start}_G^T(q) \land t < N_{G'}^T} & \text{E \eqref{lem:deleted:prev unchanged}, \eqref{lem:deleted:succ unchanged}}
 \\ &= \Set{ \event{T}{t}{G'} | t < N_{G'}^T }
\end{align*}
Together with \cref{lem:deleted:U unchanged} this shows \cref{lem:deleted:same events}. By \cref{lem:deleted:same events} we conclude that the union over all threads $U$ of events in $G'.\mathrm{E}_U$ is equal to the events generated by all threads:
\[ \bigcup_{U \in \mathcal T} G'.\mathrm{E}_U \ = \ \bigcup_{U \in \mathcal T}  \Set{ \event{U}{ t}{G'} \rangle | t < N_{G'}^T } \]
It immediately follows that the set of all events in $G'$ is equal to the set of events generated by all threads
\[ G'.\mathrm{E} = \Set{ \event{U}{ t}{G'} \rangle | U \in \mathcal T,\ t < N_{G'}^T } \]
which is the definition of $\mathrm{cons}^P(G')$, \ie, the claim.
\end{proof}

Next we iteratively eliminate all wasteful iterations of awaits. This takes us from any graph $G \in \mathbb G^F$ to a graph $G' \in \finitegraphs{}$ but preserves at least some error events.
\begin{lemma} \label{lem:cut finite}
\[ G \in \mathbb G^F \ \land \ \langle -, -, E \rangle \in G.\mathrm{E} \quad\to\quad \exists G' \in \finitegraphs{}. \ \langle -, -, E \rangle \in G'.\mathrm{E} \]
\end{lemma}
\begin{proof}
We construct a series \( G^{(i)} \) of graphs in which $i$ wasteful iterations have been eliminated from $G$. Since $G \in \mathbb G^F$ there are only finitely many failed iterations we need to remove, thus the sequence is finite.
We begin graph $G$ in which $0$ iterations have been deleted
\[ G^0 = G \]
and then progressively cut one additional arbitrary failed iteration at a time
\[ G^{i+1} = G^i-\epsilon \Set{ (T,q)  | \mathit{WI}_G^T(q) } \]
The last graph $G^{I-1}$ in the sequence with index $I = \lvert G^{(i)} \rvert$ by definition does not have any wasteful iterations
\[ \nexists T,\,q.\ \mathit{WI}_{G^{I-1}}^T(q)\]
and thus is not wasteful
\[ \neg \mathrm{W}(G^{I-1}) \]
By repeated application of \cref{lem:cons preserved} we conclude that all the graphs in the sequence, including $G^{I-1}$, are consistent with the program
\[ \mathrm{cons}^P(G^{I-1}) \]
and thus $G':=G^{I-1}$ is in $\finitegraphs{}$
\[ G^{I-1} \in \finitegraphs{} \]

It now suffices to show that the error event is preserved (although possibly generated by a different step). 
Assume
\[ \langle -, -, E \rangle \in G.\mathrm{E} \]
By definition we only delete events in wasteful iterations. 
By \cref{lem:failed progress} every event in such an iteration is repeated in the next iteration. The events outside the deleted iteration are maintained (cf. \cref{lem:deleted:aux}).
We conclude: the event $E$ is generated in every graph in the sequence, in particular also in the last one
\[ \langle -, -, E \rangle \in G^{I-1}.\mathrm{E} \]
which is the claim.
\end{proof}

This shows that no bugs are missed by AMC. Our next goal is to show that AMC can terminate.
We first show that the number of writes generated in a graph is bounded.
\begin{lemma} \label{lem:bounded W}
\[ \exists b. \ \forall G. \ \mathrm{cons}_M^P(G) \quad\to\quad \lvert \Set{ (T, t) | e_G^T(t) = W^-(-,-), T \in \mathcal T, t < N_G^T} \rvert \le b \]
\end{lemma}
\begin{proof}
The bound is equal to the sum of the program lengths of each thread 
\[ b:= \sum_{T \in \mathcal T} \lvert P_T \rvert \]
The reason for this is that threads only repeat statements in case they fail a loop iteration; but these failed loop iterations by the Bounded-Effect principle do not produce writes.
We show: if step $t$ of thread $T$ generates a write event, statement $k_G^T(t)$ is never executed again.
\begin{equation} e_G^T(t) = W^-(-,-) \quad\to\quad \forall u >t.\ k_G^T(u) > k_G^T(t) \label{lem:bound:T}\end{equation}
Let $u$ w.l.o.g. be the first step after $t$ in which the position of control is at $k_G^T(t)$ or before
\[ k_G^T(u) \le k_G^T(t)  \ \land \ k_G^T(u-1) > k_G^T(t) \]
By the semantics of the language, step $u-1$ must have been a failed await iteration
\[ P_T(k_G^T(u-1)) = \lstinline[language=Lambda]|await($n$,$\loopcons$)| \ \land \ \loopcons(\sigma_G^T(u-1)) = 1 \]
Since there are no nested awaits, all lines between $k_G^T(u)$ and $k_G^T(u-1)$ are not awaits
\[ \forall k \in [k_G^T(u):k_G^T(u-1)). \ P_T(k) \not= \lstinline[language=Lambda]|await(-,-)| \]
Of course $k_G^T(t)$ is in that interval; thus all steps between $k_G^T(t)$ and $K_G^T(u-1)$ are not awaits.
By the semantics of the language, these steps moved control ahead one statement per step. Thus at most $n$ steps have passed between $t$ and $u-1$
\[ (u-1) - t  = k_G^T(u-1) - k_G^T(t) < n \]
Since step $u-1$ ends an iteration $q$ of an await of length $n$
\[ \mathit{end}_G^T(q) = u-1 \ \land \ \mathit{len}_G^T(q) = n \]
it follows that step $t$ is one of the steps in that iteration
\[ t \in [\mathit{start}_G^T(q):\mathit{end}_G^T(q)] \]
Because the loop condition is satisfied, the iteration is a failed iteration
\[ \mathit{fail}_G^T(q) \]
and we conclude from the Bounded-Effect principle: step $t$ does not produce a write
\[ e_G^T(t) \not= W^-(-,-) \]
which is a contradiction. This proves \cref{lem:bound:T}, \ie, each statement can produce at most one write.
Thus the total number of writes produced by thread $T$ is at most the size $\lvert P_T \rvert$ of the program text of thread $T$
\[ \lvert \Set{ t | e_G^T(t) = W^-(-,-),\ t <N_G^T } \rvert \ \le \ \lvert P_T \rvert  \]
The claim follows:
\[ \lvert \Set{ (T, t) | e_G^T(t) = W^-(-,-),\ T \in \mathcal T,\ t < N_G^T} \rvert \ =\  \sum_{T \in \mathcal T}  \lvert \Set{ t | e_G^T(t) = W^-(-,-),\ t <N_G^T } \rvert \ \le \ \sum_{T \in \mathcal T} \lvert P_T \rvert \]
\end{proof}
In a graph which satisfies the progress condition, each iteration of await reads from a different combination of writes. Since the set of writes is bounded, the possible number of combinations is bounded as well. Thus the number of failed iterations is also bounded.
\begin{lemma} \label{lem:bounded F}
\[ \finitegraphs{} \quad\text{is finite} \]
\end{lemma}
\begin{proof}
We show instead that each $G \in \finitegraphs{}$ has bounded length (bounded by some constant $B$). This implies that the graphs are finite as every $G$ can then be encoded as a pair of a sequence of $B$ events and a sequence of $B$ numbers in the range $[0:b-1]$ indicating \rf-edges. Since the set of events is finite, the number of such encodings is finite, and so is $\finitegraphs{}$.

Assume for the sake of contradiction that no such bound exists. Thus graphs in $\finitegraphs{}$ can be arbitrarily large. Since they are consistent with the program, this means that the program execution can become arbitrarily long.
Since there are only finitely many threads, one thread $T$ can be executed for arbitrarily long. Since the program text $P_T$ has finite length and thus only finitely many awaits, there has to be one await that can be made to fail arbitrarily often.
Let the line number of this await be $k$, and let $G^0$, $G^1$, $G^2$, \ldots be graphs in which the await fails zero times, once, twice, etc.
The await jumps back some number $n$ of statements
\[ P_T(k) = \lstinline[language=Lambda]|await($n$,$\loopcons$)| \]
and thus produces at most $n$ reads. By the progress condition, at least one of them must read from a new write in each iteration. Coherence forbids going back to the previous writes. By \cref{lem:bounded W} the number of available writes in each graph $G^i$ is at most $b$.
Thus there can be at most $n \cdot b$ iterations of this loop. Consider $G^{n\cdot b+1}$, \ie, the graph in which one iteration beyond that has been executed.
Note that after $l$ iterations of the await, at most $n \cdot b-l$ writes are still available to read from.
Thus in the final iteration $n \cdot b + 1$, the thread has at most $-1$ write to read from, which is a contradiction. 
This shows that such a bound $B$ must exist and thus $\finitegraphs{}$ is finite.
\end{proof}

This proves our claims about $\mathbb G^F$/$\finitegraphs{}$. Next we consider the graphs in $\infinitegraphs{}$ which indicate await violations.
\begin{lemma} \label{lem:traverse infinite}
Every graph $G \in \infinitegraphs{}$ is finite.
\end{lemma}
\begin{proof}
Observe that every $G \in \infinitegraphs{}$ is by definition an element of $\finitegraphs{}$ and thus finite.
\end{proof}
We proceed to show that graphs in $\infinitegraphs{}$ always indicate await terminations. This depends on the definition of stagnancy which we have not shown yet. We define it as follows. Let 
\[ \mathcal V_G = \Set{ T \in \mathcal T | k_G^T(N_G^T) < \lvert P_T \rvert }  \]
be the set of threads which have not terminated (yet) in $G$.
$G$ is stagnant iff all of the following are true:
\begin{enumerate}
\item Some threads have not terminated
\[ \mathcal V_G \not= \emptyset \]
\item All of those threads have just completed a failed await loop iteration
\[ \forall T \in V_G.\ N_G^T = \mathit{end}_G^T(q) \ \land \ \mathit{fail}_G^T(q) \]
\item There is no extension $G'$ of $G$ (with $G.X \subseteq G'.X$ for $X\in \Set{ \mathrm{E}, \rf,\mo}$) where any threads in $V_G$ have terminated and threads read only from stores that are already available in $G$
\[ \mathrm{cons}_M^P(G') \land \mathrm{Dom}(G'.\mathrm{rf}) \subseteq G.\mathrm{E} \ \to \ \mathcal V_G = \mathcal V_{G'} \]
\end{enumerate}


\begin{lemma}\label{lem:extend infinite}
\[ \infinitegraphs{} \not=\emptyset \quad\to\quad \mathbb G^\infty \not=\emptyset \]
\end{lemma}
\begin{proof}
Assume that $\infinitegraphs{}$ is non-empty. Thus there is some stagnant graph $G \in \infinitegraphs{}$ 
\[ \mathrm{stagnant}(G) \]
The set of threads that have not completed their program in $G$
\[ \mathcal V_G = \Set{ T \in \mathcal T | k_G^T(N_G^T) < \lvert P_T \rvert } \]
is by definition of stagnancy non-empty.

By definition of stagnancy, each of these threads in $\mathcal V_G$ must be in the $q_T+1$-th iteration of an await, and 2) iteration $q_T$ of the same await failed, and 3) there are no other writes to read from that would result in the thread terminating.
We extend the graph to an infinite graph $G'$ by adding failed await loop iterations in which each load reads from the \mo-maximal store to its location.
Let for $T \in \mathcal V_G$ the index of the previous (failed) iteration of the await be $q_T$
\[ \mathit{fail}_G^T(q_T) \ \land \ N_G^T = \mathit{end}_G^T(q_T) \]

We add for $T \in \mathcal V_G$ events for an infinite number of iterations, numbering each with an index $n \in \mathbb N$, and each step inside that iteration with an index $m \le \mathit{len}_G^T(q)$, replicating the same events over and over
\[ G'.\mathrm{E}_T = G.\mathrm{E} \cup \Set{  \langle T,\ N_G^T+n \cdot \mathit{len}_G^T(q) + m,\ e_G^T(\mathit{start}_G^T(q)+m) \rangle | n \in \mathbb N, \ m \le \mathit{len}_G^T(q) } \]
The newly added read events will read from the \mo-maximal stores to their locations. Let 
\[ e = \langle T,\ N_G^T+n \cdot \mathit{len}_G^T(q) + m,\ e_G^T(\mathit{start}_G^T(q)+m) \rangle \]
 be such a read event. We define
\[ G'.\rf(e) = \max_{G.\mo} G.\mathrm{W}_{\mathrm{loc}(e)} \]
Other than that we change nothing.
Obviously this repetition results in a wasteful execution. 
With \cref{lem:failed progress} it is easy to show that all of these new loop iterations are consistent with the program and are themselves failed wasteful iterations
\[ \mathrm{cons}_M^P(G') \]
This graph has infinitely many failed iterations of awaits and is thus in $\mathbb G^\infty$, which is the claim
\[ G' \in \mathbb G^\infty\]
\end{proof}

It only remains to show that every graph $G \in \mathbb G^\infty$ can be cut to a graph in $\infinitegraphs{}$.
\begin{lemma} \label{lem:cut infinite}
\[ \mathbb G^\infty \not= \emptyset \quad\to\quad  \infinitegraphs{} \not= \emptyset \]
\end{lemma}
\begin{proof}
Assume that there are graphs which violate await termination and let $G \in \mathbb G^\infty$ be such a graph.
Analogous to how we showed in \cref{lem:bounded F} that there is some thread $T$ which executes some await in line $k$ an arbitrary number of times, we can show that some thread executes infinitely many steps. Let $\mathcal V$ be the (non-empty) set of these threads
\[ \mathcal V = \Set{ T | N_G^T = \infty } \]
Each of these threads $T \in \mathcal V$ executes some await in line $k_T$ infinitely often
\[ P_T(k_T) = \lstinline[language=Lambda]|await($n$,$\kappa$)|  \ \land \ \forall t.\ \exists u \ge t.\ k_G^T(u) = k_T \]
By \cref{lem:bounded W} the number of writes is bounded, thus (as also shown in \cref{lem:bounded F}) this loop have only finitely many non-wasteful iterations. Starting from some iteration $q_T$, the writes observed by the await never change; iteration $q_T$ and all subsequent iterations are wasteful.
\[ \forall q' \ge q_T. \ \mathit{WI}_G^T(q') \]
Due to fairness, this is only allowed if there is no possibility for the reads to read anything else: otherwise, eventually one of the subsequent iterations would need to read from one of the other available writes and hence not be a wasteful iteration. This allows us to construct a graph $G'\in \infinitegraphs{}$ which stagnates.
We generate $G'$ by cutting down all failed wasteful iterations of threads in $\mathcal T$. Assume w.l.o.g. that the iteration $q_T$ is the first wasteful iteration of thread $T$ (the finitely many preceding wasteful iterations can be deleted iteratively with \cref{lem:deleted:aux})
\[ \neg \mathit{WI}_G^T(q'-1) \]
We define:
\[ G'.\mathrm{E}_T = \Set{ \langle T,\ t,\ e \rangle \in G.\mathrm{E}_T | t \le \mathit{end}_G^T(q_T) } \]
This ensures all conditions of $\mathit{stagnant}$: these threads have not terminated, by definition they just finished a failed await loop iterationevent, and the only available writes force the thread to stay in the loop indefinitely.
Thus we have
\[ \mathrm{stagnant}(G') \]
Furthermore, the graph is still consistent with the program because the beginning of the thread-local execution is exactly the same as for $G$
\[ \mathrm{cons}^P(G') \]
Since all wasteful iterations have been deleted, the graph $G'$ is not wasteful
\[ \neg \mathrm{W}(G') \]
and can thus in the set of graphs $\infinitegraphs{}$ that are searched by AMC
\[ G' \in \infinitegraphs{} \]
which proves the claim.
\end{proof}
The main theorem (\cref{thm:amc}) follows from \cref{lem:cut infinite,lem:cut finite,lem:extend infinite,lem:traverse infinite}.

%% file: cases.tex
\section{Study Cases}
\label{s:cases}

In this section, we discuss in detail three study cases:
a bug in the MCS lock of the DPDK library, a bug in the MCS lock of an internal Huawei product, and a comparison of expert-optimization and \sys-optimization of the Linux qspinlock.
We report about bugs found with \sys as well as limitations.

\subsection{DPDK MCS lock}
\label{s:dpdk}

The Data Plane Development Kit (DPDK\footnote{\url{https://github.com/DPDK/dpdk}}) is a popular set of libraries used to
develop packet processing software in user space.
\sys found a bug in the MCS lock of the current DPDK version (v20.05).
\Cref{fig:dpdk-mcs} shows the part of the implementation that concern us.
At the end of the code, we added the bug scenario in which two threads, Alice and Bob,
are involved.
Alice wants to acquire the lock (see \verb|run_alice()| function), and Bob currently
holds the lock and is about to release it (see \verb|run_bob()|).
Note that we removed the slowpath of \verb|rte_mcslock_unlock()| since the bug
only occurs in the fastpath.
The core of the bug is a missing \rel\ barrier before or at Line~\ref{ln:dpdk-bug}, which
causes Alice to hang and never enter the critical section.

\begin{figure}[ht]

\begin{multicols}{2}
    \begin{lstlisting}[style=casecode]
/* SPDX-License-Identifier: BSD-3-Clause
 * Copyright(c) 2019 Arm Limited
 */
 typedef struct rte_mcslock {
	struct rte_mcslock *next;
	int locked; /* 1 if the queue locked, 0 otherwise */
} rte_mcslock_t;

static inline void
rte_mcslock_lock(rte_mcslock_t **msl, rte_mcslock_t *me) {
  rte_mcslock_t *prev;

  /* Init me node */
  __atomic_store_n(&me->locked, 1, __ATOMIC_RELAXED); |\label{ln:dpdk-init}|
  __atomic_store_n(&me->next, NULL, __ATOMIC_RELAXED);

  /* If the queue is empty, the exchange operation is |\label{ln:dpdk-comment1}|
   * enough to acquire the lock. Hence, the exchange
   * operation requires acquire semantics. The store to
   * me->next above should complete before the node is
   * visible to other CPUs/threads. Hence, the exchange
   * operation requires release semantics as well. */
  prev = __atomic_exchange_n(msl, me, __ATOMIC_ACQ_REL); |\label{ln:dpdk-xchg}|
  if (prev == NULL) {
  	return;
  }
  __atomic_store_n(&prev->next, me, __ATOMIC_RELAXED); |\label{ln:dpdk-bug}|

  /* The while-load of me->locked should not move above
   * the previous store to prev->next. Otherwise it will
   * cause a deadlock. Need a store-load barrier. */
  __atomic_thread_fence(__ATOMIC_ACQ_REL); |\label{ln:dpdk-fence}|
  while (__atomic_load_n(&me->locked, __ATOMIC_ACQUIRE))
    rte_pause();
}

static inline void
rte_mcslock_unlock(rte_mcslock_t **msl, rte_mcslock_t *me) {
	if (__atomic_load_n(&me->next, __ATOMIC_RELAXED) == NULL) {
        // **ignore this branch**
	}

	/* Pass lock to next waiter. */
	__atomic_store_n(&me->next->locked, 0, __ATOMIC_RELEASE);
}
//---------------------------------------------------------
// bug scenario
//---------------------------------------------------------
// 2 threads: alice and bob.
rte_mcslock_t alice, bob;
// bob has the lock
rte_mcslock_t *tail = &bob;

void run_alice() { rte_mcslock_lock(&tail, &alice); }
void run_bob() { rte_mcslock_unlock(&tail, &bob); }
\end{lstlisting}
\end{multicols}
\caption{Part of the DPDK MCS lock implementation describing the scenario in which
Alice hangs.}\label{fig:dpdk-mcs}
\end{figure}

\begin{figure}[t]
    \begin{minipage}{.48\textwidth}
        \resizebox{\textwidth}{!}{\input{figures/dpdk-imm-bug}}
        \vspace{-3em}
        \caption{{\bf IMM}. Bug results in Alice hanging: Alice writes to {\tt bob->next} with \rlx\ mode,
        and Bob reads \rlx\ mode, which causes allows Bob's write to before the initialization of {\tt me->locked}.}\label{fig:dpdk-imm-bug}
    \end{minipage}\hfill
    \begin{minipage}{.48\textwidth}
        \resizebox{\textwidth}{!}{\input{figures/dpdk-imm-fix}}
        \vspace{-3em}
        \caption{{\bf IMM}. With the bug fixed, Alice writes to {\tt prev->next} with \rel\ mode,
        and Bob reads with \acq\ mode, creating a synchronizes-with edge, which forces Bob's write to occur after the initialization.}\label{fig:dpdk-imm-fix}
    \end{minipage}
\end{figure}

\paragraph{The bug on IMM.}
\Cref{fig:dpdk-imm-bug} shows an execution graph in which Alice hangs.
AMC gives exactly this execution graph as counter-example for await termination, but in the text form.
The $U_\seqc$ pair of events are the ``read part'' and ``write part'' of the atomic exchange (Line~\ref{ln:dpdk-xchg} of code in \cref{fig:dpdk-mcs}).
To understand why exchange is modeled with two events, remember that atomic exchange is implemented with
load-linked/store-conditional instruction pairs in many architectures.
For example in ARMv8, \verb|atomic_exchange_n(msl, me, __ATOMIC_ACQ_REL)| is compiled to
\begin{verbatim}
  38:	c85ffc02 	ldaxr	x2, [x0]
  3c:	c803fc01 	stlxr	w3, x1, [x0]
  40:	35ffffc3 	cbnz	w3, 38 <run_alice+0x20>
\end{verbatim}
Intuitively, the load instruction is the ``read part'' of the exchange,
whereas the store instruction is ``write part'' of the exchange.
Also note that for the sake of this bug, \verb|__ATOMIC_ACQ_REL| is equivalent to \verb|__ATOMIC_SEQ_CST|, \ie, even with the stronge \seqc\ barrier mode, the bug can still manifest.

Returning to the bug in IMM,
Alice starts by initializing her node, in particular, setting {\tt alice->locked} to 1.
After exchanging the tail, Alice writes to {\tt bob->next}.
Although Bob reads from Alice's write to {\tt bob->next}, IMM allows Alice's write
to {\tt alice->locked} to happen after Bob's write to {\tt alice->locked}
because no happens-before relation is established between Alice and Bob.
The \mo\ relation shows this order of modifications.
If that occurs, Alice's fate is to await {\tt alice->locked} to become 0 forever.
To establish the correct happens-before relation between Alice and Bob,
Alice's write to {\tt bob->next} has to be \rel,
and Bob's read of {\tt bob->next} has to be \acq\ (see \cref{fig:dpdk-imm-fix}).
That causes both events to ``synchronize-with'', guaranteeing Alice's write to {\tt alice->locked} to
happen before Bob's.
Note that in IMM the happens-before relation projected to one memory location, \eg, {\tt alice->locked}, implies the same order in the visible memory updates of that location, \ie, in the modification order \mo of {\tt alice->locked}.
The happens-before relation does not, however, imply an ordering between the writes to distinct memory locations\cite{podkopaev2019bridging}.

\paragraph{The bug on ARM.}
The bug is not exclusive to the IMM model.
\Cref{fig:dpdk-arm-bug} shows an execution graph manually adapted to the ARM memory model (ARM for short).
In this memory model, a global order of events exist; so, we number the events with one possible global order.
ARM allows the write to {\tt alice->locked} (event 7) to happen after the read part of the atomic exchange $U^R_\seqc$ (event 3),
but does not allow it happen after the write part $U^W_\seqc$ (event 8).
Moreover, since Alice's write to {\tt bob->next} is \rlx\ (event 4),
ARM allows the write to {\tt alice->locked} (along with $U^W_\seqc$) to happen after it.
As a consequence, although Bob reads (event 5) from Alice's write to {\tt bob->next} (event 4), the effect of
setting {\tt alice->locked} to 1 (event 7) happens after Bob has set it to 0 (event 6).
In contrast to IMM, making Alice's write to be \rel\ is sufficient in ARM (see \cref{fig:dpdk-arm-fix}) because
the control/address dependency between Bob's events guarantees that writes that happen before the read event also happen before the subsequent dependent events.

\paragraph{Validation of the bug.}
So far we were not able to reproduce the effect of the bug on real hardware.
The situation that triggers the bug is very unlikely to happen, but nevertheless
possible and still a potential problem for code using DPDK on ARM platforms.
To have a higher confidence about the bug on ARM, we checked the scenario
of \cref{fig:dpdk-mcs} with Rmem\cite{pulte2017simplifying}, a stateful model checking tool capable of verifying
small pieces of binary code compiled to ARMv8 architecture.
Although Rmem cannot deal with the infinite loop of Alice, we can reproduce the
bug by asserting Bob does not see {\tt alice->locked} being reverted to 1 after
he has set it to 0 -- as expected, the assertion fails.

\paragraph{Discussion.}
The DPDK MCS lock bug is a good example of how understanding WMMs can be
challenging even for experts.
Note the MCS lock has been contributed by ARM Limited to the DPDK project.
In \cref{fig:dpdk-mcs}, Line~\ref{ln:dpdk-comment1}, the developer
considers exactly the situation observed in this bug:
``The store to {\tt me->next} above should complete before the node is
visible to other CPUs/threads. Hence, the exchange operation requires release
semantics as well.''
However, making the exchange \rel\ is not sufficient because the node can also become
visible to another thread via the write at Line~\ref{ln:dpdk-bug},
and nothing stops the store of Line~\ref{ln:dpdk-init} and the write part of the exchange of Line~\ref{ln:dpdk-xchg}
to be reordered after Line~\ref{ln:dpdk-bug}.
Another interesting finding in this code is that, as far as we can verify,
the explicit fence at Line~\ref{ln:dpdk-fence} is useless and can be removed.

\begin{figure}[t]
    \begin{minipage}{.48\textwidth}
        \resizebox{\textwidth}{!}{\input{figures/dpdk-arm-bug}}
        \vspace{-3em}
        \caption{{\bf ARM memory model}. Bug results in Alice hanging: Alice writes to {\tt bob->next} with \rlx\ mode
        causing the initialization of {\tt alice->locked} to be reordered after Bob's write.}\label{fig:dpdk-arm-bug}
    \end{minipage}\hfill
    \begin{minipage}{.48\textwidth}
        \resizebox{\textwidth}{!}{\input{figures/dpdk-arm-fix}}
        \vspace{-3em}
        \caption{{\bf ARM memory model}. To fix the bug, Alice has to writes to {\tt bob->next} with \rel\ mode,
        forcing Bob's write to occur after the initialization of {\tt alice->locked}.}\label{fig:dpdk-arm-fix}
    \end{minipage}
\end{figure}

\subsection{MCS lock of an internal Huawei product}

\begin{figure}[h]
\begin{minipage}{.40\textwidth}
    \begin{lstlisting}[style=casecode]
static inline void
mcslock_acquire(volatile mcslock_t *tail,
                volatile mcs_node_t *me)
{
  mcs_node_t *prev;

  me->next = NULL;
  me->spin = 1;

  smp_wmb(); // ** consider to be SC fence **

  // equivalent to xchg_acq
  prev = __sync_lock_test_and_set(tail, me);
  if (!prev)
    return;

  prev->next = me;
  smp_mb(); // ** consider to be SC fence ** |\label{ln:huawei-useless-fence}|
  while(me->spin); |\label{ln:huawei-loop}|
  // BUG: Missing ACQ barrier, eg, smp_mb();
}

static inline void
mcslock_release(volatile mcslock_t *tail,
                volatile mcs_node_t *me)
{
  if (!me->next) {
    // SC cmpxchg
    if (__sync_val_compare_and_swap(
                   tail, me, NULL) == me) {
      return;
    }
    while(!me->next);
  }

  smp_mb(); // ** consider to be SC fence **
  me->next->spin = 0;
}
    \end{lstlisting}
    \caption{MCS implementation in a commercial OS. A barrier bug cases data races in the critical section.}\label{fig:huawei-mcs}
\end{minipage}\hfill
\begin{minipage}{.57\textwidth}
\resizebox{\linewidth}{!}{ \input {figures/huawei-mcs}}
\caption{In IMM, Alice's read of {\tt x} may happen before Bob's write to {\tt x}.}
\label{fig:huawei-graph}
\end{minipage}
\end{figure}
Our next study case is concerned with the MCS lock implementation found in an internal Huawei product.
In this implementation, \sys\ identified a missing \acq\ barrier that causes serious data corruption problems.
We were able to reproduce the problem on real hardware and reported the bug along with a simple fix to the maintainers.
Here, we describe this issue to illustrate the challenges of porting x86 code to ARM, which is the reason why such a bug was introduced in the code base.
With the recent increased demand of software for ARM servers, we believe that similar bugs are going to become more and more common in production.

\paragraph{The bug on IMM.}
\Cref{fig:huawei-mcs} presents a slightly simplified version of the original MCS lock implementation.
The bug is a missing \acq\ barrier at the end of {\tt mcslock_acquire()}.
To understand the scenario, consider the execution graph in \cref{fig:huawei-graph},
where the critical section is a simple increment \verb|x++|,
Alice wants to enter, and Bob is inside the critical section.
Similarly to the DPDK bug, Bob sees Alice's node when releasing the lock and sets Alice's {\tt spin = 0}; this flag is called {\tt locked} in DPDK.
The first fence in Alice's \verb|mcslock_acquire| synchronizes with the fence in Bob's \verb|mcslock_release| due to the write and read of {\tt bob->next} field.
That establishes a happens-before relation marked with the dashed arrows in the figure.
The happens-before relation, however, does not specify whether Bob's critical-section execution happens before Alice's critical-section execution, or vice-versa.
In this execution graph, Alice and Bob run their critical sections concurrently and both read from the initial write to {\tt x}, causing one of the increments to be lost.
By introducing an \acq\ barrier in the reads of {\tt me->spin} or after them (\cref{fig:huawei-mcs}, Line~\ref{ln:huawei-loop}), Alice is guaranteed to execute her critical section after Bob.
Note that, although the ARM model also introduces control dependencies, reads of {\tt me->spin} and the read of{\tt x} inside the critical section of Alice, a reordering of these operations is not precluded because they are all read operations.

\paragraph{Discussion.}
Besides the barrier bug, some issues may be interesting to point out.
The developers that implemented this code opted in using compiler specific atomic operations.
We do not recommend their use because they hide the barrier mode used underneath.
In particular,
{\tt _sync_lock_test_and_set} has an \acq\ mode,
whereas {\tt _sync_val_compare_and_swap} has an \seqc\ mode.
Moreover, the developers overuse fences:
the {\tt smb_mb()} fence in \verb|mcs_acquire()|, Line~\ref{ln:huawei-useless-fence}, is redundant and can be eliminated.


\input{cases-qspinlock}



%% file: figures/dpdk-imm-bug.tex
\begin{tikzpicture}[
        node distance = 12mm and 2mm,
        po/.style={->,thick, blue!50!black!50},
        rf/.style={->,thick, green!70!black},
        mo/.style={->,thick, red},
        ctrl/.style={->,thick, dashed},
        ev/.style={rounded corners, fill=gray!10},
        evx/.style={rounded corners, fill=red!70!black!40},
        event/.style={ev, anchor=west},
        thread/.style={text=white, fill=black!70},
    ]
    \draw[draw=none, thick, dotted, use as bounding box] (-55mm,-120mm) rectangle (55mm,5mm);
    \node[ev] (l) {$W_\mathit{init}$ {\tt locked, 0}};

    \node[below left=5mm and -4mm of l,ev] (t1 init next) {$W_\rlx$ {\tt alice->next, NULL}};
    \node[below=of t1 init next.west,event] (t1 init lock) {$W_\rlx$ {\tt alice->locked, 1}};
    \node[below=of t1 init lock.west,event] (t1 xchg read) {$U^\text{R}_\seqc$ {\tt tail, bob}};
    \node[below=of t1 xchg read.west,event] (t1 xchg tail) {$U^\text{W}_\seqc$ {\tt tail, alice}};
    \node[below=of t1 xchg tail.west,event,evx] (t1 write next) {$W_\rlx$ {\tt bob->next,  alice}};
    \node[below=of t1 write next.west,event] (t1 fence) {$F_\seqc$};
    \node[below=of t1 fence.west,event] (t1 await) {$R_\acq$ {\tt alice->locked, 1}};
    \node[below=of t1 await.west,event] (t1 await2) {$R_\acq$ {\tt alice->locked, 1}};
    \node[anchor=center,rotate=90] at ($(t1 await2.south west)+(5mm,-8mm)$)  (t1 await3) {\ldots};

    \node[below right=5mm and -4mm of l,evx] (t21) {$R_\rlx$ {\tt bob->next, alice}};
    \node[below=of t21.west, event] (t22) {$W_\rel$ {\tt alice->locked, 0}};

    \draw[po] ($(t1 init next.south west)+(5mm,0)$) -- ($(t1 init lock.north west)+(5mm,0)$);
    \draw[po] ($(t1 init lock.south west)+(5mm,0)$) -- ($(t1 xchg read.north west)+(5mm,0)$);
    \draw[po] ($(t1 xchg read.south west)+(5mm,0)$) -- ($(t1 xchg tail.north west)+(5mm,0)$);
    \draw[po] ($(t1 xchg tail.south west)+(5mm,0)$) -- ($(t1 write next.north west)+(5mm,0)$);
    \draw[po] ($(t1 write next.south west)+(5mm,0)$) -- ($(t1 fence.north west)+(5mm,0)$);
    \draw[po] ($(t1 fence.south west)+(5mm,0)$) -- ($(t1 await.north west)+(5mm,0)$);
    \draw[po] ($(t1 await.south west)+(5mm,0)$) -- ($(t1 await2.north west)+(5mm,0)$);
    \draw[po] ($(t1 await2.south west)+(5mm,0)$) -- (t1 await3);

    \draw[ctrl] ($(t1 init lock.south west)+(10mm,0)$) -- ($(t1 xchg read.north west)+(10mm,0)$);
    \draw[ctrl] ($(t1 xchg read.south west)+(10mm,0)$) -- ($(t1 xchg tail.north west)+(10mm,0)$);
    \draw[ctrl] ($(t1 xchg tail.south west)+(10mm,0)$) -- 
    node[midway, right, fill=white, font=\footnotesize] {hb}
    ($(t1 write next.north west)+(10mm,0)$);
    \draw[ctrl] ($(t21.south west)+(10mm,0)$) --
    node[midway, right, fill=white, font=\footnotesize] {hb}
    ($(t22.north west)+(10mm,0)$);

    \draw[po] ($(t21.south west)+(5mm,0)$) -- ($(t22.north west)+(5mm,0)$);

    \draw[rf] (t1 init lock.west) to[bend right=50]
    node[midway, fill=white, font=\footnotesize] {rf}
    (t1 await.west);
    \draw[rf] (t1 write next.east) to[out=30,in=200]
    node[midway, fill=white, font=\footnotesize] {rf}
    (t21.west);

    \draw[mo] (l.south) to[]
    node[midway, fill=white, font=\footnotesize] {mo}
    (t22.north west);
    \draw[mo] (t22.west) to
    node[midway, fill=white, font=\footnotesize] {mo}
    (t1 init lock.east);

    \node[thread, above=2mm of t1 init next] {Alice};
    \node[thread, above=2mm of t21] {Bob};
\end{tikzpicture}

%% file: figures/dpdk-imm-fix.tex
\begin{tikzpicture}[
    node distance = 12mm and 2mm,
        po/.style={->,thick, blue!50!black!50},
        rf/.style={->,thick, green!70!black},
        mo/.style={->,thick, red},
        ctrl/.style={->,thick, dashed},
        ev/.style={rounded corners, fill=gray!10},
        evx/.style={rounded corners, fill=red!70!black!40},
        evc/.style={rounded corners, fill=green!70!black!40},
        event/.style={ev, anchor=west},
        thread/.style={text=white, fill=black!70},
    ]
    \draw[draw=none, thick, dotted, use as bounding box] (-55mm,-120mm) rectangle (55mm,5mm);
    \node[ev] (l) {$W_\mathit{init}$ {\tt locked, 0}};

    \node[below left=5mm and -4mm of l,ev] (t1 init next) {$W_\rlx$ {\tt alice->next, NULL}};
    \node[below=of t1 init next.west,event] (t1 init lock) {$W_\rlx$ {\tt alice->locked, 1}};
    \node[below=of t1 init lock.west,event] (t1 xchg read) {$U^\text{R}_\seqc$ {\tt tail, bob}};
    \node[below=of t1 xchg read.west,event] (t1 xchg tail) {$U^\text{W}_\seqc$ {\tt tail, alice}};
    \node[below=of t1 xchg tail.west,event,evc] (t1 write next) {$W_\rel$ {\tt bob->next,  alice}};
    \node[below=of t1 write next.west,event] (t1 fence) {$F_\seqc$};
    \node[below=of t1 fence.west,event] (t1 await) {$R_\acq$ {\tt alice->locked, 1}};
    \node[below=of t1 await.west,event] (t1 await2) {$R_\acq$ {\tt alice->locked, 0}};

    \node[below right=5mm and -4mm of l,evc] (t21) {$R_\acq$ {\tt bob->next, alice}};
    \node[below=of t21.west, event] (t22) {$W_\rel$ {\tt alice->locked, 0}};

    \draw[po] ($(t1 init next.south west)+(5mm,0)$) -- ($(t1 init lock.north west)+(5mm,0)$);
    \draw[po] ($(t1 init lock.south west)+(5mm,0)$) -- ($(t1 xchg read.north west)+(5mm,0)$);
    \draw[po] ($(t1 xchg read.south west)+(5mm,0)$) -- ($(t1 xchg tail.north west)+(5mm,0)$);
    \draw[po] ($(t1 xchg tail.south west)+(5mm,0)$) -- ($(t1 write next.north west)+(5mm,0)$);
    \draw[po] ($(t1 write next.south west)+(5mm,0)$) -- ($(t1 fence.north west)+(5mm,0)$);
    \draw[po] ($(t1 fence.south west)+(5mm,0)$) -- ($(t1 await.north west)+(5mm,0)$);
    \draw[po] ($(t1 await.south west)+(5mm,0)$) -- ($(t1 await2.north west)+(5mm,0)$);

    \draw[ctrl] ($(t1 init lock.south west)+(10mm,0)$) -- ($(t1 xchg read.north west)+(10mm,0)$);
    \draw[ctrl] ($(t1 xchg read.south west)+(10mm,0)$) -- ($(t1 xchg tail.north west)+(10mm,0)$);
    \draw[ctrl] ($(t1 xchg tail.south west)+(10mm,0)$) to[bend right] ($(t1 write next.north west)+(15mm,0)$);
    \draw[ctrl] ($(t1 write next.north west)+(25mm,0)$)
    to[out=30, in = 180]
    node[midway, yshift=-4mm, fill=white, font=\footnotesize] {hb}
    (t21.175);
    \draw[ctrl] ($(t21.south west)+(10mm,0)$) -- ($(t22.north west)+(10mm,0)$);

    \draw[po] ($(t21.south west)+(5mm,0)$) -- ($(t22.north west)+(5mm,0)$);

    \draw[rf] (t1 init lock.west) to[bend right=50]
    node[midway, fill=white, font=\footnotesize] {rf}
    (t1 await.west);
    \draw[rf] (t1 write next.east) to[out=30,in=200]
    node[midway, fill=white, font=\footnotesize] {rf}
    (t21.west);
    \draw[rf,<-] (t1 await2.east) to[out=30, in=-90]
    node[midway, fill=white, font=\footnotesize] {rf}
    (t22.190);

    \draw[mo] (l.south) to[]
    node[midway, fill=white, font=\footnotesize] {mo}
    (t1 init lock.north east);
    \draw[mo] (t1 init lock.east) to
    node[midway, fill=white, font=\footnotesize] {mo}
    (t22.west);
    \node[thread, above=2mm of t1 init next] {Alice};
    \node[thread, above=2mm of t21] {Bob};

\end{tikzpicture}

%% file: figures/dpdk-arm-bug.tex
\begin{tikzpicture}[
        node distance = 12mm and 2mm,
        po/.style={->,thick, blue!50!black!50},
        rf/.style={->,thick, green!70!black},
        mo/.style={->,thick, red},
        ctrl/.style={->,thick, dashed},
        ev/.style={rounded corners, fill=gray!10},
        evx/.style={rounded corners, fill=red!70!black!40},
        event/.style={ev, anchor=west},
        thread/.style={text=white, fill=black!70},
    ]
    \draw[draw=none, thick, dotted, use as bounding box] (-55mm,-120mm) rectangle (55mm,5mm);
    \node[ev] (l) {$W_\mathit{init}$ {\tt locked, 0}};

    \node[below left=5mm and -4mm of l,ev] (t1 init next) {$W_\rlx$ {\tt alice->next, NULL}};
    \node[below=of t1 init next.west,event] (t1 init lock) {$W_\rlx$ {\tt alice->locked, 1}};
    \node[below=of t1 init lock.west,event] (t1 xchg read) {$U^\text{R}_\seqc$ {\tt tail, bob}};
    \node[below=of t1 xchg read.west,event] (t1 xchg tail) {$U^\text{W}_\seqc$ {\tt tail, alice}};
    \node[below=of t1 xchg tail.west,event,evx] (t1 write next) {$W_\rlx$ {\tt bob->next,  alice}};
    \node[below=of t1 write next.west,event] (t1 fence) {$F_\seqc$};
    \node[below=of t1 fence.west,event] (t1 await) {$R_\acq$ {\tt alice->locked, 1}};
    \node[below=of t1 await.west,event] (t1 await2) {$R_\acq$ {\tt alice->locked, 1}};
    \node[anchor=center,rotate=90] at ($(t1 await2.south west)+(5mm,-8mm)$)  (t1 await3) {\ldots};

    \node[below right=5mm and -4mm of l,ev] (t21) {$R_\rlx$ {\tt bob->next, alice}};
    \node[below=of t21.west, event] (t22) {$W_\rel$ {\tt alice->locked, 0}};

    \draw[po] ($(t1 init next.south west)+(5mm,0)$) -- ($(t1 init lock.north west)+(5mm,0)$);
    \draw[po] ($(t1 init lock.south west)+(5mm,0)$) -- ($(t1 xchg read.north west)+(5mm,0)$);
    \draw[po] ($(t1 xchg read.south west)+(5mm,0)$) -- ($(t1 xchg tail.north west)+(5mm,0)$);
    \draw[po] ($(t1 xchg tail.south west)+(5mm,0)$) -- ($(t1 write next.north west)+(5mm,0)$);
    \draw[po] ($(t1 write next.south west)+(5mm,0)$) -- ($(t1 fence.north west)+(5mm,0)$);
    \draw[po] ($(t1 fence.south west)+(5mm,0)$) -- ($(t1 await.north west)+(5mm,0)$);
    \draw[po] ($(t1 await.south west)+(5mm,0)$) -- ($(t1 await2.north west)+(5mm,0)$);
    \draw[po] ($(t1 await2.south west)+(5mm,0)$) -- (t1 await3);

    \draw[ctrl] (t1 xchg read.south east) to[bend left]
    node[right=-2mm, midway, fill=white, font=\footnotesize] {ctrl+addr}
    ($(t1 write next.north east)+(-3mm,0)$);
    \draw[ctrl] ($(t21.south west)+(10mm,0)$) to[bend left]
    node[right=2mm, midway, fill=white, font=\footnotesize] {ctrl+addr}
    ($(t22.north west)+(10mm,0)$);

    \draw[po] ($(t21.south west)+(5mm,0)$) -- ($(t22.north west)+(5mm,0)$);

    \draw[rf] (t1 init lock.west) to[bend right=50]
    node[midway, fill=white, font=\footnotesize] {rf}
    (t1 await.west);
    \draw[rf] (t1 write next.east) to[out=30,in=200]
    node[midway, fill=white, font=\footnotesize] {rf}
    (t21.west);

    \draw[mo] (l.south) to[]
    node[midway, fill=white, font=\footnotesize] {mo}
    (t22.north west);
    \draw[mo] (t22.west) to
    node[midway, fill=white, font=\footnotesize] {mo}
    (t1 init lock.east);

    \node[thread, above=2mm of t1 init next] {Alice};
    \node[thread, above=2mm of t21] {Bob};

    \node[circle,font=\scriptsize\bf, inner sep=1pt, fill=white] at (t1 init next.south west) {2};
    \node[circle,font=\scriptsize\bf, inner sep=1pt, fill=white] at (t1 init lock.south west) {7};
    \node[circle,font=\scriptsize\bf, inner sep=1pt, fill=white] at (t1 xchg read.south west) {3};
    \node[circle,font=\scriptsize\bf, inner sep=1pt, fill=white] at (t1 xchg tail.south west) {8};
    \node[circle,font=\scriptsize\bf, inner sep=1pt, fill=white] at (t1 write next.south west) {4};
    \node[circle,font=\scriptsize\bf, inner sep=1pt, fill=white] at (t1 fence.south west) {9};
    \node[circle,font=\scriptsize\bf, inner sep=1pt, fill=white] at (t1 await.south west) {10};
    \node[circle,font=\scriptsize\bf, inner sep=1pt, fill=white] at (t1 await2.south west) {11};
    \node[circle,font=\scriptsize\bf, inner sep=1pt, fill=white] at (t21.south west) {5};
    \node[circle,font=\scriptsize\bf, inner sep=1pt, fill=white] at (t22.south west) {6};
    \node[circle,font=\scriptsize\bf, inner sep=1pt, fill=white] at (l.south west) {1};
\end{tikzpicture}

%% file: figures/dpdk-arm-fix.tex
\begin{tikzpicture}[
    node distance = 12mm and 2mm,
        po/.style={->,thick, blue!50!black!50},
        rf/.style={->,thick, green!70!black},
        mo/.style={->,thick, red},
        ctrl/.style={->,thick, dashed},
        ev/.style={rounded corners, fill=gray!10},
        evx/.style={rounded corners, fill=red!70!black!40},
        evc/.style={rounded corners, fill=green!70!black!40},
        event/.style={ev, anchor=west},
        thread/.style={text=white, fill=black!70},
    ]
    \draw[draw=none, thick, dotted, use as bounding box] (-55mm,-120mm) rectangle (55mm,5mm);
    \node[ev] (l) {$W_\mathit{init}$ {\tt locked, 0}};

    \node[below left=5mm and -4mm of l,ev] (t1 init next) {$W_\rlx$ {\tt alice->next, NULL}};
    \node[below=of t1 init next.west,event] (t1 init lock) {$W_\rlx$ {\tt alice->locked, 1}};
    \node[below=of t1 init lock.west,event] (t1 xchg read) {$U^\text{R}_\seqc$ {\tt tail, bob}};
    \node[below=of t1 xchg read.west,event] (t1 xchg tail) {$U^\text{W}_\seqc$ {\tt tail, alice}};
    \node[below=of t1 xchg tail.west,event,evc] (t1 write next) {$W_\rel$ {\tt bob->next,  alice}};
    \node[below=of t1 write next.west,event] (t1 fence) {$F_\seqc$};
    \node[below=of t1 fence.west,event] (t1 await) {$R_\acq$ {\tt alice->locked, 1}};
    \node[below=of t1 await.west,event] (t1 await2) {$R_\acq$ {\tt alice->locked, 0}};

    \node[below right=5mm and -4mm of l,ev] (t21) {$R_\rlx$ {\tt bob->next, alice}};
    \node[below=of t21.west, event] (t22) {$W_\rel$ {\tt alice->locked, 0}};

    \draw[po] ($(t1 init next.south west)+(5mm,0)$) -- ($(t1 init lock.north west)+(5mm,0)$);
    \draw[po] ($(t1 init lock.south west)+(5mm,0)$) -- ($(t1 xchg read.north west)+(5mm,0)$);
    \draw[po] ($(t1 xchg read.south west)+(5mm,0)$) -- ($(t1 xchg tail.north west)+(5mm,0)$);
    \draw[po] ($(t1 xchg tail.south west)+(5mm,0)$) -- ($(t1 write next.north west)+(5mm,0)$);
    \draw[po] ($(t1 write next.south west)+(5mm,0)$) -- ($(t1 fence.north west)+(5mm,0)$);
    \draw[po] ($(t1 fence.south west)+(5mm,0)$) -- ($(t1 await.north west)+(5mm,0)$);
    \draw[po] ($(t1 await.south west)+(5mm,0)$) -- ($(t1 await2.north west)+(5mm,0)$);

    \draw[ctrl] (t1 xchg read.south east) to[bend left]
    node[right=-2mm, midway, fill=white, font=\footnotesize] {ctrl+addr}
    ($(t1 write next.north east)+(-3mm,0)$);
    \draw[ctrl] ($(t21.south west)+(10mm,0)$) to[bend left]
    node[right=2mm, midway, fill=white, font=\footnotesize] {ctrl+addr}
    ($(t22.north west)+(10mm,0)$);

    \draw[po] ($(t21.south west)+(5mm,0)$) -- ($(t22.north west)+(5mm,0)$);

    \draw[rf] (t1 init lock.west) to[bend right=50]
    node[midway, fill=white, font=\footnotesize] {rf}
    (t1 await.west);
    \draw[rf] (t1 write next.east) to[out=30,in=200]
    node[midway, fill=white, font=\footnotesize] {rf}
    (t21.west);
    \draw[rf, <-] (t1 await2.east) to[out=30, in=-90]
    node[midway, fill=white, font=\footnotesize] {rf}
    (t22.190);

    \draw[mo] (l.south) to[]
    node[midway, fill=white, font=\footnotesize] {mo}
    (t1 init lock.north east);
    \draw[mo] (t1 init lock.east) to
    node[midway, fill=white, font=\footnotesize] {mo}
    (t22.west);
    \node[thread, above=2mm of t1 init next] {Alice};
    \node[thread, above=2mm of t21] {Bob};


    \node[circle,font=\scriptsize\bf, inner sep=1pt, fill=white] at (t1 init next.south west) {2};
    \node[circle,font=\scriptsize\bf, inner sep=1pt, fill=white] at (t1 init lock.south west) {4};
    \node[circle,font=\scriptsize\bf, inner sep=1pt, fill=white] at (t1 xchg read.south west) {3};
    \node[circle,font=\scriptsize\bf, inner sep=1pt, fill=white] at (t1 xchg tail.south west) {5};
    \node[circle,font=\scriptsize\bf, inner sep=1pt, fill=white] at (t1 write next.south west) {6};
    \node[circle,font=\scriptsize\bf, inner sep=1pt, fill=white] at (t1 fence.south west) {8};
    \node[circle,font=\scriptsize\bf, inner sep=1pt, fill=white] at (t1 await.south west) {9};
    \node[circle,font=\scriptsize\bf, inner sep=1pt, fill=white] at (t1 await2.south west) {11};
    \node[circle,font=\scriptsize\bf, inner sep=1pt, fill=white] at (t21.south west) {7};
    \node[circle,font=\scriptsize\bf, inner sep=1pt, fill=white] at (t22.south west) {10};
    \node[circle,font=\scriptsize\bf, inner sep=1pt, fill=white] at (l.south west) {1};
\end{tikzpicture}

%% file: figures/huawei-mcs.tex
\begin{tikzpicture}[
    node distance = 12mm and 2mm,
        po/.style={->,thick, blue!50!black!50},
        rf/.style={->,thick, green!70!black},
        mo/.style={->,thick, red},
        ctrl/.style={->,thick, dashed},
        ev/.style={rounded corners, fill=gray!10},
        evx/.style={rounded corners, fill=red!70!black!40},
        evc/.style={rounded corners, fill=green!70!black!40},
        event/.style={ev, anchor=west},
        thread/.style={text=white, fill=black!70},
        region/.style={|-|,dashed, gray!30},
    ]
    \draw[draw=none, thick, dotted, use as bounding box] (-20mm,-140mm) rectangle (98mm,5mm);
    \node[ev] (l) {$W_\mathit{init}$ {\tt x, 0}};

    \node[below right=10mm and 0mm of l,ev] (t1 init next) {$W_\rlx$ {\tt me->next, NULL}};
    \node[below=of t1 init next.west,event] (t1 init lock) {$W_\rlx$ {\tt me->spin, 1}};
    \node[below=of t1 init lock.west,event] (t1 fence0) {$F_\seqc$};
    \node[below=of t1 fence0.west,event] (t1 xchg read) {$U^\text{R}_\acq$ {\tt m, bob}};
    \node[below=of t1 xchg read.west,event] (t1 xchg tail) {$U^\text{W}_\acq$ {\tt tail, me}};
    \node[below=of t1 xchg tail.west,event] (t1 write next) {$W_\rlx$ {\tt bob->next,  me}};
    \node[below=of t1 write next.west,event] (t1 fence) {$F_\seqc$};
    \node[below=of t1 fence.west,event,evx] (t1 await) {$R_\rlx$ {\tt me->spin, 1}};
    \node[below=of t1 await.west,event,evx] (t1 await2) {$R_\rlx$ {\tt me->spin, 0}};
    \node[below=of t1 await2.west,event] (t15) {$R_\rlx$ {\tt \ x, 0}};
    \node[below=of t15.west,event]      (t16) {$W_\rlx$ {\tt \ x, 1}};

    \node[below right=10mm and 50mm of l,ev] (t2r) {$R_\rlx$ {\tt x, 0}};
    \node[below=of t2r.west,event] (t2w) {$W_\rlx$ {\tt x, 1}};
    \node[below=of t2w.west,event] (t21) {$R_\rlx$ {\tt me->next, alice}};
    \node[below=of t21.west,event] (t22) {$F_\seqc$};
    \node[below=of t22.west,event] (t23) {$R_\rlx$ {\tt me->next, alice}};
    \node[below=of t23.west, event] (t24) {$W_\rlx$ {\tt alice->spin, 0}};

    \draw[po] ($(t1 init next.south west)+(5mm,0)$) -- ($(t1 init lock.north west)+(5mm,0)$);
    \draw[po] ($(t1 init lock.south west)+(5mm,0)$) -- ($(t1 fence0.north west)+(5mm,0)$);
    \draw[po] ($(t1 fence0.south west)+(5mm,0)$) -- ($(t1 xchg read.north west)+(5mm,0)$);
    \draw[po] ($(t1 xchg read.south west)+(5mm,0)$) -- ($(t1 xchg tail.north west)+(5mm,0)$);
    \draw[po] ($(t1 xchg tail.south west)+(5mm,0)$) -- ($(t1 write next.north west)+(5mm,0)$);
    \draw[po] ($(t1 write next.south west)+(5mm,0)$) -- ($(t1 fence.north west)+(5mm,0)$);
    \draw[po] ($(t1 fence.south west)+(5mm,0)$) -- ($(t1 await.north west)+(5mm,0)$);
    \draw[po] ($(t1 await.south west)+(5mm,0)$) -- ($(t1 await2.north west)+(5mm,0)$);
    \draw[po] ($(t1 await2.south west)+(5mm,0)$) -- ($(t15.north west)+(5mm,0)$);
    \draw[po] ($(t15.south west)+(5mm,0)$) -- ($(t16.north west)+(5mm,0)$);

    \draw[ctrl] ($(t1 init lock.south west)+(3mm,0)$) -- ($(t1 fence0.north west)+(3mm,0)$);
    \draw[ctrl] (t1 fence0.east)
    to
    node[midway, fill=white, font=\footnotesize] {hb}
    (t22.north west);
    \draw[ctrl] ($(t22.south west)+(2mm,0)$) to[out=260, in=100] ($(t24.north west)+(1mm,0)$);

    \draw[po] ($(t2r.south west)+(5mm,0)$) -- ($(t2w.north west)+(5mm,0)$);
    \draw[po] ($(t2w.south west)+(5mm,0)$) -- ($(t21.north west)+(5mm,0)$);
    \draw[po] ($(t21.south west)+(5mm,0)$) -- ($(t22.north west)+(5mm,0)$);
    \draw[po] ($(t22.south west)+(5mm,0)$) -- ($(t23.north west)+(5mm,0)$);
    \draw[po] ($(t23.south west)+(5mm,0)$) -- ($(t24.north west)+(5mm,0)$);

    \draw[rf] (t1 init lock.west) to[bend right=50]
    node[midway, fill=white, font=\footnotesize] {rf}
    (t1 await.west);

    \draw[rf] (t1 init lock.west) to[bend right=50]
    node[midway, fill=white, font=\footnotesize] {rf}
    (t1 await.west);
    \draw[rf] (t1 write next.east) to[out=30,in=200]
    node[midway, fill=white, font=\footnotesize] {rf}
    (t21.west);
    \draw[rf,<-] (t1 await2.east) to[out=30, in=-90]
    node[midway, fill=white, font=\footnotesize] (problem) {rf}
    (t24.190);
    \draw[rf] (l.east) to[out=0, in=140]
    node[midway, fill=white, font=\footnotesize] {rf}
    (t2r);
    \draw[rf] (l) to[out=-90, in=180]
    node[midway, fill=white, font=\footnotesize] {rf}
    (t15.west);

    \node[thread, above=2mm of t1 init next] {Alice};
    \node[thread, above=2mm of t2r] {Bob};

    \draw[region] ($(t21.north west)+(40mm,0)$) to[]
    node[midway, rotate=90, fill=white, font=\footnotesize, align=center] {{\tt mcslock_release()}}
    ($(t24.south west)+(40mm,0)$);
    \draw[region] ($(t2r.north west)+(20mm,0)$) to[]
    node[midway, rotate=90, fill=white, font=\footnotesize, align=center] {{\tt x++;}}
    ($(t2w.south west)+(20mm,0)$);

    \draw[region] ($(t1 init next.north west)+(-5mm,0)$) to[]
    node[midway, rotate=90, fill=white, font=\footnotesize, align=center] {{\tt mcslock_acquire()}}
    ($(t1 await2.south west)+(-5mm,0)$);
    \draw[region] ($(t15.north west)+(-5mm,0)$) to[]
    node[midway, rotate=90, fill=white, font=\footnotesize, align=center] {{\tt x++;}}
    ($(t16.south west)+(-5mm,0)$);

\end{tikzpicture}

%% file: cases-qspinlock.tex
\subsection{Linux qspinlock}

The qspinlock code was originally introduced in Linux version 4.2 as a new and faster spinlock \cite{qspinlock}.
Since then experts have slightly improved the algorithm and carefully optimized the barriers, achieving an excellent performance.
Currently, qspinlock is the default spinlock inside the Linux kernel for many architectures include x86 and ARM.
We now describe how we used \sys to automate the barrier optimization process of the Linux qspinlock and obtain similar barriers as the current code, at version 5.6 \cite{linux-qspinlock-version5.6}.

\paragraph{Baseline.}
Our optimization is based on Linux 4.4 \cite{linux-qspinlock-version4.4}, where the barriers were yet to be completely optimized, but other (algorithm) optimizations were mainly done.
Because our purpose is exclusively the barrier optimization, we ported a few remaining algorithmic optimizations present on 5.6 back to version 4.4.
Specifically, we backported the prefetch instructions to receive the next node in the queue.

\paragraph{Code preparation.}
To optimize existing code with \sys, we may have to perform minor changes.
In particular, if the code uses custom atomics implemented in assembly, these have to be replaced
with either compiler builtin atomics or with \sys-atomics (which compile down to builtin atomics).
In the case of qspinlock, we replaced the Linux's atomic operations with \sys-atomics.
Because custom atomics do not always follow the same model as IMM and compiler builtin atomics,
some discrepancies may arise during the replacement.
We encountered one such case:
The {\tt cmpxchg} function in Linux is defined as having a full memory barrier (\ie, an \seqc\ fence) before and after the operation
only if it succeeds\cite{lkmm-openstd}.
So we replace the Linux {\tt cmpxchg} with a wrapper that calls \sys's {\tt atomic_cmpxchg} and additionally an atomic fence in the success case to mimic the original behavior (for reference, see \cref{fig:linux-cmpxchg}) .
Another issue we encountered were unions of different-sized variables:
The qspinlock code uses a union that allows the same memory location to be read and written with either 8, 16 or 32 bits.
Currently, AMC requires that accesses to the same memory locations always have the same size -- this limitation may be fixed in the future.
To overcome this issue, we replaced accesses to the qspinlock data structure always 32 bits.

\begin{table}[h] \centering 
		\begin{tabular}{@{\extracolsep{5pt}} l|rrrrr}
			Version &  \Acquire & \Release & \SeqCst &  Time  & Correctness \\ 
			\hline \\[-1.8ex] 
			Linux 4.4\cite{linux-qspinlock-version4.4}	& $3$ & $6$ & $6$ & 2015/09/11      & Not verified \\
			Linux 4.5\cite{linux-qspinlock-version4.5}	& $6$ & $2$ & $1$ & 2015/11/09      & Barrier bug, fixed in \cite{linux-qspinlock-version4.16} \\ 
			Linux 4.8\cite{linux-qspinlock-version4.8}	& $6$ & $3$ & $0$ & 2016/06/03      & Barrier bug, fixed in \cite{linux-qspinlock-version4.16} \\ 
			Linux 4.16\cite{linux-qspinlock-version4.16}	& $6$ & $4$ & $0$ & 2018/02/13      & Not verified  \\ 
			Linux 5.6\cite{linux-qspinlock-version5.6}	& $6$ & $2$ & $1$ & 2020/01/07      & Not verified  \\ 
			$\sys$						& $7$ & $2$ & $1$ & \qspinopttime{} &  \sys-verified\\ 
		\end{tabular} 
	\caption{Barrier optimization results for Linux's qspinlock}
	\label{tab:barrier-opt}
\end{table}




\colorlet{MyRed}{red!40!black!90!}
\colorlet{MyBlue}{blue!40!black!90!}

\newcommand{\cred}[1]{\tikz[baseline]\node[anchor=base,circle,text=white,fill=MyRed,font=\bf\scriptsize,inner sep=0pt, minimum size=10pt]{#1};}
\newcommand{\cblue}[1]{\tikz[baseline]\node[anchor=base,circle,text=white,fill=MyBlue,font=\bf\scriptsize,inner sep=0pt, minimum size=10pt]{#1};}

\paragraph{Optimization results.}
\sys recommended barrier modes similar to those used by the experts (see Table~\ref{tab:barrier-opt}) in roughly \qspinopttime{} -- in contrast, the expert optimization took several release cycles.
The details of the optimization can be seen in \cref{fig:qspin-optimize}.
The bold marked text refers the the optimization suggested by \sys.
The boxes \cred{1} to \cred{5} are the Linux {\tt cmpxchg} function.
We see that \sys removes all atomic fences and transfer their barrier to atomic operations.

We now relate the \sys optimization with the optimizations made by the Linux experts over the several release cycles:
\begin{description}
	\item[Version 4.5 -- optimization of cmpxchg:]
	The first optimization by the experts was exactly in the \mbox{\tt cmpxchg} functions (\cred{1} to \cred{5}), changing it from \seqc\ to a more relaxed mode\cite{linux-qspinlock-version4.5}.

\item[Version 4.8 -- optimization of unlock function:]
	Experts optimized the code in \cblue{6},
	removing the fence and changing the {\tt atomic_sub} to \rel\ mode\cite{linux-qspinlock-version4.8} --
	identically to \sys suggestion.

\item[Version 4.16 -- bug fix:]
	The optimization from version 4.5 introduced a bug that was only found and fixed in version 4.16 \cite{linux-qspinlock-version4.16}.
	The bug allowed the node initialization to occur after the update of \mbox{\tt prev->next}, similarly to DPDK's bug discussed in \S\ref{s:dpdk}.
	In version 4.16, the experts used a \rel\ barrier in the atomic write immediately after the {\tt decode_tail} function, but finally replace that with an atomic \seqc-fence in the current version.
	Optimizations with \sys are verified and hence not affected by such bugs.

\item[Version 5.6 -- current version:]
	\Cref{fig:qspin-curr} shows the barrier modes used in the current version of qspinlock\cite{linux-qspinlock-version5.6}.
	The dotted lines connect our barriers in \cref{fig:qspin-optimize} with the equivalent barriers in the current version.
	The few different barrier modes are due to two reasons:
	First, there exists multiple maximally-relaxed combinations that are correct.
	Second, both optimizations are based on different memory models (LKMM and IMM).
	\sys extended with an LKMM module would likely suggest the same barriers as used in Linux.
\end{description}

\begin{figure}[H]
	\begin{minipage}[b]{.44\textwidth}
\begin{lstlisting}[style=verbcodelarge]
lock
  |\tikz[overlay, remember picture]\coordinate (a1);
  |atomic32_cmpxchg_rel --> $acquire$|\tikz[overlay, remember picture]\coordinate (o1);|
  |\tikz[overlay, remember picture]\coordinate (b1);
  |atomic_fence --> $remove$
  queued_spin_lock_slowpath
    atomic32_await_neq_rlx
    |\tikz[overlay, remember picture]\coordinate (a2);
    |atomic32_cmpxchg_rel --> $acquire$|\tikz[overlay, remember picture]\coordinate (o2);|
    |\tikz[overlay, remember picture]\coordinate (b2);
    |atomic_fence --> $remove$
    atomic32_await_mask_eq_acq --> $relaxed$
    |\tikz[overlay, remember picture]\coordinate (b8);
    |atomic32_add_rlx --> $acquire$|\tikz[overlay, remember picture]\coordinate (o3);|
    encode_tail
    atomic32_write_rlx
    atomicptr_write_rlx
    atomic32_read_rlx
    |\tikz[overlay, remember picture]\coordinate (a3);
    |atomic32_cmpxchg_rel --> $acquire$|\tikz[overlay, remember picture]\coordinate (o4);|
    |\tikz[overlay, remember picture]\coordinate (b3);
    |atomic_fence --> $remove$
    |\tikz[overlay, remember picture]\coordinate (a9);
    |atomic32_read_rlx
    |\tikz[overlay, remember picture]\coordinate (a4);
    |atomic32_cmpxchg_rel --> $seq_cst$|\tikz[overlay, remember picture]\coordinate (o5);|
    |\tikz[overlay, remember picture]\coordinate (b4);
    |atomic_fence --> $remove$
    decode_tail
    atomicptr_write_rlx
    atomic32_await_neq_acq|\tikz[overlay, remember picture]\coordinate (o6);|
    atomicptr_read_rlx
    |\tikz[overlay, remember picture]\coordinate (a10);
    |atomic32_await_mask_eq_acq --> $relaxed$
    atomic32_or_rlx --> $acquire$|\tikz[overlay, remember picture]\coordinate (o7);|
    |\tikz[overlay, remember picture]\coordinate (a5);
    |atomic32_cmpxchg_rel --> $acquire$|\tikz[overlay, remember picture]\coordinate (o8);|
    |\tikz[overlay, remember picture]\coordinate (b5);
    |atomic_fence --> $remove$
    atomicptr_await_neq_rlx
    |\tikz[overlay, remember picture]\coordinate (b10);
    |atomic32_write_rel|\tikz[overlay, remember picture]\coordinate (o9);|
unlock
  |\tikz[overlay, remember picture]\coordinate (a6);
  |atomic_fence --> $remove$
  |\tikz[overlay, remember picture]\coordinate (b6);
  |atomic32_sub_rlx --> $release$|\tikz[overlay, remember picture]\coordinate (o10);
  \begin{tikzpicture}[overlay, remember picture, b/.style={opacity=0.2, draw=none, rounded corners}, n/.style={circle, inner sep=1pt, font=\rm\footnotesize\bf, draw=none, text=white}]

  \coordinate (rex) at ($(a1)+(58mm,3mm)$);
  \draw[b,fill=MyRed]    ($(b1)+(-1mm,-1mm)$) rectangle (rex);
  \node[n,fill=MyRed] at ($(a1)+(-3mm,-.8mm)$) {1};

  \draw[b,fill=MyRed]    ($(b2)+(-1mm,-1mm)$) rectangle ($(rex|-a2)+(4mm,3mm)$);
  \node[n,fill=MyRed] at ($(a2)+(-3mm,-.8mm)$) {2};

  \draw[b,fill=MyRed]    ($(b3)+(-1mm,-1mm)$) rectangle ($(rex|-a3)+(4mm,3mm)$);
  \node[n,fill=MyRed] at ($(a3)+(-3mm,-.8mm)$) {3};

  \draw[b,fill=MyRed]    ($(b4)+(-1mm,-1mm)$) rectangle ($(rex|-a4)+(4mm,3mm)$);
  \node[n,fill=MyRed] at ($(a4)+(-3mm,-.8mm)$) {4};

  \draw[b,fill=MyRed]    ($(b5)+(-1mm,-1mm)$) rectangle ($(rex|-a5)+(4mm,3mm)$);
  \node[n,fill=MyRed] at ($(a5)+(-3mm,-.8mm)$) {5};

  \draw[b,fill=MyBlue] ($(b6)+(-1mm,-1mm)$) rectangle ($(a6)+(50mm,3mm)$);
  \node[n,fill=MyBlue] at ($(a6)+(-3mm,-.8mm)$) {6};
  \end{tikzpicture}
  |
\end{lstlisting}
\caption{Barrier modes in version 4.4 and \sys optimizations in {\bf bold}.
	Optimizations in the red boxes are similar to version 4.5;
those in the blue box are identical to version 4.8.}
		\label{fig:qspin-optimize}
	\end{minipage}\hfill
	\begin{minipage}[b]{.42\textwidth}
\begin{lstlisting}[style=verbcodelarge]
lock
  |\tikz[overlay,remember picture]\coordinate (c1);
  |@atomic32_cmpxchg_acq@
  queued_spin_lock_slowpath
    atomic32_await_counter_neq_rlx
    |\tikz[overlay,remember picture]\coordinate (c2);
    |@atomic32_get_or_acq@
    atomic32_sub_rlx
    |\tikz[overlay,remember picture]\coordinate (c3);
    |@atomic32_await_mask_eq_acq@
    atomic32_add_rlx
    encode_tail
    grab_mcs_node
    atomic32_write_rlx
    atomicptr_write_rlx
    atomic32_read_rlx
    |\tikz[overlay,remember picture]\coordinate (c4);
    |@atomic32_cmpxchg_acq@
    |\tikz[overlay,remember picture]\coordinate (c5);
    |@atomic_fence@
    atomic32_read_rlx
    atomic32_cmpxchg_rlx
    decode_tail
    atomicptr_write_rlx
    |\tikz[overlay,remember picture]\coordinate (c6);
    |atomic32_await_neq_acq
    atomicptr_read_rlx
    |\tikz[overlay,remember picture]\coordinate (c7);
    |@atomic32_await_mask_eq_acq@
    atomic32_cmpxchg_rlx
    atomic32_or_rlx
    atomicptr_await_neq_rlx
    |\tikz[overlay,remember picture]\coordinate (c8);
    |@atomic32_write_rel@
unlock
  |\tikz[overlay,remember picture]\coordinate (c9);
  |@atomic32_sub_rel@
  
  |
  \begin{tikzpicture}[overlay, remember picture, bline/.style={dotted, thick}]
	    \draw[bline] ($(o1)+(.5mm,.5mm)$) to[out=0,in=180] ($(c1)+(-.5mm,.5mm)$);
	    \draw[bline] ($(o2)+(.5mm,.5mm)$) to[out=0,in=180] ($(c2)+(-.5mm,.5mm)$);
	    \draw[bline] ($(o3)+(.5mm,.5mm)$) to[out=0,in=180] ($(c3)+(-.5mm,.5mm)$);
	    \draw[bline] ($(o4)+(.5mm,.5mm)$) to[out=0,in=180] ($(c4)+(-.5mm,.5mm)$);
	    \draw[bline] ($(o5)+(.5mm,.5mm)$) to[out=0,in=180] ($(c5)+(-.5mm,.5mm)$);
	    \draw[bline] ($(o6)+(.5mm,.5mm)$) to[out=0,in=180] ($(c6)+(-.5mm,.5mm)$);
	    \draw[bline] ($(o7)+(.5mm,.5mm)$) to[out=0,in=180] ($(c7)+(-.5mm,.5mm)$);
	    \draw[bline] ($(o8)+(.5mm,.5mm)$) to[out=0,in=180] ($(c7)+(-.5mm,.5mm)$);
	    \draw[bline] ($(o9)+(.5mm,.5mm)$) to[out=0,in=180] ($(c8)+(-.5mm,.5mm)$);
	    \draw[bline] ($(o10)+(.5mm,.5mm)$) to[out=0,in=180] ($(c9)+(-.5mm,.5mm)$);
  \end{tikzpicture}
|
\end{lstlisting}
\caption{Barrier mode information for qspinlock in Linux version 5.6 (current version).
	Dotted lines connect related barrier optimizations of \sys and the current version.
}
		\label{fig:qspin-curr}
	\end{minipage}
\end{figure}

\begin{figure}[h]
		\centering
		\begin{minipage}[b]{.6\linewidth} %
			\begin{lstlisting}[style=casecode,tabsize=8]
#define __linux_atomic_cmpxchg(mod1, mod2, l, a, b)    \
({                                                     \
	typeof(a) __r = @atomic_cmpxchg@##$mod1$(l, a, b)  \
	if (__r == a) @atomic_fence@##$mod2$();            \
	__r;                                           \
})

#define linux_cmpxchg(l, a, b)		__linux_atomic_cmpxchg($_rel$,     , l, a, b)
#define linux_cmpxchg$_rlx$(l, a, b)	__linux_atomic_cmpxchg($_rlx$, $_rlx$, l, a, b)
#define linux_cmpxchg$_acq$(l, a, b)	__linux_atomic_cmpxchg($_acq$, $_rlx$, l, a, b)
#define linux_cmpxchg$_rel$(l, a, b)	__linux_atomic_cmpxchg($_rel$, $_rlx$, l, a, b)
#define linux_cmpxchg$_seq$(l, a, b)	__linux_atomic_cmpxchg(    , $_rlx$, l, a, b)
\end{lstlisting}
\end{minipage}
\caption{Using \sys atomics to implement code compatible with Linux's {\tt cmpxchg}.}
	\label{fig:linux-cmpxchg}
\end{figure}

%% file: eval.tex
\section{Optimized-code Evaluation}
\label{s:eval}

\newcommand{\plotwidth}{0.6}

In this section, we present details of the setup used in our ``optimized-code evaluation'' section.
Moreover, we discuss the results obtained with microbenchmarks in length.
See the full paper for results with real-world workloads~\cite{vsync}.

\subsection{Setup details}

        \subsubsection{Evaluation platforms}
        \label{subsub:eval-platforms}

            We conduct our experiments in the following hardware platforms:

            \begin{itemize}
                \item a Huawei \textbf{TaiShan 200} (Model 2280)\footnoteurl{
                    https://e.huawei.com/uk/products/servers/taishan-server/taishan-2280-v2
                }
                rack server that has \texttt{128 GB} of RAM and 2 \textbf{Kunpeng 920-6426} processors, a HiSilicon chip
                with \textbf{64} \texttt{ARMv8.2} 64-bit cores\footnoteurl{
                    https://en.wikichip.org/wiki/hisilicon/kunpeng/920-6426
                },
                totaling 128 cores running at a nominal 2.6 GHz frequency.
                The identifier to denote this machine in this document is \texttt{taishan200-128c}.

                \item a GIGABYTE \textbf{R182-Z91-00}\footnoteurl{
                    https://www.gigabyte.com/Rack-Server/R182-Z91-rev-100
                }
                rack server that has \texttt{128 GB} of RAM and 2 \textbf{EPYC 7352} processors, an AMD chip
                with \textbf{24} \texttt{x86\_64} cores\footnoteurl{
                    https://www.amd.com/en/products/cpu/amd-epyc-7352
                },
                totaling 48 cores (96 if counting hyperthreading) running at a nominal 2.3 GHz frequency.
                The identifier to denote this machine in this document is \texttt{gigabyte-96c}.
            \end{itemize}

            We installed on all these servers the Ubuntu 18.04.4 LTS (aarch64) operating system, with the following
            Linux kernel version: \texttt{5.3.0-42-generic}.

        \subsubsection{Environment setup}

            To produce stable benchmark results on a kernel as complex as Linux, we took some precautions in terms of
            environment configuration of our experiment target platforms.
            We list here such precautions:

            \begin{enumerate}
                \item \textbf{Atomic types isolation.}
                Linux and \sys each declare their own atomic types such as \texttt{atomic\_t} and \texttt{atomic64\_t}.
                When writing the kernel benchmark module, to avoid name conflicts between Linux kernel headers and \sys library headers, we separated into different translation units the benchmark ``main'' code (where the entry point lies and where the kernel threads are created) from the lock primitives function definitions and data structures instantiations (where the contention loops are executed, see Section~\ref{subsub:eval-micro-details}).
                Therefore, the ``main'' code of the benchmark kernel module could use the classic Linux headers for its needs while the \sys test units could include the \sys library headers (that define the required atomic types) and use from Linux only the primitives types (such as \texttt{uint32\_t} and the likes).
                This technique also enables the possibility to benchmark individually module from the Linux kernel, such as the qspinlock located in \texttt{"linux/spinlock.h"} header.

                \item \textbf{Thread to core affinity assignment.}
                The benchmark module spawns as many \emph{kernel} threads as requested on the module invocation command.
                To measure the multi-core overheads of the locks, these threads must be pinned on individual cores (both within the same NUMA nodes and on different NUMA nodes).
                For this purpose, the Linux \texttt{kthread\_bind()} function is used.

                \item \textbf{Operating frequency fixing.}
                To avoid suffering from thermal effects, and thus the OS dynamically changing the operating frequency while the benchmark were running (and by doing so skewing our results), we fixed the frequency to 1.5 GHz, a frequency point available on all the platforms used in our evaluations.
                For this purpose, we used the Linux \texttt{cpufreq} mechanism.
                We set the governor to \texttt{userspace} to be able to choose the frequency.
                We observed that using a fixed governor such as \texttt{userspace} instead of an adaptive one (such as \texttt{ondemand}) yields way better predictability in our results.

                \item \textbf{Disable network.}
                In the preliminary experiments we conducted for our work, we observed that network introduced lots of noise into the evaluations by widely spreading the distributions of results (more than 10\% difference between minimum and maximum observed throughput for a benchmark repeated with the same parameters).
                The simplest solution to avoid these interferences was to disable the network and therefore to operate the benchmarks directly on the machine workstation (on the server \texttt{tty}).

                \item \textbf{Disable IRQ balancing.}
                \texttt{irqbalance} is a Linux daemon in charge of distributing the hardware Interrupts Requests (IRQs) among the different processing cores of the platform for the purpose of overall system performance.
                However, sporadic IRQs and subsequent execution of Interrupt Service Routines (ISRs) occurring on an \emph{uncontrolled} set of cores would bring unpredictability in the system response-time and interfere with our benchmark measurements.
                We simply disable this mechanism.
                Therefore, the Linux fallback strategy is to pin all IRQs to the first core, which we remove from our thread affinity assignment to completely avoid the issue of running ISRs and benchmarks concurrently on the same cores.

                \item \textbf{Disable NUMA balancing.}
                On platforms with a large number of cores such as the ones used in these experiments (see Section~\ref{subsub:eval-platforms}), the CPU cores are organized in NUMA nodes (for \emph{Non-Uniform Memory Access}).
                This structure allows to palliate the unavoidable pressure on the memory bus due to the high amount of processing cores operating in parallel.
                Banks of memory are allocated per NUMA node, reflecting the cache hierarchy.
                \textbf{NUMA balancing} is a feature of Linux that periodically moves the tasks closer to the memory they use, \ie in the right NUMA node.
                \textbf{NUMA control} is an additional tool allowing to configure NUMA-aware task scheduling and memory allocation in a fine-grain manner, this overriding the overall system NUMA balancing.
                We disable system NUMA balancing, and we enable task-local NUMA control, using this syntax:
\begin{lstlisting}[language=bash]
sudo numactl --cpubind=0 --membind=0 <cmd_to_insert_kernel_module>
\end{lstlisting}
                This goes one step further as task affinity assignment, as it forces memory allocation to be bound on NUMA node 0.
                Our benchmark being inherently concurrent, as soon as there will be more threads than cores per NUMA node, these threads will be allocated on the next NUMA node.
		For userspace benchmarks, we use {\tt libnuma} directly in our {\tt pthread} wrapper to pin threads to cores and control the allocation of context data structures.
	        Therefore, the cross-node threads will suffer some performance loss when trying to access shared data (\eg{} spinlock data structures).

                \item \textbf{Completely isolate the cores.}
                The Linux kernel provides the possibility to isolate a subset of the CPU cores.
                To do so, the parameter \texttt{isolcpus} must be filled with the list of CPU core identifiers to isolate.
                This parameter is given on the Linux boot argument command line (\ie in Grub for our case, prior to the kernel boot).
                This has the effect of completely preventing the scheduling and other task balancing mechanisms (such as SMP balancing) to operate on these cores.
                Unless \emph{explicitly} required with task affinity configuration (with the corresponding system calls or by using the \emph{taskset} program), the OS will not schedule any task on these isolated cores.
                In our case, we decided to isolate all cores but the first, with the idea of running our benchmarks on the isolated cores, while the rest of the Linux processes would run on the Core 0 to avoid interfering with our results.

                \item \textbf{Kernel threads priority.}
                We tried several configurations of \texttt{niceness} for our kernel threads by calling the \texttt{set\_user\_nice()} Linux function, but this did not seem to impact the distribution of our results.
                This is to be expected with the precautions described above.
                The benchmark response time variability weren't influenced by the priority of the kernel threads.
            \end{enumerate}

    \subsection{Microbenchmark evaluation}

        In this section, we present the microbenchmark experiments carried-out for the paper.
        We first describe the experiment itself and then discuss the results obtained.

        \subsubsection{Experiment details}
        \label{subsub:eval-micro-details}

        The microbenchmark works as follow:
        each thread repeatedly acquires a (writer) lock, increments a shared counter, and releases the lock.
        This is summarized in pseudo-C in Listing~\ref{listing:eval-benchstructure}.

\begin{lstlisting}[style=tutorialcode,label=listing:eval-benchstructure,caption=Pseudo-C code of microbenchmark]
/* the tested lock is parameterized here */
#include <vsync/<chosen>lock.h>

/* lock variable */
static lock_s lock_var;

/*
 * supposedly already allocated,
 * we also take care of cacheline-alignment
 * to avoid false sharing.
 */
unsigned long long* shared_counter;

/* ... */

unsigned long long run()
{
    *shared_counter = 0ull;
    lock_init(&lock_var);

    do {
        lock_acquire(&lock_var);

        (*shared_counter)++;

        lock_release(&lock_var);
    } while (!thread_should_stop());

    return *shared_counter;
}
\end{lstlisting}

        \input{generated/microbench-constants}

        The returned counter is used to compute the throughput (number of times the critical sections was accessed by any thread).
        We vary the number of threads in the following set~\footnote{
            Obviously, the 127-thread case can only be run on platforms with 128 cores.
            It is omitted in the other cases.
        }:
        $\lbrace \evalMicrobenchThreadRange{} \rbrace$.
        We run each experiment for a fixed period of time (\evalMicrobenchDuration{} seconds) and measure the throughput (number of critical sections per second).
        We run the experiments \evalMicrobenchNbRuns{} times to ensure the stability of the results (for each case, we pick the median of these repeated runs).

        For spinlocks and reader-writer locks, the benchmark runs as a Linux kernel module.
        It means that we insert in the kernel a kernel module (\texttt{*.ko} file) that is linked with the code listed in Listing~\ref{listing:eval-benchstructure}.
        When inserted, the module runs an initialization routine that consists in spawning as many kernel threads as requested (as explained above, this is a parameter of the experiment), and each thread runs the code listed in Listing~\ref{listing:eval-benchstructure}, until interrupted by a timer (the \evalMicrobenchDuration{}-seconds execution).
        This timer is triggered externally, when removing the module from the kernel, as the exit routine of the module is requiring all the threads to finish their execution (it is not possible to kill a kernel thread with a kill signal).
        This kernel module is inspired by the previous work of Kashyap \etal~\cite{shflock2019}, where they (nano-)benchmarked the behavior of a hash-table in kernel-mode to evaluate several locking primitives.

        For mutexes, we replace \verb|pthread_mutex| using \verb|LD_PRELOAD|, and the benchmark runs in Linux userspace.

        About the selected locking primitives, we compare two variants of each primitive:
        an \seqc-only variant,
        and a \sys-optimized variant with barriers.

        \subsubsection{Results}

        \paragraph{Grouping and filtering records.}

            The raw experiment results look like a list of records as showed in Table~\ref{table:raw-records}.
            The \texttt{count} column represents the value returned by the function \texttt{run()} of Listing~\ref{listing:eval-benchstructure}, \ie the number of times a thread access the critical section.
            The \texttt{duration} column is the measured duration (even if it is fixed, some deviation may occur).
            Lastly, the \texttt{throughput} column is simply $\frac{\texttt{count}}{\texttt{duration}}$, effectively capturing the number of critical sections per second.

            \begin{table}
                \begin{minipage}{\textwidth}
                    \begin{center}
                        \resizebox{1.0\textwidth}{!}{%
                            \input{generated/raw-records}

                        }
                    \end{center}
                \end{minipage}
                \caption{Raw captured records, with parameters and output values.}
                \label{table:raw-records}
            \end{table}

            The records then get grouped together by parameters, and throughput mean, median and stability are computed, as reported on Table~\ref{table:grouped-records}.

            \begin{table}
                \begin{minipage}{\textwidth}
                    \begin{center}
                        \resizebox{1.0\textwidth}{!}{%
                            \input{generated/grouped-records}

                        }
                    \end{center}
                \end{minipage}
                \caption{Records grouped by target platform, lock algorithm, \seqc-only/\sys-optimized version and number of threads.
                Computed values are median, mean, standard deviation and stability of the \texttt{throughput} column of Table~\ref{table:raw-records}.}
                \label{table:grouped-records}
            \end{table}

            As can be seen on the table, these values are computed for different versions of the algorithm.
            \texttt{opt} refers to the \sys-optimized version of the algorithm, while \texttt{seq} refers to the \seqc-only variant.

            Mean, median and standard deviation are computed using the usual definitions, while the stability is computed by dividing the maximum throughput by the minimum throughput, effectively giving an indication on the stability of the data set.
            The closer the stability is to $1.00$, the more stable the sample is for these fixed values of the parameters.
            Figure~\ref{fig:eval-microbench-stability} shows the repartition of the stability among the records of the above table.
            As can be observed on the density chart, most observed values are stable.

            \begin{figure}[t]
                \centering
                \includegraphics[width=\plotwidth\textwidth]{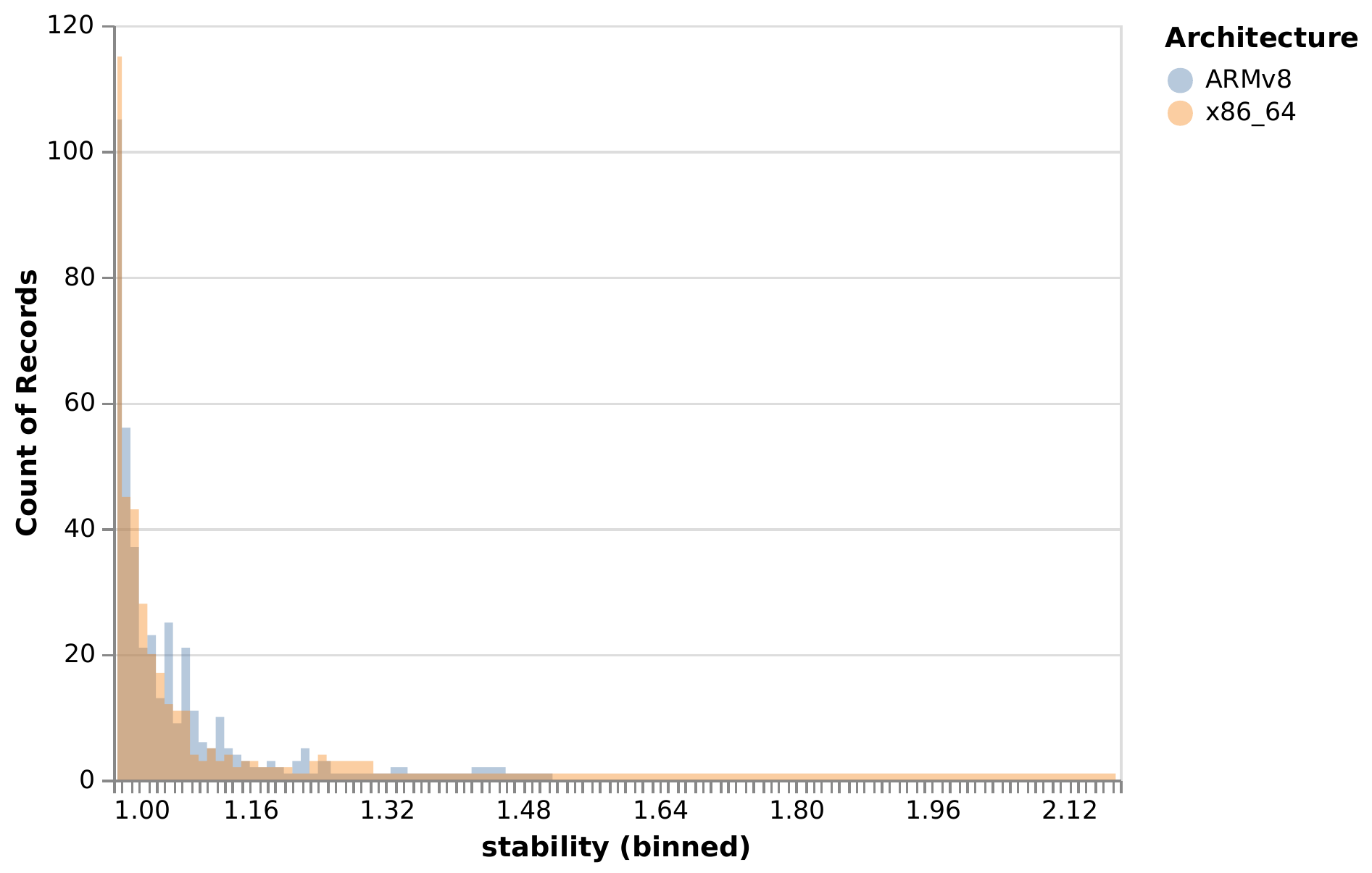}
                \caption{Density of the stability of the different records, per architecture.
                The chart exhibits the fact that most results are very stable ($<1.16$ for stability).}
                \label{fig:eval-microbench-stability}
            \end{figure}

            Another way to see them is to group the lines of Table~\ref{table:grouped-records} by stability values as done in Table~\ref{table:stability-records}.
            We see than more than 84\% of the results have a stability inferior to 10\% (value $\leq 1.1$ in the table), which we believe is satisfactory.

            \begin{table}
                \begin{minipage}{\textwidth}
                    \begin{center}
                        \input{generated/stability-table}

                    \end{center}
                \end{minipage}
                \caption{Number of experiments categorized by stability.
                The mentioned records are lines of Table~\ref{table:grouped-records}.}
                \label{table:stability-records}
            \end{table}

            In practice, we filter out all records with more than 20\% stability to avoid skewing the results.
            Indeed, sometimes the optimization speedup can be in the 0 to 20\% improvement range.
            Therefore, statistically-speaking computing speedups for unstable values would make little to no sense at all.
            For our evaluations, it means we had to throw less than 7\% of the results (see Table~\ref{table:stability-records}).

            Instead of simply dropping unstable records, an alternate method would be to redraw experiment samples until it becomes stable (\eg with a stability threshold value of $1.1$).
            However, for time reasons, we did not use this method.
            We believe the results obtained with this method and the method we actually used would be very similar.

            \paragraph{Analysis of speedups of \sys-optimized over \seqc-only implementations.}

            Then, we use the values in the filtered table of records to compute the speedup $\frac{T_o}{T_s} - 1$, where $T_o$ is median throughput of \sys-optimized and $T_s$ is the median throughput \seqc-only variants, respectively.
            Descriptive statistics aggregates about the observed speedups are showed in Table~\ref{tab:eval-microbench-speedup} and the density of the speedup values is showed in Figure~\ref{fig:eval-microbench-speedups}.
            In the paper, for the sake of space, we only reported maximum observed speedups (the \texttt{max} column in Table~\ref{tab:eval-microbench-speedup}).
            We can observe on Figure~\ref{fig:eval-microbench-speedups} that most speedups are close to 0.
            This effect can mainly be observed on the plot because of the highly-contended cases (number of threads from $8$ and up), where the impact of optimizing barrier is negligible.
            On the other hand, if we observe the same data but split for all measured contention levels (\ie number of threads) as depicted per architecture on Figures~\ref{fig:eval-microbench-speedups-heatmap-arm} and~\ref{fig:eval-microbench-speedups-heatmap-x86}, we can analyze the results with finer-grain details.
            For \texttt{ARMv8 (taishan200-128c)} (Fig.~\ref{fig:eval-microbench-speedups-heatmap-arm}), good results are scattered across the different contentions levels, but speedups tend to be better for low contention level (especially the $1$ thread case).
            In the case of \texttt{x86} (Fig.~\ref{fig:eval-microbench-speedups-heatmap-x86}), the tremendous low-contention speedup case (up to $7\times$ for $1$ thread) is emphasized.
            This is so big that it overshadows the other cases.
            However, the \texttt{qspinlock} column is clearly better than the others, illustrating that in the \texttt{x86} case, \texttt{qspinlock} has \emph{no negative speedup}.

            \begin{table}
            \begin{minipage}{\textwidth}
                \begin{center}
                    \resizebox{1.0\textwidth}{!}{%
                        \input{generated/speedup-table}

                    }
                \end{center}
            \end{minipage}
            \caption{Speedups of \sys-optimized version of the algorithm over \seqc-only variant.
            This descriptive summary must be read with care (especially for the values of \texttt{mean}), as they are only aggregated from \emph{our own} experiment samples (our arbitrary selected thread number, etc.).}
            \label{tab:eval-microbench-speedup}
            \end{table}

            \begin{figure}[t]
                \centering
                \includegraphics[width=\plotwidth\textwidth]{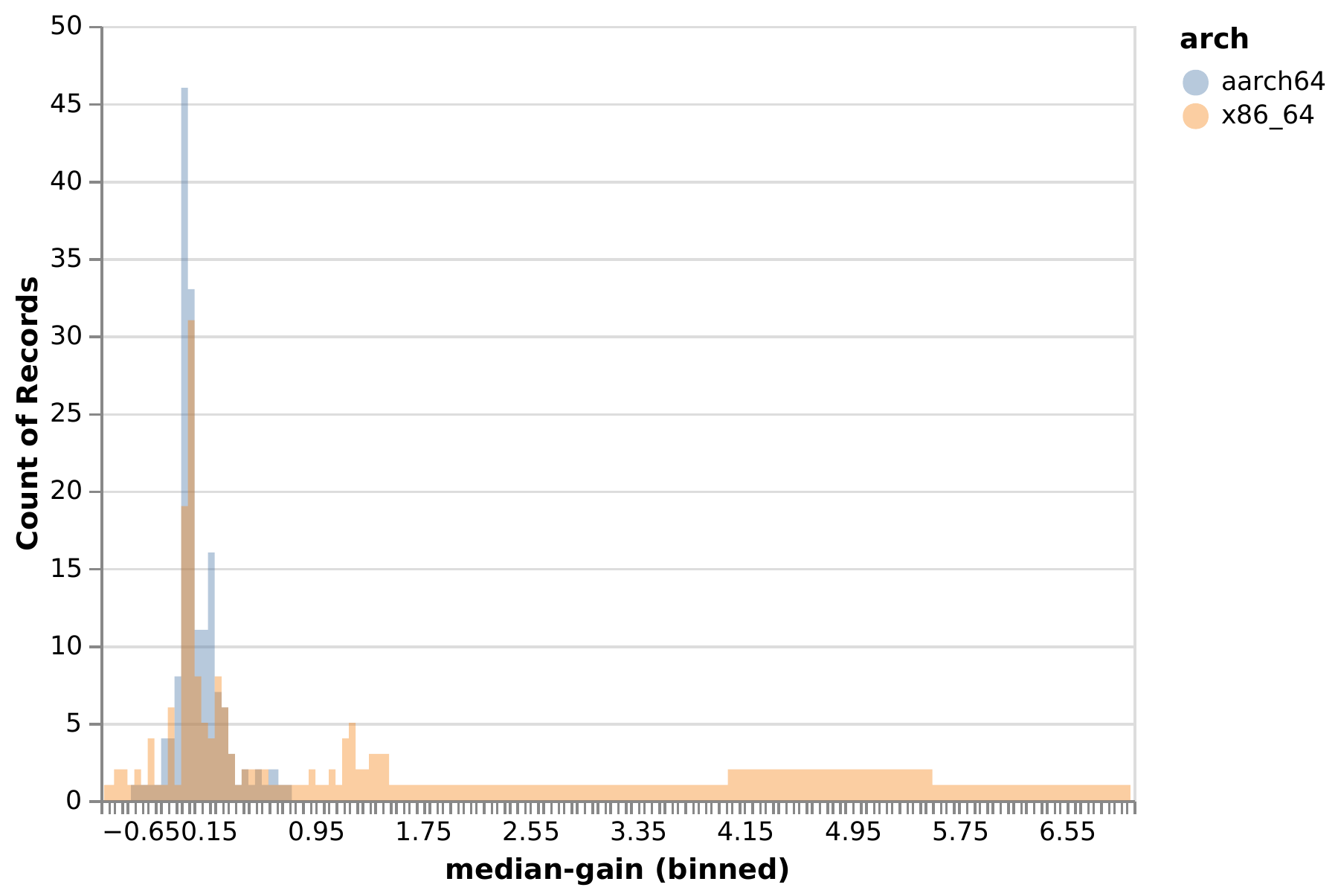}
                \caption{Density of the speedups of the different locks, per architecture.}
                \label{fig:eval-microbench-speedups}
            \end{figure}

            \begin{figure}[t]
                \centering
                \includegraphics[width=\plotwidth\textwidth]{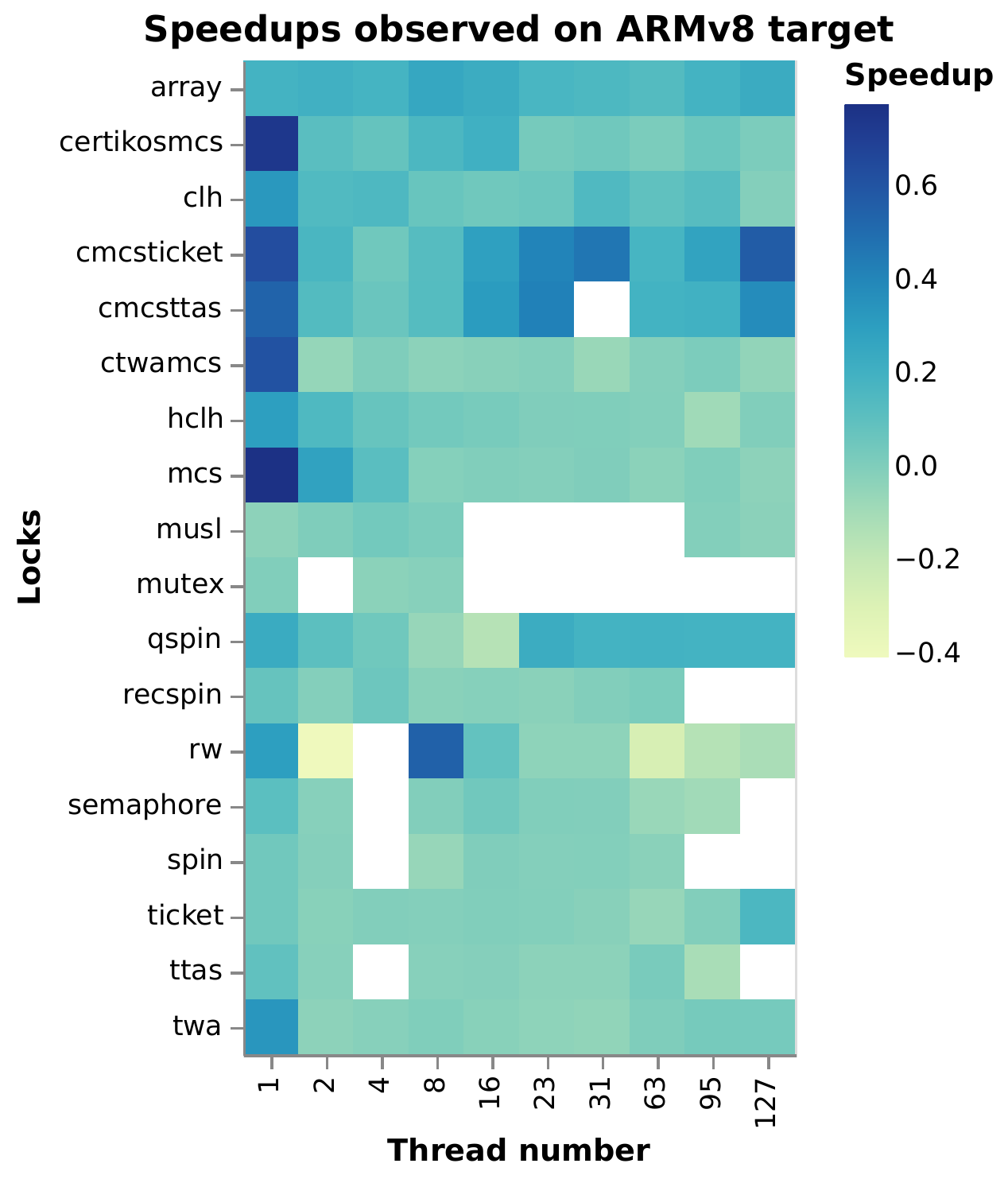}
                \caption{Heat map showing the speedups of the different locks on \texttt{ARMv8 (taishan200-128c)}.\\
                White squares correspond to data filtered-out for instability.}
                \label{fig:eval-microbench-speedups-heatmap-arm}
            \end{figure}

            \begin{figure}[t]
                \centering
                \includegraphics[width=\plotwidth\textwidth]{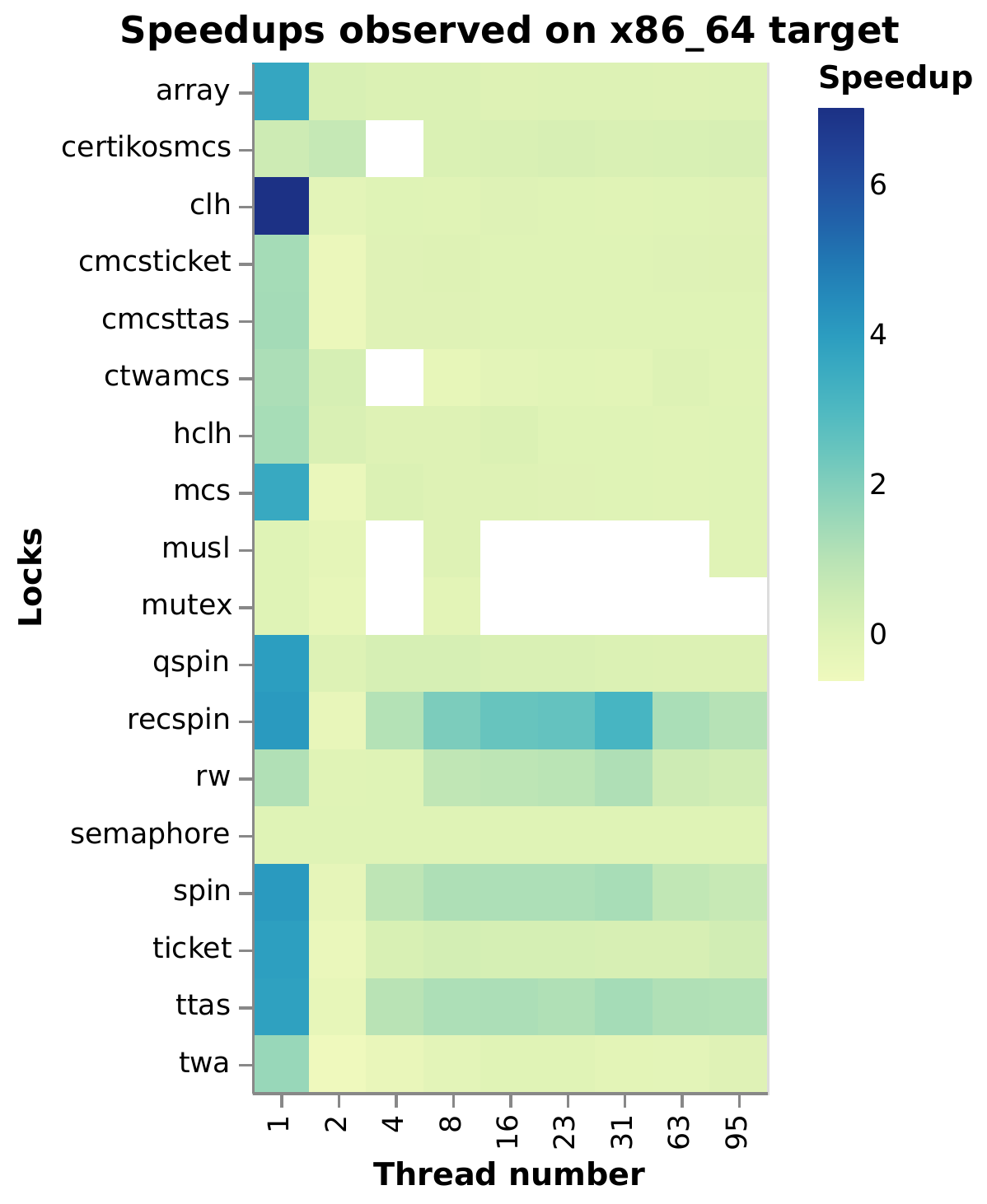}
                \caption{Heat map showing the speedups of the different locks on \texttt{x86\_64 (gigabyte-96c)}.\\
                White squares correspond to data filtered-out for instability.}
                \label{fig:eval-microbench-speedups-heatmap-x86}
            \end{figure}

            \paragraph{MCS lock comparisons.}

            Figure~\ref{fig:eval-microbench-mcs} compares the performance of several MCS lock implementations on \texttt{ARMv8 (taishan200-128c)} and \texttt{x86\_64 (gigabyte-96c)}.
            As reported in the paper, the different MCS lock implementations are: DPDK\cite{DPDK}, Concurrency Kit (ck)\cite{CK}, CertiKOS\cite{gu2016certikos} and our \sys-optimized.

            \begin{figure}[t]
                \centering
                \includegraphics[width=1.0\textwidth]{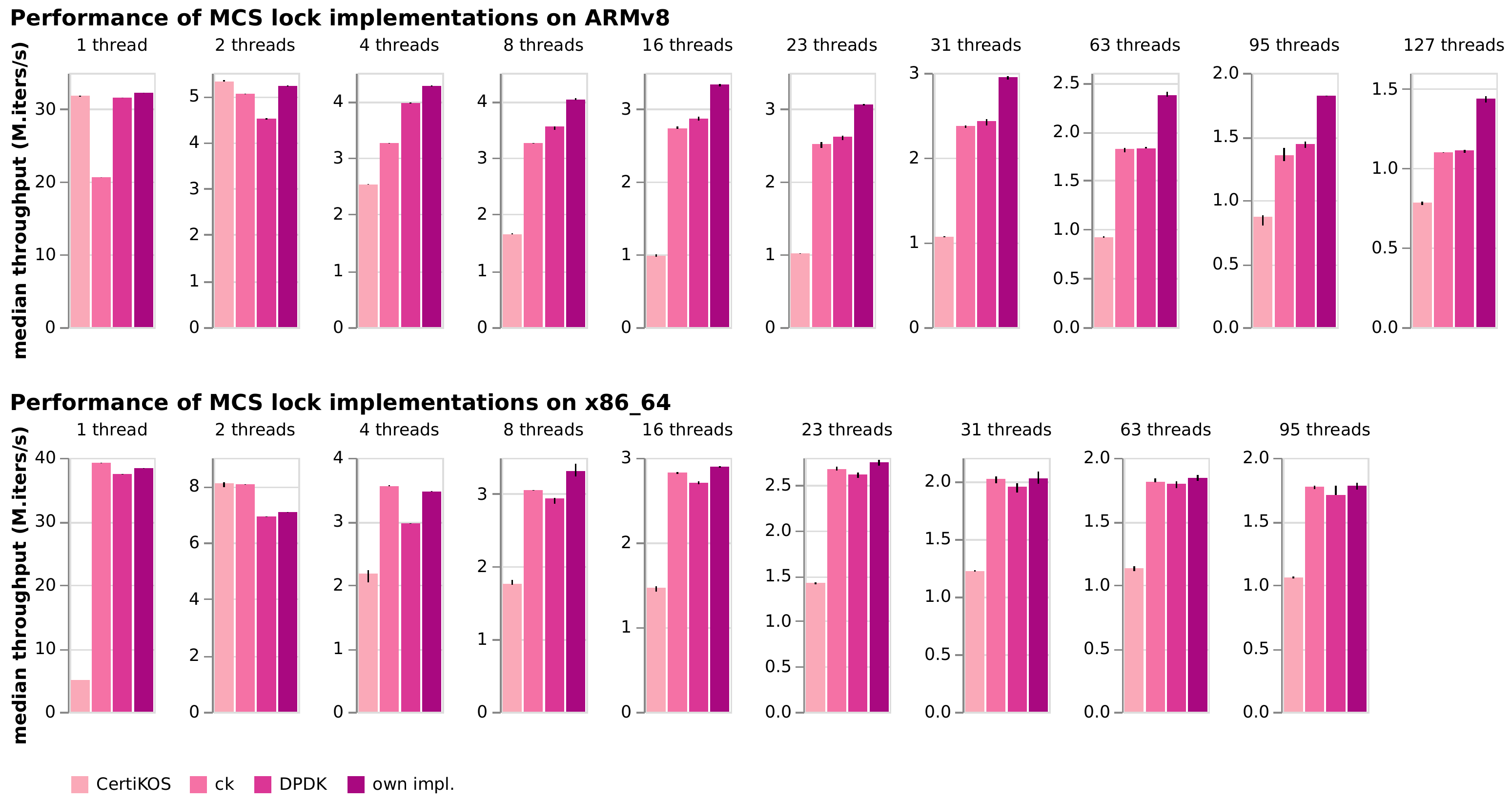}
                \caption{Comparisons of performance of different MCS lock implementations on \texttt{ARMv8 (taishan-128c)} and \texttt{x86\_64 (gigabyte-96c)}.}
                \label{fig:eval-microbench-mcs}
            \end{figure}

            \paragraph{Critical and non-critical section sizes.}

            We observed a few other things by conducting this campaign of experiments.
            Our benchmark setting allows for additional parameters: \texttt{cs\_size} and \texttt{es\_size} (not reported in the charts of this report).
            \begin{itemize}
                \item The \texttt{cs\_size} parameter (for ``critical section size'') allows to artificially increase (and control) the size of the critical section.
            Instead of only touching one cache line by increasing a counter (as depicted in Listing~\ref{listing:eval-benchstructure}), we can touch an arbitrary number of cache lines, which would corresponds to the value of the \texttt{cs\_size} parameter.

                \item The \texttt{es\_size} parameter allows to set an arbitrary number of cache lines touched \emph{outside} the critical section, to simulate different relative sizes of the critical section with regards to the size of the whole program.
            \end{itemize}
            We observed the following:
            \begin{enumerate}
                \item The \texttt{es\_size} parameter did not influence the results, meaning that the lock primitive performances and speedup obtained in \sys-optimized over \seqc-only are not affected by the size of the program that is not in the critical section.

                \item The \texttt{cs\_size} parameter strongly influenced the results in the following way: the bigger the critical section was, the less was the impact of the barrier optimization.
                Additionally, overall, all locking primitives are converging towards the same performance value for an increasing critical section, which is expected as the entry/exit protocols are negligible relatively to a sufficiently large critical section.
            \end{enumerate}
            From this, we can conclude that barrier optimizations and locking protocols make sense especially for small critical sections and fine-grain locking.
            For the final results of the paper, we decided to set \texttt{cs\_size} to $1$ and \texttt{es\_size} to $0$.

            \paragraph{Hash table benchmarks.}

            Linked to these last findings, prior running our custom-made kernel benchmark module, we tried to use the work of Kashyap \etal~\cite{shflock2019} (which is publicly available on Github~\footnote{\url{https://github.com/sslab-gatech/shfllock/tree/master/benchmarks/kernel-syncstress}}), but we were not able to produce predictable results.
            Indeed, the variability of such results was very high, and each time we changed a small parameter it produced different output values.
            It happened even for details that should not influence the results, such as changing the linking order of the object modules in the makefile.
            This was unusable for our work, and could be explained in the following way:
            basically, the critical section in the kernel \emph{syncstress} module of Kashyap~\etal{} is accessing nodes of a hash table.
            The data of this hash table is randomly populated.
            However, accessing a hash table is not a predictable operation in terms of run-time, and different access can yield very different execution times (especially if the hash table is seeded with different random values at each run).
            This would lead to critical section size being very different for different runs (even with the same parameter values), and was therefore not usable to compare different techniques.
            In comparison, our microbenchmark framework, although being simpler in its structure, produce very predictable results (with very small deviations and good stability, as showcased above in this section) and allow to precisely measure the overheads of barriers and the performance of different locking primitive implementations.



%% file: generated/microbench-constants.tex
\newcommand{\evalMicrobenchThreadRange}{1, 2, 4, 8, 16, 23, 31, 63, 95, 127}
\newcommand{\evalMicrobenchDuration}{30}
\newcommand{\evalMicrobenchNbRuns}{5}

%% file: generated/raw-records.tex
\begin{tabular}{llllllllll}
\toprule
{} & architecture & algorithm & seqopt & threads\_nb & run\_nb & atomics &      count & duration &   throughput \\
\midrule
0    &        ARMv8 &     array &    opt &          1 &      1 &     a64 &  957109580 &  30.0122 &  3.18907e+07 \\
1    &        ARMv8 &     array &    opt &          1 &      2 &     a64 &  957161287 &  30.0133 &  3.18913e+07 \\
2    &        ARMv8 &     array &    opt &          1 &      3 &     a64 &  957576858 &   30.025 &  3.18926e+07 \\
3    &        ARMv8 &     array &    opt &          1 &      4 &     a64 &  957238417 &  30.0143 &  3.18927e+07 \\
4    &        ARMv8 &     array &    opt &          1 &      5 &     a64 &  957129609 &  30.0141 &  3.18893e+07 \\
5    &        ARMv8 &     array &    opt &          2 &      1 &     a64 &  209273223 &  30.0116 &  6.97308e+06 \\
6    &        ARMv8 &     array &    opt &          2 &      2 &     a64 &  205836422 &  30.0118 &  6.85851e+06 \\
7    &        ARMv8 &     array &    opt &          2 &      3 &     a64 &  205883982 &  30.0119 &  6.86008e+06 \\
...  &          ... &       ... &    ... &        ... &    ... &     ... &        ... &      ... &          ... \\
3697 &       x86\_64 &      musl &    seq &         63 &      3 &     a64 &   10917032 &  30.0052 &       363838 \\
3698 &       x86\_64 &      musl &    seq &         63 &      4 &     a64 &   12659470 &  30.0059 &       421900 \\
3699 &       x86\_64 &      musl &    seq &         63 &      5 &     a64 &   10882122 &  30.0047 &       362681 \\
3700 &       x86\_64 &      musl &    seq &         95 &      1 &     a64 &   11842053 &  30.0067 &       394647 \\
3701 &       x86\_64 &      musl &    seq &         95 &      2 &     a64 &   11655763 &  30.0056 &       388453 \\
3702 &       x86\_64 &      musl &    seq &         95 &      3 &     a64 &   13233013 &  30.0067 &       441002 \\
3703 &       x86\_64 &      musl &    seq &         95 &      4 &     a64 &   13896114 &  30.0062 &       463108 \\
3704 &       x86\_64 &      musl &    seq &         95 &      5 &     a64 &   13857038 &  30.0065 &       461801 \\
\bottomrule
\end{tabular}

%% file: generated/grouped-records.tex
\begin{tabular}{llllllll}
\toprule
       &     &     &    &         mean &       median &      std & stability \\
arch & algorithm & seqopt & threads\_nb &              &              &          &           \\
\midrule
aarch64 & array & opt & 1 &  3.18913e+07 &  3.18913e+07 &  1436.94 &   1.00011 \\
       &     &     & 2 &  6.87696e+06 &  6.85993e+06 &  54935.1 &   1.02047 \\
       &     &     & 4 &  4.08881e+06 &  4.10817e+06 &  57940.8 &   1.03005 \\
       &     &     & 8 &  3.91338e+06 &  3.90199e+06 &  49164.6 &   1.03246 \\
       &     &     & 16 &  3.75618e+06 &  3.74699e+06 &    15217 &   1.00958 \\
       &     &     & 23 &  2.48333e+06 &  2.41812e+06 &   105903 &   1.09302 \\
       &     &     & 31 &  2.23512e+06 &  2.23639e+06 &   6561.7 &   1.00759 \\
       &     &     & 63 &  1.74074e+06 &  1.74047e+06 &  8351.25 &   1.01191 \\
       &     &     & 95 &  1.32026e+06 &  1.32048e+06 &  5805.31 &   1.01172 \\
       &     &     & 127 &  1.11277e+06 &  1.10708e+06 &  11596.4 &   1.02398 \\
       &     & seq & 1 &  2.67509e+07 &  2.67693e+07 &  30079.4 &    1.0026 \\
       &     &     & 2 &  5.66552e+06 &  5.69127e+06 &  62049.8 &   1.02652 \\
       &     &     & 4 &  3.47092e+06 &  3.46877e+06 &  15011.9 &   1.01014 \\
       &     &     & 8 &  3.10139e+06 &  3.10546e+06 &  7845.61 &     1.006 \\
       &     &     & 16 &  3.05401e+06 &  3.05231e+06 &  4886.13 &   1.00395 \\
       &     &     & 23 &  2.08272e+06 &  2.06173e+06 &  64897.5 &   1.07889 \\
       &     &     & 31 &  1.92977e+06 &  1.92845e+06 &  10816.1 &    1.0132 \\
       &     &     & 63 &  1.52597e+06 &  1.53138e+06 &  17730.6 &   1.02815 \\
       &     &     & 95 &  1.11215e+06 &  1.10892e+06 &  13310.7 &   1.02923 \\
       &     &     & 127 &       899750 &       898612 &  10382.7 &   1.02754 \\
... & ... & ... & ... &          ... &          ... &      ... &       ... \\
x86\_64 & ttas & seq & 63 &  1.16617e+06 &  1.16119e+06 &  14598.2 &   1.03091 \\
       &     &     & 95 &  1.16657e+06 &  1.16654e+06 &   9672.9 &   1.01849 \\
       & twa & opt & 1 &  3.64516e+07 &  3.64537e+07 &  29737.1 &   1.00216 \\
       &     &     & 2 &  6.02238e+06 &  6.02219e+06 &  1489.88 &   1.00069 \\
       &     &     & 4 &  2.10028e+06 &  2.12157e+06 &  68586.1 &   1.09087 \\
       &     &     & 8 &  2.27264e+06 &  2.28098e+06 &  34452.2 &   1.03792 \\
       &     &     & 16 &  2.06737e+06 &  2.06864e+06 &  3631.41 &   1.00465 \\
       &     &     & 23 &  1.85262e+06 &  1.85316e+06 &  1579.55 &   1.00227 \\
       &     &     & 31 &  1.48692e+06 &  1.48739e+06 &  3060.34 &   1.00574 \\
       &     &     & 63 &  1.23218e+06 &  1.23213e+06 &   3422.3 &   1.00665 \\
       &     &     & 95 &  1.11774e+06 &  1.11715e+06 &  5174.11 &   1.01166 \\
       &     & seq & 1 &  1.39532e+07 &  1.39525e+07 &  8965.07 &   1.00157 \\
       &     &     & 2 &  1.51132e+07 &  1.50894e+07 &   119766 &   1.02114 \\
       &     &     & 4 &   3.3227e+06 &  3.31492e+06 &  13600.1 &   1.00948 \\
       &     &     & 8 &  2.69598e+06 &   2.6822e+06 &  19174.8 &    1.0132 \\
       &     &     & 16 &  2.09911e+06 &  2.10757e+06 &  17185.9 &   1.01866 \\
       &     &     & 23 &  1.89835e+06 &  1.89785e+06 &  3783.31 &   1.00542 \\
       &     &     & 31 &  1.71193e+06 &   1.7124e+06 &   5194.5 &    1.0076 \\
       &     &     & 63 &  1.44299e+06 &  1.44049e+06 &  12048.1 &   1.01896 \\
       &     &     & 95 &  1.09068e+06 &  1.08852e+06 &     8834 &   1.01751 \\
\bottomrule
\end{tabular}

%% file: generated/stability-table.tex
\begin{tabular}{rrr}
\toprule
Stability values & Amount (absolute) &          Amount (\%) \\
\midrule
      $\leq 1.1$ &               627 &            $84.62\%$ \\
         $> 1.1$ &             $114$ &            $15.38\%$ \\
         $> 1.2$ &              $45$ &             $6.07\%$ \\
         $> 1.3$ &              $16$ &             $2.16\%$ \\
         $> 1.4$ &               $9$ &             $1.21\%$ \\
  \textbf{Total} &    $\textbf{741}$ &  $\textbf{100.00\%}$ \\
\bottomrule
\end{tabular}

%% file: generated/speedup-table.tex
\begin{tabular}{lrrrrrrrr}
\toprule
                                   Lock & \multicolumn{4}{l}{aarch64} & \multicolumn{4}{l}{x86_64} \\
                                        &       max &      mean &       min &       std &       max &      mean &       min &       std \\
\midrule
       ArrayQ lock\cite{herlihy2011art} &  0.256496 &  0.195695 &  0.136534 &  0.035925 &  3.704002 &  0.512900 &  0.050954 &  1.198285 \\
      CertiKOS MCS\cite{gu2016certikos} &  0.741137 &  0.148102 &  0.014506 &  0.217878 &  0.711755 &  0.323380 &  0.184185 &  0.191676 \\
          CLH lock\cite{herlihy2011art} &  0.326751 &  0.116884 & -0.008331 &  0.090660 &  7.034888 &  0.767551 & -0.150228 &  2.350886 \\
           c-TKT-MCS\cite{dice2012lock} &  0.633937 &  0.317441 &  0.046994 &  0.196817 &  1.379088 &  0.122309 & -0.457079 &  0.497779 \\
          c-TTAS-MCS\cite{dice2012lock} &  0.538990 &  0.265129 &  0.063337 &  0.157456 &  1.388887 &  0.114989 & -0.454959 &  0.502378 \\
                              c-MCS-TWA &  0.610119 &  0.040497 & -0.065989 &  0.201884 &  1.250191 &  0.121885 & -0.282231 &  0.487366 \\
   HCLH lock\cite{herlihy2011art, HCLH} &  0.297295 &  0.050683 & -0.084593 &  0.105824 &  1.331446 &  0.201763 & -0.019287 &  0.430856 \\
 MCS lock\cite{MCSLock, herlihy2011art} &  0.776025 &  0.111308 & -0.031085 &  0.252447 &  3.605063 &  0.387545 & -0.404894 &  1.216811 \\
                  musl mutex\cite{musl} &  0.039510 & -0.000101 & -0.031193 &  0.026376 &  0.048818 & -0.050181 & -0.206353 &  0.110743 \\
 3-state mutex\cite{drepper2005futexes} &  0.000157 & -0.013167 & -0.025476 &  0.012847 & -0.001148 & -0.140026 & -0.287827 &  0.143548 \\
              qspinlock\cite{qspinlock} &  0.235271 &  0.117755 & -0.148553 &  0.131182 &  3.966020 &  0.606631 &  0.085625 &  1.261861 \\
                          rec. CAS lock &  0.078310 &  0.010387 & -0.024441 &  0.037942 &  4.099000 &  1.950000 & -0.338620 &  1.324669 \\
                                RW lock &  0.546915 & -0.007887 & -0.405652 &  0.289225 &  1.194217 &  0.671709 & -0.015319 &  0.462634 \\
                              Semaphore &  0.112009 & -0.001584 & -0.089416 &  0.062090 &  0.004707 &  0.000205 & -0.003467 &  0.002518 \\
                               CAS lock &  0.045456 & -0.007812 & -0.060903 &  0.031734 &  4.099994 &  1.247507 & -0.255443 &  1.175026 \\
            Ticketlock\cite{Ticketlock} &  0.162076 &  0.009903 & -0.060949 &  0.059517 &  3.935485 &  0.633132 & -0.392997 &  1.260925 \\
         TTAS lock\cite{herlihy2011art} &  0.094544 & -0.011786 & -0.110472 &  0.057227 &  3.875872 &  1.322232 & -0.286289 &  1.078801 \\
             TWA lock\cite{dice2019twa} &  0.337384 &  0.026431 & -0.044020 &  0.112353 &  1.612696 &  0.023385 & -0.600899 &  0.627365 \\
\bottomrule
\end{tabular}